\documentclass[11pt]{article}
\usepackage{amssymb}
\usepackage[sectionbib]{natbib}
\usepackage{array,epsfig,rotating}
\usepackage{amsmath}
\usepackage{amsfonts}
\usepackage{multirow}
\usepackage{amsthm}
\usepackage{array}
\usepackage{booktabs}
\usepackage{graphicx}
\usepackage{graphics}
\usepackage{epsfig,latexsym,verbatim}
\usepackage{color}
\usepackage{multicol}
\usepackage{threeparttable}
\usepackage{float}
\usepackage{cancel}
\usepackage{url}
\usepackage{ulem}
\usepackage{rotating}
\usepackage{xr}
\usepackage{secdot}
\usepackage{subcaption}
\usepackage[lined,algonl,boxed]{algorithm2e}
\usepackage{cancel}

\definecolor{red}{rgb}{0,0,0}
\definecolor{green}{rgb}{0,0,0}
\definecolor{blue}{rgb}{0,0,0}

\def\bs{\boldsymbol}
\def\code#1{\texttt{#1}}

\floatstyle{ruled}
\newtheorem{theorem}{Theorem}
\newtheorem{lemma}{Lemma}

\newtheorem{proposition}{Proposition}
\theoremstyle{definition}

\newtheorem{example}{Example}
\newtheorem{remark}{Remark}

\begin{document}

\renewcommand{\baselinestretch}{1.2}
\markboth{\hfill{\footnotesize\rm Xinyi Li, Shan Yu, Yueying Wang, Guannan Wang, Ming-Jun Lai and Li Wang}\hfill}
{\hfill {\footnotesize\rm Spline smoothing of 3D geometric data} \hfill}
\renewcommand{\thefootnote}{}
$\ $\par \fontsize{10.95}{14pt plus.8pt minus .6pt}\selectfont
\vspace{0.8pc} \centerline{\large\bf Nonparametric Regression for 3D Point Cloud Learning}
\vspace{.4cm} \centerline{ 
Xinyi Li$^{a}$, 
Shan Yu$^{b}$, 
Yueying Wang$^{c}$, 
Guannan Wang$^{d}$, 
Li Wang$^{e}$ and  
Ming-Jun Lai$^{f}$
\footnote{\emph{Address for correspondence}: 
Li Wang, Department of Statistics, George Mason University, Fairfax, VA 22030, USA. Email: lwang41@gmu.edu}} \vspace{.4cm} 
\centerline{\it 
$^{a}$Clemson University, 
$^{b}$University of Virginia,
$^{c}$Amazon.com, Inc } 
 \centerline{\it 
 $^{d}$College of William \& Mary,
 $^{e}$George Mason University and 
$^{f}$University of Georgia} \vspace{.55cm}
\fontsize{9}{11.5pt plus.8pt minus .6pt}\selectfont


\begin{quotation}
\noindent {\it Abstract:}
In recent years, there has been an exponentially increased amount of point clouds collected with irregular shapes in various areas. Motivated by the importance of solid modeling for point clouds, we develop a novel and efficient smoothing tool based on multivariate splines over the triangulation to extract the underlying signal and build up a 3D solid model from the point cloud. The proposed method can denoise or deblur the point cloud effectively, provide a multi-resolution reconstruction of the actual signal, and handle sparse and irregularly distributed point clouds to recover the underlying trajectory. In addition, our method provides a natural way of numerosity data reduction. We also introduce a bootstrap method to quantify the uncertainty of the estimators. We establish the theoretical guarantees of the proposed method, including the convergence rate and asymptotic normality of the estimator, and show that the convergence rate achieves optimal nonparametric convergence. Through extensive simulation studies and a real data example, we demonstrate the superiority of the proposed method over traditional smoothing methods in terms of estimation accuracy and efficiency of data reduction. For example, we can reduce an image point cloud of 510,340 voxels to a vector of 4,856 spline coefficients with a peak signal-to-noise ratio 30.76. 

\vspace{9pt}
\end{quotation}

\fontsize{10.95}{14pt plus.8pt minus .6pt}\selectfont
\thispagestyle{empty}

\label{SEC:intro}
\setcounter{equation}{0}
\noindent \textbf{1. Introduction} \vskip 0.1in
\renewcommand{\thefigure}{1.\arabic{figure}} \setcounter{figure}{0}

\noindent Recent advances in computer and information technology have dramatically boosted the availability of three-dimensional (3D) point clouds in many fields, such as geography, environmental science, computer graphics, engineering, economics, and medical imaging. These point clouds are usually collections of enormous measurements in space defined by a given coordinates system. Sometimes, these point clouds are rendered and inspected directly; examples can be found in \cite{Rusinkiewicz:Levoy:00}. More often, the point clouds are converted to polygon mesh or triangle mesh models, e.g., 3D reconstruction is widely used for automated driving vehicles \citep{Lee:Song:Jo:16, Shao:21}. Many efforts have been devoted to investigating the converting technique, and a survey of surface reconstruction from point clouds can be found in \cite{Berger:etal:17}. 

These 3D data typically contain more information than the shape of an object, as they represent some meaningful values in real life. For example, healthcare professionals can use brain activity levels contained in 3D neuroimages to access new angles, resolutions, and details to better understand tissue. Another example is using cloud-based 3D scans to explore the location and access the grade and type of ore in the mining industry. Despite the potential of these applications, there has been limited research on identifying underlying signals from 3D solid objects, as there are many challenges in achieving this goal.

First of all, as shown in the point cloud inputs in Figure \ref{FIG:flowchart_entire}, the shape of the point clouds is usually irregular. To extract the underlying signal over the irregular-shaped 3D solid object, many conventional methods, such as kernel smoothing \citep{Zhu:Fan:Kong:14}, tensor product smoothing \citep{Reiss:Ogden:10}, thin plate spline smoothing \citep{Duchon:77, yue2010nonstationary}, and wavelet smoothing \citep{Morris:Carroll:06}, suffer from the problem of ``leakage'', which is referred as the inaccurate signal reconstruction across the complex domains; see \cite{Ramsay:02, Wang:Ranalli:07, kim2020generalized} for more details. {Take the pillar of the horseshoe domain in Figure \ref{FIG:hs_tri_eg} as an example. Although the Euclidean distance between points A and B is very short, erroneously borrowing information from point B when estimating point A will lead to unsatisfactory results.} Several techniques have been developed in the last two decades to tackle the ``leakage'' problem for two-dimensional (2D) smoothing, such as spline smoothing \citep{Wang:Ranalli:07, wilhelm2016igs, Lai:Wang:13, Awanou:Lai:Wenston:05}, finite element analysis \citep{Arnone:22, lindgren2011explicit}, and kernel smoothing \citep{guo2010non}. For 3D images, \cite{chung2018discrete} applied the discrete heat kernel smoothing, and for irregular-shaped 3D objects,  \cite{huo2020smoothed} considered the finite element method.

\begin{figure}[ht]
\begin{center}
    \includegraphics[height=1.2in]{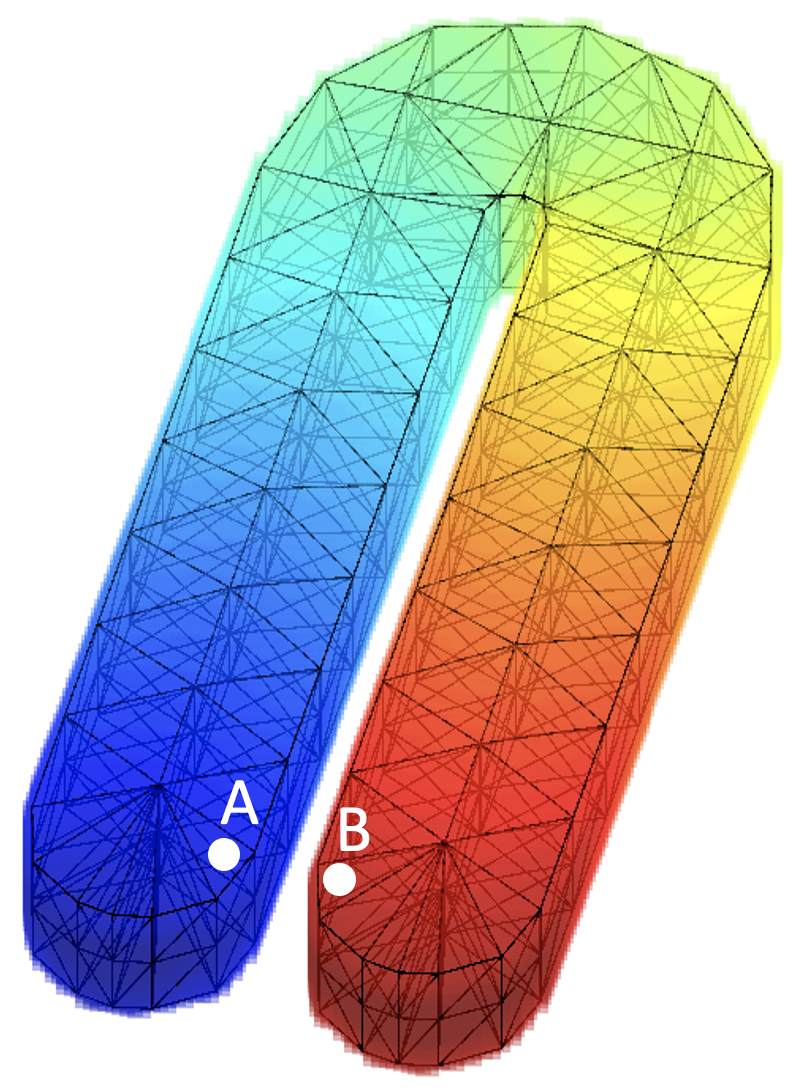} \includegraphics[height=1.15in]{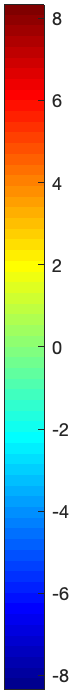} 
    \caption{A illustration of ``leakage'' problem.}
\label{FIG:hs_tri_eg}
\end{center}
\end{figure}

Second, the information comprised in the 3D point clouds usually contains unwanted noise that can obscure the features of the underlying signal of interest. Consequently, denoising is a crucial step in the modeling process. Smoothing methods are widely adopted for denoising in many 2D studies \citep{goldsmith2014smooth, kang2020consistent}. However, the unique characteristics of 3D point clouds present significant challenges in denoising and other forms of analysis \citep{wong2016fiber}. 

Third, point cloud data acquisition generally relies on 3D sensors such as LiDAR and depth cameras. Raw point cloud data obtained from these devices is often sparse, uneven, or even partially missing, which poses significant challenges for advanced vision tasks such as point cloud classification, segmentation, and statistical modeling and inference. Therefore, it is necessary to preprocess the raw data to generate complete, dense, and uniform point cloud data.

Last but not least, as mentioned before, the size of the point clouds is usually enormous, and thus both the storage and the analysis method are very challenging. For example, a 3D image produced from a standard positron emission tomography (PET) brain scan contains over half a million voxels. This large amount of data can make it challenging to perform computationally intensive tasks such as data processing and  analysis.

To address the prevalent challenges in converting a 3D irregular-shaped point cloud to a solid model and accurately estimating the underlying function, we propose a novel and efficient smoothing method based on trivariate penalized spline over triangulations (TPST) with a directional derivative-based penalty function. Figure \ref{FIG:flowchart_entire} demonstrates the flowchart of the overall procedure of the TPST method for analyzing a 3D point cloud. The proposed TPST smoothing method has several appealing properties for handling data collected from a point cloud.

{One of the main advantages is its ability to effectively solve the problem of ``leakage'' across complex domains. By borrowing the information from neighboring tetrahedra, the TPST method is able to accurately denoise or deblur the collected data while preserving inherent geometric features or spatial structures. This makes it particularly useful for analyzing irregularly shaped 3D point clouds. Similar to P-splines of \cite{eilers1996flexible}, TPST uses compactly supported basis functions and a sparse penalty, which is computationally efficient for smoothing unevenly distributed data, such as sparse and irregular 3D point cloud data. By using an appropriate partition, the method can reconstruct or recover the underlying function from both global and local missing data, and provide a multi-resolution reconstruction of the actual signal.}

{Additionally, the proposed TPST has the potential to be used as a numerosity reduction method \citep{han2011data}. Assuming the underlying signal is smooth, it can significantly reduce the amount of data that needs to be stored by representing it through a smaller set of spline coefficients. Once the spline basis functions are determined, we can use the coefficients to recover the underlying signal of the point clouds. For example, as demonstrated in Section 6, a point cloud of size 50,000 can be represented by a vector of 323 spline coefficients. Let's take the biomedical image for another example, where the proposed method can reduce an image point cloud of 510,340 voxels to a vector of 4,856 coefficients, with a peak signal-to-noise ratio reaching 30.76. See Table \ref{TAB:app} in a later section for detail. This can greatly alleviate the storage and computational challenges associated with large point clouds.}

{Furthermore, the proposed TPST is computationally efficient to deal with point clouds of large size. As shown in Section 3, it is very easy to set up the spline basis functions and penalties of TPST. Similar to many other types of spline smoothing \citep{Ma:12, Wang:Xue:Yang:20, Wang:Yang:09, Liu:Zhao:21, Wang:Wang:Lai:Gao:20}, TPST is a global estimation method with an explicit model expression, which makes it computationally much more efficient compared to traditional methods. This is because it only requires solving a single linear system, which reduces the computational complexity and makes it suitable for handling large point clouds.}

\begin{figure}[ht]
\begin{center}
    \includegraphics[scale=0.4]{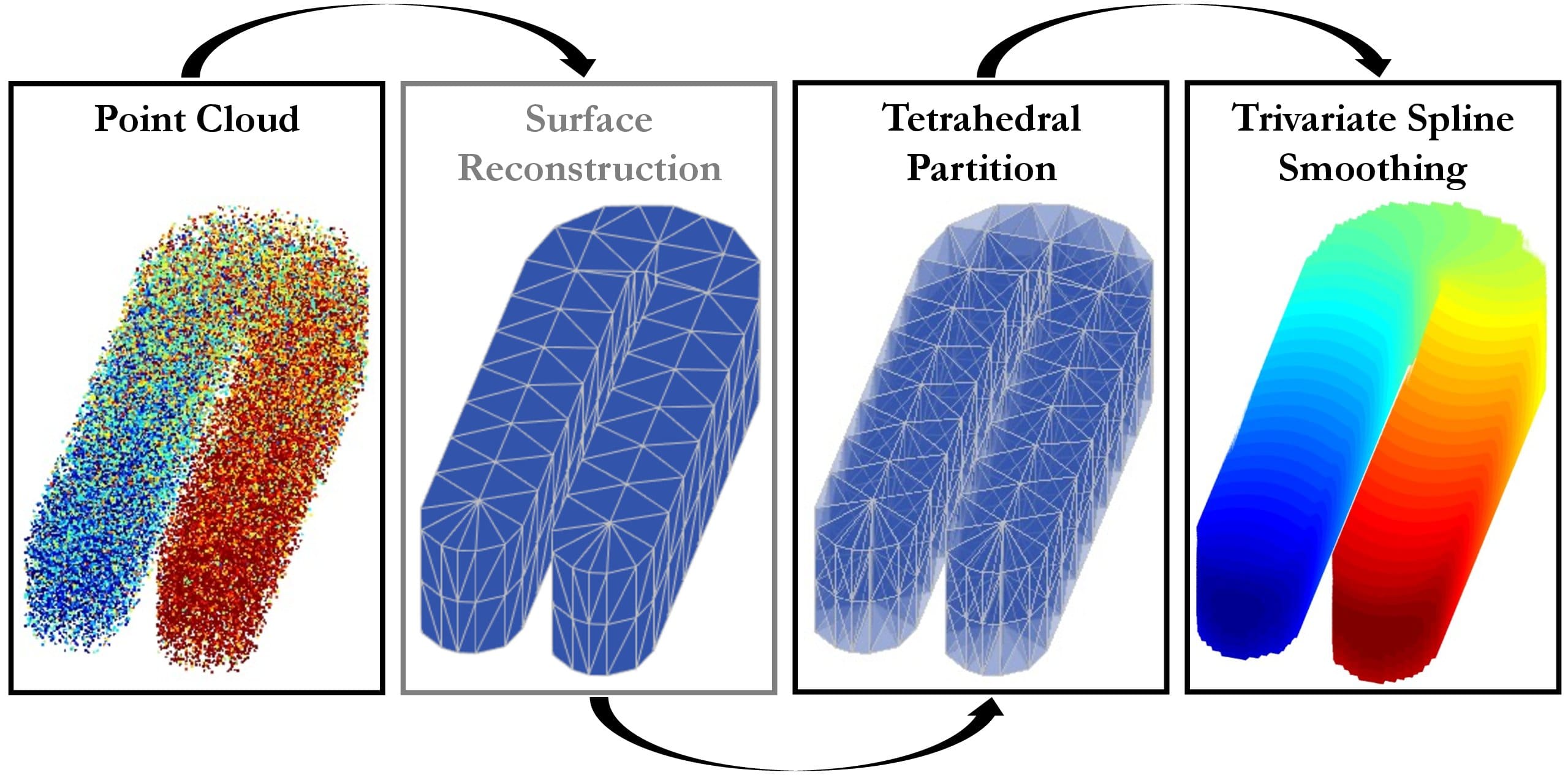}
    \caption{A flowchart demonstrating the overall procedure of TPST method analyzing a 3D point cloud.}
\label{FIG:flowchart_entire}
\end{center}
\end{figure} 

In this article, we also investigate the theoretical properties of the proposed smoothing method. To be more specific, we establish the convergence rate of the proposed TPST estimator, which is governed by the fineness of the triangulation, the trivariate spline degree and smoothness, penalty parameter, and smoothness of the unknown function. We further illustrate that the rate achieves the optimal nonparametric convergence rate when the penalty parameter is zero. Additionally, we derive the asymptotic normality for the TPST estimator.

{Quantifying the uncertainty of estimators is essential when working with point clouds, as it allows for the assessment of the reliability of estimates, determination of the robustness of the estimators, and making accurate conclusions based on the results. To address this, one may use the asymptotic normality of the proposed TPST estimator to establish pointwise confidence intervals in theory. In practice, however, it can be difficult to derive the exact form of the standard error due to the characteristics of the trivariate spline basis functions. As an alternative, we propose using a wild bootstrap method to estimate the standard error. This method is particularly useful for 3D point clouds that may have varying levels of noise or uncertainty as it can handle the heterogeneity effectively in data. Our simulation results show that the bootstrap standard error is very close to the true standard errors, providing a reliable method for quantifying the uncertainty of the estimators.}

The rest of the paper is structured as follows. In Section 2, we give an overview of the triangulations and the methods for constructing the triangulations from a 3D point cloud. Once a triangulation is obtained, we can construct the spline space over this triangulation. We introduce the penalized spline estimators in Section 3. In Section 4, we present the theoretical results. In Section 5, we discuss implementation details for the proposed TPST method, including the creation of the penalty matrix, selection of the penalty parameter, choice of triangulation, and evaluations of uncertainties. We then demonstrate the performance of the proposed smoothers on several simulation examples and a real data example in Section 6. Finally, this paper is concluded, and future work is outlined in Section 7. We put all the theoretical details and more detailed introductions in the supplementary material.

\vskip 0.2in \noindent \textbf{2. Trivariate splines over triangulations and basic properties} \vskip 0.1in
\renewcommand{\thetable}{2.\arabic{table}} \setcounter{table}{0} 
\renewcommand{\thefigure}{2.\arabic{figure}} \setcounter{figure}{0}
\renewcommand{\theequation}{2.\arabic{equation}} \setcounter{equation}{0} 
\label{SEC:preliminaries}

In this section, we provide a basic framework for triangulations and trivariate splines on those triangulations. Please refer to Appendix A for more detailed introductions.

\vskip 0.1in \noindent \textbf{2.1. Triangulations} \vskip .10in

We assume there exists a tight polyhedral domain, $\Omega$, that contains all the data locations of a point cloud. In the theoretical development below, we consider the $\Omega$ is fixed and given. Additionally, we take into account that the domain may have one or multiple holes that do not contain any observations, as seen in the 3D domain of Figure \ref{FIG:tetrademos} (b). Triangulations are very popular in approximating the domain $\Omega$. In the following, we use $T$ to denote a tetrahedron, that is, a convex hull of four noncoplanar points in $\mathbb{R}^3$. Then, a collection $\triangle = \{T_1, \ldots, T_N\}$ of $N$ tetrahedra is called a \textsf{triangulation} of $\Omega = \cup_{h=1}^N T_h$, provided that any pair of tetrahedra in $\triangle$ intersect at most at a common vertex, along a common edge, or along a common triangular face. An example of triangulation $\triangle = \{T_1 = \langle\bs{v}_5, \bs{v}_1, \bs{v}_4, \bs{v}_3\rangle, T_2=\langle \bs{v}_2, \bs{v}_1, \bs{v}_3, \bs{v}_4\rangle\}$ of the domain $\Omega = \langle\bs{v}_2, \bs{v}_5, \bs{v}_3, \bs{v}_4\rangle$ is illustrated in Figure \ref{FIG:TriPartition} (a), where $\bs{v}_{\imath},~\imath = 1, \ldots, 5,$ are the vertices of the tetrahedra. In contrast, as illustrated in Figure \ref{FIG:TriPartition} (b) for tetrahedra $T_1$, $T_3 = \langle\bs{v}_2, \bs{v}_1, \bs{v}_3, \bs{v}_6\rangle$ and $T_4 = \langle\bs{v}_2, \bs{v}_1, \bs{v}_6, \bs{v}_4\rangle$, $\{T_1, T_3, T_4\}$ does not form a triangulation of $\Omega$, because of the intersections between the pairs of tetrahedra $(T_1, T_3)$ and $(T_1, T_4)$.

\begin{figure}[ht]
	\begin{center}
		\begin{tabular}{ccccc}
			\includegraphics[width=.78in,height=.78in]{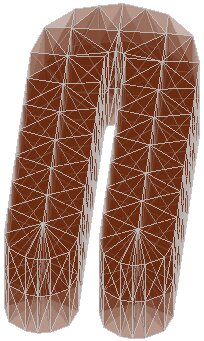} &
			\includegraphics[height=.8in]{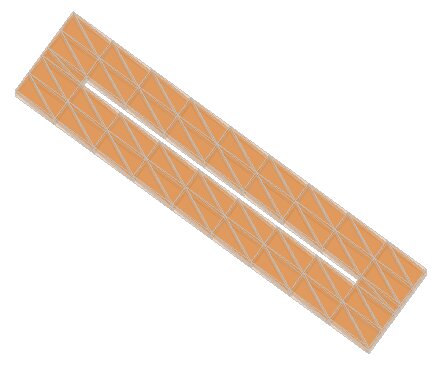}  &
			\includegraphics[height=.82in]{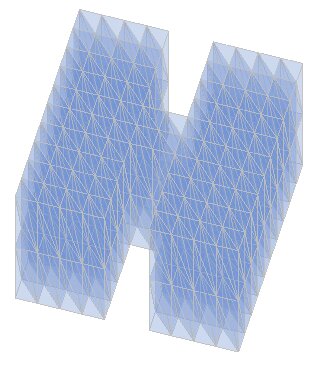} & \includegraphics[width=1in,height=.78in]{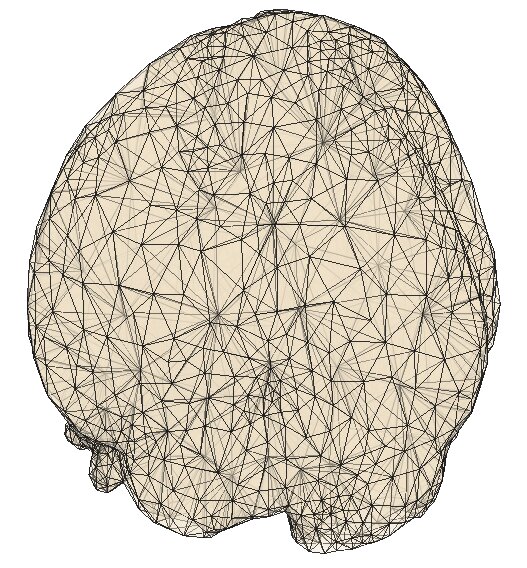} &
			\includegraphics[scale=0.12]{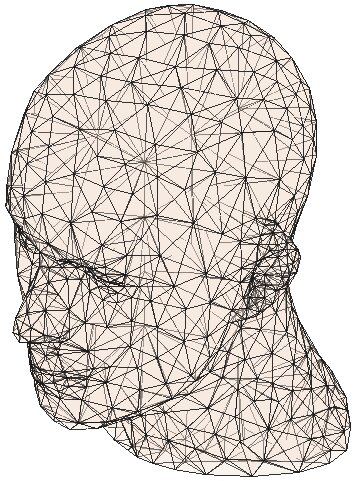}\\
			(a) & (b) & (c) & (d) & (e) \\
		\end{tabular} 
	\end{center}
	\caption{Examples of 3D point clouds and corresponding triangulation.}
	\label{FIG:tetrademos}
\end{figure}

\begin{figure}
\begin{center}
\begin{tabular}{ccccc}
\includegraphics[width=.23\linewidth]{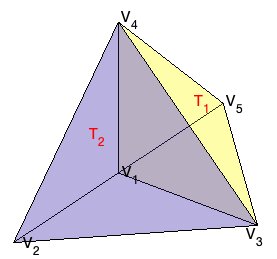} &\includegraphics[width=.23\linewidth]{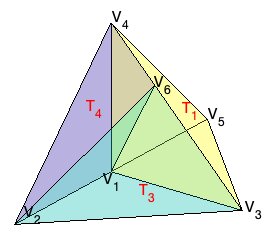}\\
(a) A triangulation &(b) A non-triangulation\\
\end{tabular}
\caption{(a) and (b) provide an example and a counterexample of triangulation.}
\label{FIG:TriPartition}
\end{center}
\end{figure}

Furthermore, let $|T|$ be the length of the longest edge of $T$, and $\varrho_{T}$ be the radius of the largest ball that can be inscribed in $T$, then the ratio $\beta_T:=|T|/\varrho_{T}$ is the \textsf{shape parameter} of $T$; see an illustration example in Figure \ref{FIG:TriShape} (a). In general, the shape parameter $\beta_T$ describes the shape of $T$, {and the larger $\beta_T$ is, the flatter the tetrahedron $T$ becomes}. If $T$ is a regular tetrahedron whose six edges are all of the same lengths, then $\beta_T = 2\sqrt{6}$, and for any other tetrahedron, $\beta_T > 2\sqrt{6}$.  Based on this shape parameter, we can define that a triangulation $\triangle $ is $\beta$\textsf{-quasi-uniform} if there is a positive value $\beta$ such that $\triangle$ satisfies
\begin{equation}
|\triangle|/\varrho_{T} \le \beta < \infty, \quad \mbox{ for all } T\in \triangle, 
\label{DEF:beta}
\end{equation}
where $|\triangle| := \max\{|T|,T \in \triangle\}$ is referred as the \textsf{size of $\triangle$}, i.e., the length of the longest edge of $\triangle$. Let $N$ be the number of the tetrahedra in the polygonal domain $\Omega$. From (\ref{DEF:beta}), we have for a $\beta$-quasi-uniform partition, $N\leq (4\pi|\triangle|^3)^{-1} 3V_{\Omega}\beta^3$, where $V_{\Omega}$ denotes the volume of $\Omega$.

To illustrate the definition of quasi-uniform partition, we show an example of non-$\beta$-quasi-uniform triangulation in Figure \ref{FIG:TriShape} (b),
\begin{align*}
    \triangle=\{
        &T_1=\langle \bs{v}_2, \bs{v}_1, \bs{c}_1, \bs{v}_4\rangle, 
        T_2=\langle \bs{v}_1, \bs{c}_2, \bs{c}_1, \bs{v}_4\rangle, 
        T_3=\langle \bs{c}_1, \bs{v}_1, \bs{c}_3, \bs{c}_2\rangle, \\
        &T_4=\langle \bs{c}_3, \bs{v}_1, \bs{c}_4, \bs{c}_2\rangle, 
        T_5=\langle \bs{c}_3, \bs{v}_1, \bs{v}_3, \bs{c}_4\rangle,\ldots\},    
\end{align*}
where the edges $\langle\bs{v}_4, \bs{c}_1\rangle$ and $\langle\bs{c}_2, \bs{c}_3\rangle$ are perpendicular to the edge $\langle\bs{v}_2, \bs{v}_3\rangle$, and $\langle\bs{c}_1, \bs{c}_2\rangle$ and $\langle\bs{c}_3, \bs{c}_4\rangle$ are perpendicular to $\langle\bs{v}_3, \bs{v}_4\rangle$. 
\begin{figure}
\begin{center}
\begin{tabular}{cc}
	\includegraphics[scale=.45]{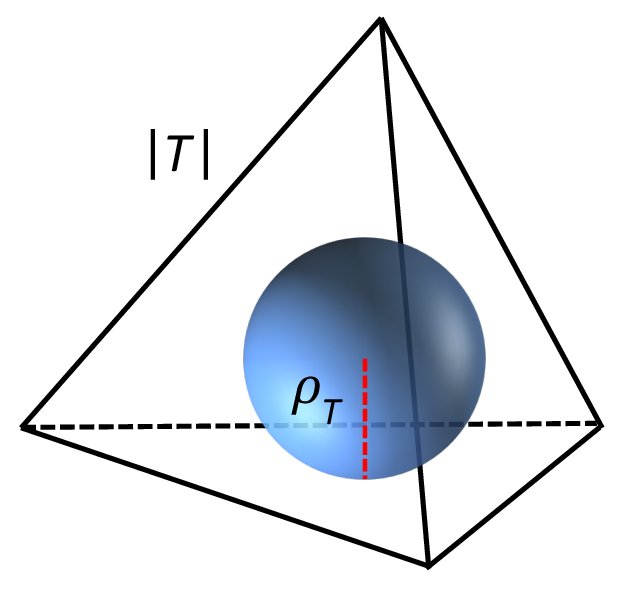} &\includegraphics[scale=.11]{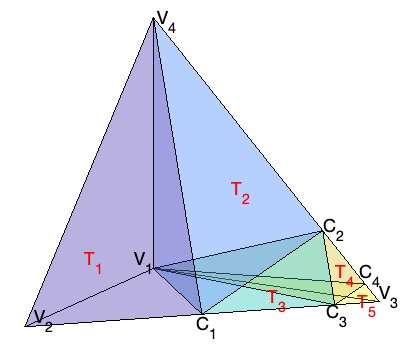}\\
	(a) Shape parameters &(b) Non-$\beta$-quasi-uniform triangulation\\
\end{tabular}
\caption{(a) shows an illustration of the shape parameters of a tetrahedron; (b) shows an example of a non-$\beta$-quasi-uniform triangulation.}
\label{FIG:TriShape}
\end{center}
\end{figure}

\vskip 0.1in \noindent \textbf{2.2. Point clouds to triangulations} \vskip .10in

We now describe how to construct a triangulation from a 3D point cloud. The first step is to determine the domain of the point cloud, which is often given as a polygonal domain $\Omega$, a geometric region defined by a set of polygons. This can be used to represent a wide range of objects, such as buildings, terrain, and objects in medical imaging. If the domain of the point cloud is not known, we can use surface reconstruction techniques, such as triangular meshes, to find a tight polygonal domain that encases the point cloud. As $\Omega$ has piecewise linear boundary faces, we have a normal direction perpendicular to each planar face of $\Omega$, which can be used to determine whether a point is inside or outside of $\Omega$. Then, we use this normal direction information of all boundary faces and the observed data locations to partition $\Omega$ into a collection of tetrahedra. Many mesh generation algorithms for 3D domains available in various software packages and toolboxes, such as the MATLAB functions \code{delaunay} and \code{distmesh/distmeshnd} \citep{Persson:Strang:04}, the C++ library \code{CGAL} \citep{Jamin:2015,cgal:eb-20b}, the \code{TetGen} \citep{Si:2015}, and \code{iso2mesh} \citep{Fang:2009} built on \code{CGAL}. In this article, we consider a triangulation, $\triangle$, with vertices containing the partial or whole set of observed data locations to achieve interpolation at the given data locations. This type of triangulation can be constructed using the constrained Delaunay triangulation method, as described in \cite{Shewchuk:98}. Based on a similar method proposed in \cite{Xu:19}, we obtain all the triangulations using MATLAB. Figure \ref{FIG:flowchart} illustrates our flowchart for obtaining a triangulation. This generation procedure can be applied to many other 3D point clouds collected in various fields, as shown in the triangulation examples in Figure \ref{FIG:tetrademos}.

\begin{figure}[ht]
\begin{center}
    \includegraphics[scale=0.45]{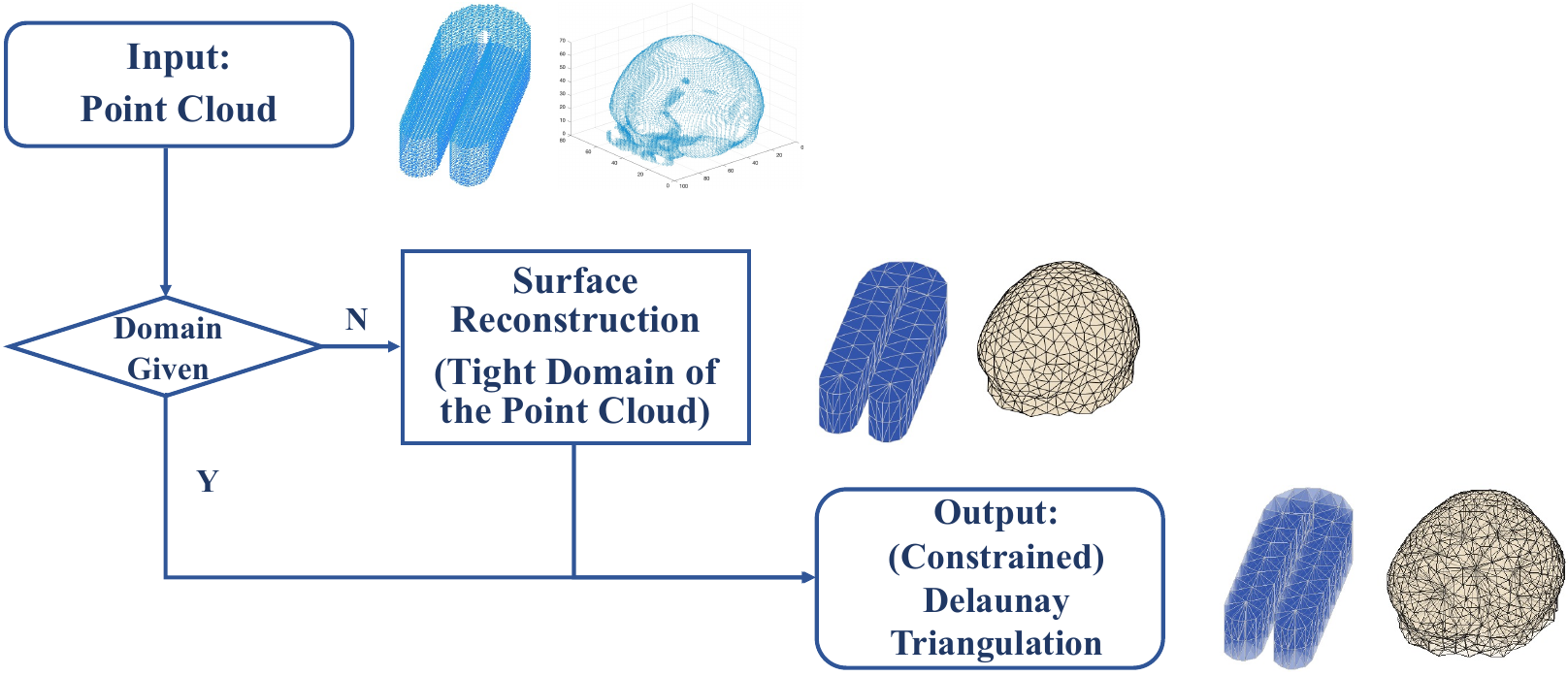}
    \caption{{Flowchart demonstrating the process of constructing a triangulation from a 3D point cloud.}}
\label{FIG:flowchart}
\end{center}
\end{figure}

\vskip 0.1in \noindent \textbf{2.3. Trivariate splines on a triangulation} \vskip .10in
\label{SUBSEC:spline}

Suppose we have obtained a triangulation $\triangle$ for a tight domain $\Omega$ of a given point cloud. For any tetrahedron $T = \langle\bs{v}_1, \bs{v}_2, \bs{v}_3, \bs{v}_4\rangle \in \triangle$, any point $\bs{p} = (x,y,z) \in \mathbb{R}^3$ has a unique representation in terms of $\langle\bs{v}_1, \bs{v}_2, \bs{v}_3, \bs{v}_4\rangle$,
\[
\bs{p} = b_1\bs{v}_1 + b_2\bs{v}_2 + b_3\bs{v}_3 + b_4\bs{v}_4,~\text{with}~b_1+ b_2 + b_3 + b_4 = 1,
\]
where $(b_1, b_2, b_3, b_4)$ are called the \textsf{barycentric coordinates} of $\bs{p}$ \textsf{relative to the tetrahedron} $T$, and they are nonnegative if $\bs{p}$ is inside or on the faces of $T$. Accordingly, for some nonnegative integers $i$, $j$, $k$, $l$ with $i + j + k + l = d$, define \textsf{trivariate Bernstein basis polynomial of degree} $d$ \textsf{relative to} $T$ as
\begin{equation}
\label{Bform}
B_{ijkl}^{d,T}(\bs{p}) := \frac{d!}{i!j!k!l!} b_1^i b_2^j b_3^k b_4^l,~\text{with}~i + j + k + l = d.
\end{equation}
For any positive integer $d$, let $\mathcal{P}_d$ be the space of all trivariate polynomials with total degrees less than or equal to $d$. Note that the dimension of $\mathcal{P}_d$ is $\binom{d+3}{3}$. According to Theorem 15.8 in \cite{Lai:Schumaker:07} and Lemma \ref{LEM:basis} in the Appendix A, the set of Bernstein basis polynomials in (\ref{Bform}) forms a basis for the space of polynomials $\mathcal{P}_d$. Thus, any polynomial $\phi(\bs{p}) \in \mathcal{P}_d$ on with domain $T$ can be written uniquely as the \textsf{B-form}, 
\begin{equation}
\label{EQN:Bform-1}
\phi(\bs{p}) = \sum_{i + j + k + l = d} \gamma_{T; ijkl}B_{ijkl}^{d, T}(\bs{p}) = \mathbf{B}^d_T(\bs{p})^{\top}\bs{\gamma}_T,
\end{equation}
where the coefficients $\bs{\gamma}_T = \{\gamma_{T; ijkl}\}_{i + j + k + l = d}$ are called \textsf{B-coefficients} of $\phi(\cdot)$.

For any nonnegative integer $r$, we use $\mathcal{C}^r(\Omega)$ to denote the collection of all $r$-th continuously differentiable functions over $\Omega$. For triangulation $\triangle = \{T_1, \ldots, T_N\}$, let $\mathcal{S}_d^r(\triangle) = \{s \in \mathcal{C}^r(\Omega):s|_T \in \mathcal{P}_d(T), T \in \triangle\}$ be a spline space of degree $d$ and smoothness $r$ over $\triangle$, where $s|_T$ is the polynomial piece of spline $s$ restricted on tetrahedron $T$. According to (\ref{EQN:Bform-1}), for any $s\in\mathcal{S}_d^r(\triangle)$, there exists a coefficient vector $\bs{\gamma} = (\bs{\gamma}_{T_1}^{\top}, \ldots, \bs{\gamma}_{T_N}^{\top})^{\top}$ with $\bs{\gamma}_{T_J}={\{\gamma_{T_J;ijkl}\}_{i + j + k + l = d}}$ such that 
\begin{equation}
\label{EQN:s_Bform}
s|_{T_J}(\bs{p}) = \sum_{i + j + k + l = d} \gamma_{T_J; ijkl}B_{ijkl}^{d, T_J}(\bs{p}) = \mathbf{B}^d_{T_J}(\bs{p})^{\top}\bs{\gamma}_{T_J}, ~J = 1, \ldots, N.
\end{equation}
That is, $s(\bs{p}) = \mathbf{B}_d(\bs{p})^{\top}\bs{\gamma}$, 
where $\mathbf{B}_d = \{(\mathbf{B}^d_{T_1})^{\top}, \ldots, (\mathbf{B}^d_{T_N})^{\top}\}^{\top}$.

For noise-free data, the approximation order of trivariate spline spaces was studied in \cite{Lai:Schumaker:07} when $d\geq 8r + 1$ and \cite{Lai:89} when $d\geq 6r+3$. Particularly, \cite{Lai:89} proved that when $d\geq 6r+3$, the space $\mathcal{S}_d^r(\triangle)$ can {attain the optimal converge rate}; see Lemma \ref{LEM:appord} below. For any index $\bs{\alpha} = (\alpha_1, \alpha_2, \alpha_3)$ of order $|\bs{\alpha}| = \alpha_1 +\alpha_2 + \alpha_3$, we denote the derivatives $D^{\bs{\alpha}} s = \partial^{|\bs{\alpha}|} s / \partial_x^{\alpha_1} \partial_y^{\alpha_2} \partial_z^{\alpha_3}$. For any function $f$ over the closure of domain $\Omega$, denote the $L^{q}(\Omega)$ norm ($1 \leq q < \infty$) and supremum norm as $\Vert f \Vert_{L^q(\Omega)} = \{\int_{\Omega}|f(\bs{v})|^q d\bs{v}\}^{1/q}$ and $\Vert f \Vert_{\infty, \Omega} = \sup_{\bs{v}\in \Omega} |f(\bs{v})|$, respectively. Next, for any $k \geq 0$, denote $|f|_{k, \infty, \Omega} = \max_{|\bs{\alpha}|=k}\Vert D^{\bs{\alpha}} f \Vert_{\infty, \Omega}$, and $|f|_{k,q,\Omega}=\sum_{|\bs{\alpha}|\leq k}\|D^{\bs{\alpha}} f \|_{L^{q}(\Omega)}$. Further, we define the Sobolev space $\mathcal{W}^{\ell,q }(\Omega )=\left\{f:|f|_{k,q, \Omega}<\infty, 0\le k\le \ell \right\}$ for $1\leq q\leq \infty$, $\ell\geq1$.
\begin{lemma}
\label{LEM:appord}
\citep[Theorem 3.5.2 in][]{Lai:89}
For all $f\in \mathcal{W}^{d+1,q}(\Omega)$ with $1\leq q \leq \infty$, $r\geq0$, and $d\geq 6r+3$, there exists a spline $s_f\in\mathcal{S}_d^r(\triangle)$ such that 
\begin{equation*}
\Vert D^{\bs{\alpha}}\left(f-s_{f}\right) \Vert_{L^q(\Omega)} \leq K|\triangle|^{d +1-|\bs{\alpha}|}|f|_{d +1, q,\Omega},  
\label{eqGS}
\end{equation*}
for all $0\leq|\bs{\alpha}|\leq m$, where $K>0$ is a constant independent of $f$ and $|\triangle|$ but is dependent on the geometry of $\triangle$.
\end{lemma}

\begin{remark} \label{remark1}
Both $d$ and $r$ are parameters for the spline space $\mathcal{S}_d^r(\triangle)$, which determines the smoothness of the TPST estimator and are usually predetermined by the user. Practically, as long as $d \geq r$, we can construct the spline bases. If $d \geq 6r+3$, we can achieve the full approximation (approximation with optimal convergence rate) order, as shown in Lemma \ref{LEM:appord}. A higher value of $d$ indicates a higher degree polynomial, which can result in a higher computation burden; see our analysis in Section 6. 
In practice, the choice of $d$ and $r$ is closely tied to the intended interpretation of the estimated function. If the goal is to enhance the signal-to-noise ratio for visualization or to suggest a simple parametric model, then a slightly oversmoothed function with a subjectively chosen parameter may be appropriate. However, if the focus is on accurately estimating the regression function and preserving local structures, then a slightly undersmoothed function may be more suitable.
\end{remark}

\vskip 0.2in \noindent \textbf{3. Penalized spline estimators} \vskip 0.1in
\renewcommand{\thetable}{3.\arabic{table}} \setcounter{table}{0} 
\renewcommand{\thefigure}{3.\arabic{figure}} \setcounter{figure}{0}
\renewcommand{\theequation}{3.\arabic{equation}} \setcounter{equation}{0} 
\label{SEC:TPS}

With all the preparations in the previous introduction, we apply the trivariate spline over the triangulation to recognize the underlying signal from 3D point clouds in this section. 

{In the following, for any $i = 1, \ldots, n$, let point $\bs{p}_i = (x_i, y_i, z_i)\in\mathbb{R}^3$ be the location or design point of the $i$th observation in a point cloud of sample size $n$. Let $W_i$ be the response variable observed on the $i$th location $\bs{p}_i$.} Then, we regard any point cloud as a set of $n$ observations $\{(\bs{p}_i, W_i)\}_{i=1}^n$ in general. To extract the underlying signal, we consider the following nonparametric regression model
\begin{equation}
W_i = m(\bs{p}_i) + {\sigma(\bs{p}_i)} \varepsilon_i,
\label{DEF:model}
\end{equation}
where $m(\cdot)$ is some smooth but unknown 3D function, {$\sigma(\cdot)$ is the unknown conditional standard deviation function}, and $\varepsilon_i$ is the random error term with mean zero. {Assume $(W_i,\varepsilon_i)$ are general iid copies of $(W,\varepsilon)$.}

In nonparametric smoothing, the roughness penalty approach is widely used when smoothing noisy data \citep{Green:Silverman:94, Wood:03, Lai:Wang:13}. Including a roughness penalty and choosing a proper tuning parameter can avoid overfitting problems and balance the bias and variance of the estimator of the function. To estimate the underlying function $m$ in (\ref{DEF:model}), we formulate the roughness penalty approach as the following penalized least squares problem:
\begin{equation}
    \min_{s\in\mathcal{S}_d^r(\triangle)} \sum_{i=1}^n\{W_i-s(\bs{p}_i)\}^2 + \rho_n\mathcal{E}(s),
\label{EQ:mhatlambda}
\end{equation}
where the roughness penalty
\begin{equation}
\mathcal{E}(s) = \sum_{|\bs{\alpha}| = 2} {2 \choose \alpha_1} {2-\alpha_1 \choose \alpha_2} \int_{\Omega} \{D^{\bs{\alpha}} s(\bs{p})\}^2 \mathrm{d}\bs{p},
\label{penalty}
\end{equation}
and $\rho_n$ is the roughness penalty parameter. We aim to find the minimizer of (\ref{EQ:mhatlambda}) in $\mathcal{S}_d^r(\triangle)$, denoted as $\widehat{m}_{\rho_n}$, which is the \textsf{Trivariate Penalized Spline over triangulation} (TPST) estimator of $m$. The tuning parameter $\rho_n$ controls the smoothness of the fitted spline function. A larger $\rho_n$ leads to a less fluctuating function. If $\rho_n$ goes to infinity, our estimator shrinks to linear functions where the roughness penalty $\mathcal{E}(s) = 0$. On the other hand, when $\rho_n = 0$, the estimator becomes the standard unpenalized least squares spline estimator. A proper penalty parameter $\rho_n$ balances the goodness of fit for the data and the volatility of estimated functions. 

{The penalty $\mathcal{E}(s)$ in (\ref{penalty}) is a commonly used penalty; see \cite{Green:Silverman:94} for the 2D case.} For a spline function $s \in \mathcal{S}_d^r(\triangle)$, combining (\ref{EQN:s_Bform}), the roughness penalty in (\ref{penalty}) can be written as follows: 
\begin{align}
\label{EQN:EsP}
\mathcal{E}(s) = \sum_{T \in \triangle} \mathcal{E}(s|_T)
& = \sum_{T \in \triangle} \sum_{|\bs{\alpha}| = 2} {2 \choose \alpha_1}{2 - \alpha_1 \choose \alpha_2} \int_T \{D^{\bs{\alpha}}s|_T(\bs{p})\}^2 \mathrm{d}\bs{p}
=\sum_{T \in \triangle}\bs{\gamma}^{\top}_T \mathbf{P}_T \bs{\gamma}_T,
\end{align}
where $\mathbf{P}_T$ is the corresponding penalty matrix. Let $\mathbf{B}_d = (\mathbf{B}_d(\bs{p}_1), \ldots, \mathbf{B}_d(\bs{p}_n))^{\top}$, then we have $\mathcal{E}(\mathbf{B}_d^{\top} \bs{\gamma}) = \bs{\gamma}^{\top} \mathbf{P} \bs{\gamma}$, where $\mathbf{P}$ is the block diagonal penalty matrix. See Section A.6 for more details in derivations and calculations of penalty matrices.

Since $s\in {\cal C}^r(\triangle)$, the coefficients of $s$ satisfy some smoothness conditions across each interior faces of $\triangle$. One can obtain a smoothness constraint matrix $\mathbf{H}$ such that $\mathbf{H}\bs{\gamma} = \bs{0}$ by repeatedly applying equation (\ref{EQ:H_r}) in the Appendix A over all shared triangular faces, and more details are available in Section A.2 in the supplementary material. A MATLAB implementation is discussed in \cite{Awanou:Lai:Wenston:05}. Thus, the objective function can be written as
\begin{equation}
\min_{\bs{\gamma}} \sum_{i=1}^n \left\{W_i-\mathbf{B}_d(\bs{p}_i)^{\top} \bs{\gamma}\right\}^{2} + \rho_n\bs{\gamma}^{\top}\mathbf{P} \bs{\gamma},~\textrm{subject to~} \mathbf{H}\bs{\gamma}=\mathbf{0},
\label{EQU:objective_2}
\end{equation}
where $\mathbf{H}$ is the matrix for all smoothness conditions across shared edges or faces of tetrahedra, which depends on the smoothness parameter $r$ and the structure of the triangulation. One can use QR decomposition to get rid of the constraint in (\ref{EQU:objective_2}). Specifically, $\mathbf{H}^{\top} = \mathbf{Q}\mathbf{R}=(\mathbf{Q}_1~\mathbf{Q}_2) \binom{\mathbf{R}_1}{\mathbf{R}_2}$, where $\mathbf{Q}$ is an orthogonal matrix, $\mathbf{R}$ is an upper triangular matrix, $\mathbf{R}_1$ is a full rank matrix with the same rank as $\mathbf{H}$, and $\mathbf{R}_2$ is a matrix of zeros. Note that for any vector $\bs{\gamma}$ satisfying $\mathbf{H}\bs{\gamma} = \mathbf{0}$, there exists some $\bs{\theta}$ such that $\bs{\gamma} = \mathbf{Q}_2\bs{\theta}$. Also, for any $\bs{\theta}$, $\mathbf{H}(\mathbf{Q}_2\bs{\theta}) = \mathbf{0}$ holds. Thus, (\ref{EQU:objective_2}) is equivalent to a penalized regression problem without constraint:
\begin{equation}
\sum_{i=1}^n\left\{W_i-\mathbf{B}_d(\bs{p}_i)^{\top}
\mathbf{Q}_2\bs{\theta}\right\}^{2}+\rho_n\bs{\theta}^{\top}\mathbf{Q}_{2}^{\top}\mathbf{P}\mathbf{Q}_{2}\bs{\theta},
\label{EQ:PLS}
\end{equation}
which leads to a closed form of the solution. To be specific, let $\mathbf{W}=(W_1,\ldots, W_n)^{\top}$, then the minimizer of (\ref{EQ:PLS}) is given by $\widehat{\bs{\theta}}_{\rho_n}=(\mathbf{Q}^{\top}_2\mathbf{B}_d^{\top}\mathbf{B}_d\mathbf{Q}_2  + \rho_n\mathbf{Q}_2^{\top} \mathbf{P} \mathbf{Q}_2)^{-1} \mathbf{Q}^{\top}_2\mathbf{B}_d^{\top} \mathbf{W}$. Consequently, the spline coefficient in (\ref{EQU:objective_2}) can be estimated by $\widehat{\bs{\gamma}}_{\rho_n}=\mathbf{Q}_{2}\widehat{\bs{\theta}}_{\rho_n}$, which yields the TPST estimator $\widehat{m}_{\rho_n}(\bs{p})=\mathbf{B}_d(\bs{p})^{\top}\widehat{\bs{\gamma}}_{\rho_n}$. Several methods can be used for choosing the penalty parameter $\rho_n$, such as the block cross-validation (block CV) and generalized cross-validation (GCV). Detailed discussion is given in Section 5.2.

\vskip 0.2in \noindent \textbf{4. Theoretical results} \vskip 0.1in
\renewcommand{\thetable}{4.\arabic{table}} \setcounter{table}{0} 
\renewcommand{\thefigure}{4.\arabic{figure}} \setcounter{figure}{0}
\renewcommand{\theequation}{4.\arabic{equation}} \setcounter{equation}{0} 
\label{SEC:results}

In the previous section, we discuss how to construct TPST and capture the underlying signal using TPST from the point clouds. In the following, we investigate the theoretical support of TPST. 

For random variables $X_n$, $X_n=O_P(b_n)$ if $\lim\nolimits_{c\rightarrow \infty }\limsup_n P(|X_n|\geq cb_n)=0$, $n\geq 1$. Similarly, for any constant $c>0$, $X_n=o_P(b_n)$ if $\lim_{n\to\infty}P(|X_n|\geq cb_n)=0$. And $a_n \asymp b_n$ if there exist two positive constants $c_1, c_2$ such that $c_1 |a_n| \le |b_n| \le c_2 |a_n|$, for all $n\ge 1$.

\vskip 0.1in \noindent \textbf{4.1. Convergence rate} \vskip .10in
\label{SSEC:convergence}

Before we state our main results, we make the following assumptions, which are standard in nonparametric literature \citep{Lai:Wang:13, Yu:etal:20, Huang:03b}.
\begin{itemize}
	\item[(A1)] The trivariate function $m\in \mathcal{W}^{\ell+1,\infty}(\Omega)$ for an integer $\ell \ge 1$.

	\item[(A2)] The noise $\epsilon$ satisfies that $\lim_{\eta\rightarrow\infty}\mathrm{E}[\epsilon^{2}I(\epsilon>\eta)]=0$ and $\mathrm{E}|\epsilon_n^{2+\eta}|\leq\upsilon_{\eta}$ for some $\eta>0$. The standard deviation function $\sigma(\bs{p})$ is continuous on $\Omega$ and 
	$0<c_{\sigma}\leq \inf_{\bs{p}\in \Omega}\sigma(\bs{p}) \leq \sup_{\bs{p}\in \Omega}\sigma(\bs{p}) \leq C_{\sigma}<\infty$.

	\item[(A3)] The density function of the observations is bounded below and above.

	\item[(A4)] The number of the tetrahedra $N$ and the sample size $n$ satisfy that $N=Cn^{\gamma}$ for some constant $C>0$ and $\gamma<1$.
\end{itemize}

To obtain the asymptotic analysis of spline estimators, we account on an important property that the data-driven norm uniformly approach to its expectation uniformly over the entire spline space. To see this, for any function $f$ over the closure of domain $\Omega$, let $E_n(f) = n^{-1}\sum_{i=1}^n f( \bs{p}_i) $ and $E(f) = E[f(\bs{p})]$. Define the empirical inner product and norm as $\left\langle f_1, f_2\right\rangle _{n,\Omega} = E_n(f_1f_2)$ and $\left\|f_1\right\|_{n,\Omega}^2 =\left\langle f_1,f_1\right\rangle _{n,\Omega}$, respectively, for measurable functions $f_1$ and $f_2$ on $\Omega$. The theoretical $L^2$ inner product and the induced norm are given by $\left\langle f_1,f_2\right\rangle_{L^2(\Omega)} = E(f_1f_2)$ and $\left\| f_1\right\|_{L^2(\Omega)}^2 = \left\langle f_1,f_1\right\rangle_{L^2(\Omega)}$. We illustrate the uniform convergence rate of empirical inner product to the theoretical one in the following Lemma \ref{LEM:Rnorder}.

\begin{lemma}
\label{LEM:Rnorder} 
{Denote the basis for $\mathcal{S}_d^r(\triangle)$ constructed in \cite{Lai:Schumaker:07} by  $\{B_{\xi}\}_{\xi \in \mathcal{M}}$}, where $\mathcal{M}$ stands for the index set of spline bases. Let $g_1 = \sum_{\xi \in \mathcal{M}}c_{\xi}B_{\xi}$, $g_2 = \sum_{\zeta \in \mathcal{M}}\tilde{c}_{\zeta}B_{\zeta}$ be any spline functions in $\mathcal{S}_d^r(\triangle)$. Under Assumptions (A3) and (A4), we have
\[
R_n = \sup\limits_{g_1, g_2 \in \mathcal{S}_d^r(\triangle)} \left|\frac{\left\langle g_1, g_2\right\rangle_{n,\Omega} - \left\langle g_1, g_2\right\rangle _{\Omega}}{\left\|g_1\right\|_{\Omega}\left\|g_2\right\|_{\Omega}}\right| = O_P\left\{(N\log n)^{1/2} n^{-1/2}\right\}.
\]
\end{lemma}

For the purpose of illustrating theoretical development, we rewrite the penalty in (\ref{penalty}) in terms of linear operation. Let $\mathcal{B}:= \mathcal{B}(\Omega)$ be the space of all bounded real-valued functions over $\displaystyle \Omega = \cup_{T \in \triangle}T$ equipped with the inner product $n\langle f_1, f_2 \rangle_{n, \Omega} + \rho_n \langle f_1, f_2 \rangle_{\mathcal{E}}$, where
\[
\langle f_1, f_2 \rangle_{\mathcal{E}} = \sum_{|\bs{\alpha}| = 2}{2\choose \alpha_1}{2 - \alpha_1\choose \alpha_2}\sum_{T \in \triangle} \int_T\{D^{\bs{\alpha}} f_1(\bs{p})\} \{D^{\bs{\alpha}} f_2(\bs{p})\} \mathrm{d}\bs{p}
\]
for $f_1, f_2\in\mathcal{B}$. 

Next, we introduce a measure of the complexity of the spline space $\mathcal{S}_d^r(\triangle)$, and another measure which bounds the size of the derivatives:
\begin{equation}
V_n=\sup_{g\in \mathcal{S}_d^r(\triangle)}\left\{ \frac{\left\| g\right\|
_{\infty,\Omega}}{\left\| g\right\| _{n,\Omega}},\left\|
g\right\| _{n,\Omega
}\neq 0\right\},\quad \overline{V}_n=\sup_{g\in \mathcal{S}_d^r(\triangle)}\left\{ \frac{%
\left\| g\right\| _{\mathcal{E}}}{\left\| g\right\|
_{n,\Omega}},\left\| g\right\| _{n,\Omega}\neq 0\right\}.
\label{DEF:Vn}
\end{equation}
These two measures will play an important role in developing the asymptotic results. 
We use the following Lemma \ref{LEM:Vnorder} to demonstrate the upper bounds of $V_n$ and $\overline{V}_n$.

\begin{lemma}
\label{LEM:Vnorder} 
Under Assumptions (A3) and (A4), we have $V_n =O_P(|\triangle|^{-3/2}), ~\overline{V}_n =O_P(|\triangle|^{-2})$.
\end{lemma}

{Let $\epsilon_i=\sigma(\bs{p}_i)\varepsilon_i$, and denote the generic variable of $\epsilon_i$ as $\epsilon$.}
We then define a linear operator $P_{\rho_n}:\mathcal{B}\mapsto\mathcal{S}_d^r(\triangle)$ such that $P_{\rho_n}W = \widehat{m}_{\rho_n}$. Note that in general, $P_{\rho_n}$ is not a linear projection. Thus, we have $P_{\rho_n}W = P_{\rho_n}m + P_{\rho_n}\epsilon$, where $P_{\rho_n}m$ and $P_{\rho_n}\epsilon$ are the penalized spline estimators based on $\{m(\bs{p}_i)\}_{i=1}^n$ and $\{\epsilon\}_{i=1}^n$, respectively. Under some conditions \citep{Huang:03b}, $P_0$ is a bounded operator on $\mathcal{S}_d^r(\triangle)$ in the maximum norm. Denote $s_{\rho_n, m} = P_{\rho_n}m$ and $s_{\rho_n, \epsilon} = P_{\rho_n}\epsilon$. Consequently, for the penalized spline estimator $\widehat{m}_{\rho_n}$ in (\ref{EQ:mhatlambda}), 
\begin{eqnarray}
\label{EQ:decomposition}
\widehat{m}_{\rho_n}(\bs{p}) - m(\bs{v}) = \left\{s_{\rho_n,m}(\bs{p}) - m(\bs{p})\right\} + s_{\rho_n, \epsilon}(\bs{p}),
\end{eqnarray}
where $s_{\rho_n,m}(\bs{p}) - m(\bs{p})$ and $s_{\rho_n, \epsilon}(\bs{p})$ are referred to as the bias and noise terms, respectively.

According to the error decomposition in (\ref{EQ:decomposition}), to derive the convergence rate of $\widehat{m}_{\rho_n}$ to $m$, it is sufficient to evaluate the size of the bias and noise terms. The following propositions provide the upper bound of the bias size and noise size.

\begin{proposition}
\label{PROP:bias} 
Under Assumptions (A1), (A3) and (A4), if $d\ge 6r+3$ and $\triangle$ is a $\beta$-quasi-uniform triangulation, then we have
\[
\left\| s_{\rho_n, m} - m\right\| _{\infty, \Omega} = O_P\left\{\frac{\rho_n}{n\left|\triangle \right|^{7/2}}|m|_{2, \infty, \Omega} + \left(1 + \frac{\rho_n}{n\left|\triangle\right|^{11/2}}\right) \left|\triangle\right|^{\ell+1}\left|m\right|_{\ell +1, \infty,\Omega}\right\}.
\]
\end{proposition}

\begin{proposition}
\label{PROP:variance} Under Assumptions (A2) and (A4), $\Vert s_{\rho_n,\epsilon}\Vert _{L^2(\Omega)} = O_P\left(\frac{1}{\sqrt{n}|\triangle|^{3/2}}\right)$.
\end{proposition}

\begin{proposition}
\label{PROP:maxnorm} Under Assumptions (A2) and (A4), $\|s_{\rho_n, \epsilon}\|_{\infty, \Omega} = O_P\left\{\frac{(\log{n})^{1/2}}{\sqrt{n}|\triangle|^{3/2}} + \frac{\rho_n}{n^{3/2}|\triangle|^7}\right\}$.
\end{proposition}

Based on Propositions \ref{PROP:bias}--\ref{PROP:maxnorm}, we illustrate the convergence rates of the TPST estimator in the following Theorem \ref{THM:L2-convergence}, in terms of both the $L^2(\Omega)$ and supremum norms.

\begin{theorem}
\label{THM:L2-convergence} 
Under Assumptions (A1)--(A4), if $d \ge 6r+3$ and $\triangle$ is a $\beta$-quasi-uniform triangulation, we have
\begin{align*}
	\Vert\widehat{m}_{\rho_n}-m\Vert_{L^2(\Omega)}
	=&O_P\left\{\frac{\rho_n}{n|\triangle|^{7/2}}|m|_{2,\infty,\Omega} + \left(1+\frac{\rho_n}{n|\triangle|^{11/2}}\right) |\triangle|^{\ell+1}|m|_{\ell+1,\infty,\Omega} + \frac{1}{\sqrt{n}|\triangle|^{3/2}}\right\},\\
	\Vert\widehat{m}_{\rho_n}-m\Vert_{\infty,\Omega}
	=&O_P\left\{\frac{\rho_n}{n|\triangle|^{7/2}}|m|_{2,\infty,\Omega} + \left(1+\frac{\rho_n}{n|\triangle|^{11/2}}\right)|\triangle|^{\ell+1}|m|_{\ell+1,\infty,\Omega}\right.\\
	&+ \left.\frac{(\log{n})^{1/2}}{\sqrt{n}|\triangle|^{3/2}} + \frac{\rho_n}{n|\triangle|^4}\right\}.
\end{align*}
\end{theorem}

\begin{remark}
For the unpenalized spline estimator, i.e., $\rho_n=0$, if one takes $N\asymp n^{1/(\ell+2)}$, then $\Vert\widehat{m}_0-m\Vert_{L^2(\Omega)}^2 = O_P(n^{-(\ell+1)/(\ell+2)})$, which achieves the optimal convergence rate as shown in \cite{Stone:82}. Similarly, for the supremum norm, $\Vert\widehat{m}_0-m\Vert_{\infty,\Omega}^2 = O_P\{(n^{-1}\log{n})^{(\ell+1)/(\ell+2)}\}$ when $N\asymp (n/\log n)^{1/(\ell+2)}$, which is also the optimal rate of convergence. When $\rho_n>0$, one can obtain optimal convergence rates for both $L^2$ and supremum norms with $\rho_n=o(n^{\ell/2(\ell+2)})$ and same orders of $N$.
\end{remark}

\begin{remark}
Assumption (A3) is standard in nonparametric literature. Even though this assumption may not be satisfied with missing data, we usually can find a proper choice of the triangulation by adjusting the size of some of the tetrahedra to obtain a decent penalized spline fitting.
\end{remark}

\vskip .10in \noindent \textbf{4.2. Asymptotic normality} \vskip .10in
\label{SSEC:normality}

To derive the asymptotic normality of the TPST estimator, we further assume the following conditions.
\begin{itemize}
\item[(A4')] The number of the tetrahedra $N$ and the sample size $n$ satisfy that $N=Cn^{\gamma}$ for some constant $C>0$ and $1/(\ell+2)<\gamma<1$.

\item[(A5)] The penalized parameter $\rho_n$ satisfies $\rho_n=o(n^{1/2}N^{-1})$.
\end{itemize}

\begin{remark}
A sufficient condition for a negligible bias term is provided by Assumptions (A4') and (A5). Compared with Assumption (A4) in Section 4.1, (A4') further assumes that the number of tetrahedra needs to be greater than a lower bound which depends on the degree of the function. A similar assumption for the univariate case has been discussed in \cite{Li:Ruppert:08}. Meanwhile, assumption (A5) requires smaller $\rho_n$, which reduces the bias through under smoothing. 
\end{remark}

\begin{theorem}
\label{THM:normality} 
Under Assumptions (A1)--(A3), (A4') and (A5), as $n\rightarrow\infty$, for each $\bs{p}\in \Omega$,
\[
\frac{\widehat{m}_{\rho_n}(\bs{p}) - m(\bs{p})}{\sqrt{\mathrm{Var}\{\widehat{m}_{\rho_n}(\bs{p})|\mathbb{P}\}}}\stackrel{D}{\longrightarrow} N(0,1),
\]
where $\mathbb{P}$ is the collection of all $\bs{p}_i$, $i=1,\ldots,n$. 
\end{theorem}

\begin{remark}
The above asymptotic distribution result can be used to construct asymptotic confidence intervals in theory. For example, if we estimate $m(\bs{p})$ using piecewise constant splines, Lemma \ref{LEM:var-cons} in the Appendix B gives the size of the pointwise variance $\mathrm{Var}\{\widehat{m}_{\rho_n}(\bs{p})\} = \sigma^2(\bs{p})\{nf(\bs{p})V_T\}^{-1}\{1 + o(1)\}$, $\bs{p}\in \Omega$, where $V_T$ is the volume of tetrahedron $T$. Therefore, an asymptotic $100(1-\alpha)\%$ pointwise confidence envelop for $m(\bs{p})$ over $\Omega$ is $\widehat{m}_{\rho_n}(\bs{p}) \pm z_{\alpha/2} \sigma(\bs{p})\{nf(\bs{p})V_T\}^{-1/2}$, where $f(\cdot)$ stands for the density function of $\bs{p}_i$. However, it is very difficult to obtain the exact form of the standard error for general TPST estimators due to the characteristic of the trivariate spline basis functions. To overcome this, Section 5.4 proposes using a wild bootstrap method to estimate the standard errors and quantify the uncertainty of the estimators.
\end{remark}

\vskip .20in \noindent \textbf{5. Implementation details} \vskip .10in
\renewcommand{\thetable}{5.\arabic{table}} \setcounter{table}{0} 
\renewcommand{\thefigure}{5.\arabic{figure}} \setcounter{figure}{0}
\renewcommand{\theequation}{5.\arabic{equation}} \setcounter{equation}{0} 
\label{SEC:imp}

This section provides some implementation details on how to construct the penalty matrix $\mathbf{P}$ and select the penalty parameter $\rho_n$ in (\ref{EQU:objective_2}). To facilitate discussion, we first introduce the directional derivatives for basis functions, followed by the construction details of the penalty matrix, the selection criteria of the penalty parameter, and triangulation selection.

\vskip .10in \noindent \textbf{5.1. Construction of Penalty Matrix} \vskip .10in
\label{SUBSEC:P}

We list here the key steps in the construction of the penalty matrix $\mathbf{P}$ in (\ref{EQN:EsP}). One can refer to Section A.6 in Appendix A for more details.

For a general multivariate smooth function $\phi$, the directional derivative at point $\bs{p}$ with respect to direction $\bs{u}$ is defined as
\[
D_{\bs{u}}\phi(\bs{p}) := \frac{\partial}{\partial t}\phi(\bs{p} + t\bs{u}) \Big|_{t = 0} = \lim_{t\rightarrow0}\frac{\phi(\bs{p} + t\bs{u}) - \phi(\bs{p})}{t}.
\]
Accordingly, for vector $\bs{u}:=(u_x, u_y, u_z)\in\mathbb{R}^3$ and trivariate function $\phi$, the directional derivative at $\bs{p}=(x,y,z)$ is
\[
D_{\bs{u}}\phi(x,y,z) := \frac{\partial}{\partial t}\phi(x + tu_x, y + tu_y, z + tu_z) \Big|_{t = 0}.
\]
Then for the introduced Bernstein basis function with degree $d$, based on Lemma \ref{LEM:ddb} in Appendix A, we have
\[
D_{\bs{u}}B_{ijkl}^d(\bs{p}) = d\left\{a_1 B_{i-1,j,k,l}^{d-1}(\bs{p}) + a_2 B_{i,j-1,k,l}^{d-1}(\bs{p}) + a_3 B_{i,j,k-1,l}^{d-1}(\bs{p}) + a_4 B_{i,j,k,l-1}^{d-1}(\bs{p})\right\},
\]
where $(a_1, a_2, a_3, a_4)$ is the barycentric coordinate of direction $\bs{u}$. Based on this conclusion, for a tetrahedron $s_T$, the corresponding penalty term can be written as
\begin{align*}
\mathcal{E}(s_T)
&= \sum_{|\bs{\alpha}| = 2}{2 \choose \alpha_1}{2 - \alpha_1 \choose \alpha_2}\int_T \bigg\{\sum_{i + j + k + l = d} \gamma_{T; ijkl} D^{\bs{\alpha}} B_{ijkl}^{d,T}(\bs{p})\bigg\}^2 \mathrm{d}\bs{p} \nonumber \\
&= \sum_{|\bs{\alpha}| = 2} \bs{\gamma}^{\top}_T \mathbf{P}_T^{\bs{\alpha}} \bs{\gamma}_T = \bs{\gamma}^{\top}_T \mathbf{P}_T \bs{\gamma}_T, 
\end{align*}
where each $\mathbf{P}_T^{\bs{\alpha}}$  is a $\binom{d+3}{3} \times \binom{d+3}{3}$ matrix with entries $\int_T \{D^{\bs{\alpha}} B_{ijkl}^{d,T}(\bs{p})\}$ $\{D^{\bs{\alpha}} B_{i^{\prime} j^{\prime} k^{\prime} l^{\prime}}^{d,T}(\bs{p})\} \mathrm{d} \bs{p}$ for $\bs{\alpha}$ satisfying $|\bs{\alpha}| = 2$. 
As for the penalty term defined in the whole domain, recall that $\mathcal{E}(s) = \sum_{T \in \triangle} \mathcal{E}(s_T)$ and $\mathcal{E}(s_T) = \bs{\gamma}_T^{\top} \mathbf{P}_{T}\bs{\gamma}_T.$ Therefore, $\mathcal{E}(s) = \bs{\gamma}^{\top} \mathbf{P} \bs{\gamma}$, where $\bs{\gamma}= (\bs{\gamma}_1^\top, \ldots, \bs{\gamma}_N^{\top})^{\top}$, and $\mathbf{P} = \textrm{diag}(\mathbf{P}_T, T \in \triangle)$ is a block diagonal matrix.

\vskip .10in \noindent \textbf{5.2. Penalty parameter selection} \vskip .10in
\label{SUBSEC:GCV}

To balance the bias and variance of the proposed estimator and achieve a good estimation and prediction performance, it is crucial to choose a suitable value of the penalty parameter $\rho_n$. Since the in-sample fitting errors can not gauge the prediction accuracy of the fitted function, we select a criterion function that attempts to measure the out-of-sample performance of the fitted model. Minimizing the generalized cross-validation (GCV) criterion is one computationally efficient approach to selecting smoothing parameters that also has good theoretical properties. 

Note that 3D object data are often generated with spatial dependence. When performing cross-validation (CV), these dependence structures are usually ignored, leading to underestimating the predictive error \citep{Roberts:etal:17}. To tackle this problem, one can adopt the block CV strategy \citep{Roberts:etal:17,Valavi:18}. To be specific, all the sample points are first divided into 3D blocks with similar volumes. Then, these blocks are randomly allocated to the CV folds. In this paper, we adopt the triangulation to divide the 3D domain into small 3D blocks. Each tetrahedron is considered as one single 3D block. Figure \ref{FIG:blockCV} shows an example of block CV using a triangulation. In Figure \ref{FIG:blockCV}, we divide the domain into 504 tetrahedra and randomly assign these tetrahedra into five folds with colors indicating different folds.

\begin{figure}[htbp]
\begin{center}
\includegraphics[width=1.0in]{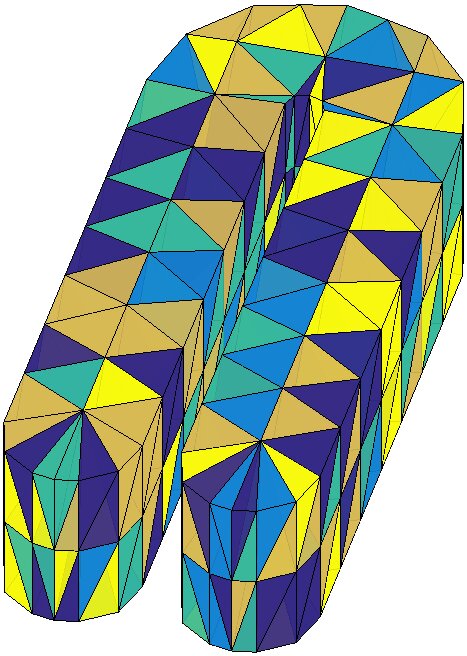} \vspace{-20pt}
\end{center} 
\caption{An illustration of block CV using a triangulation. Tetrahedra with the same color belong to the same fold.}
\label{FIG:blockCV}
\end{figure}

\vskip .10in \noindent \textbf{5.3. Triangulation selection} \vskip .10in
\label{SUBSEC:ThSelection}

As illustrated in Figure \ref{FIG:flowchart}, to form a proper triangulation, we start with surface reconstruction. A wide range of techniques has been developed to reconstruct the surface from point clouds. Our paper focuses on the methods that generalize a well-sampled point cloud to arbitrary shapes and produce a watertight surface mesh, such as a triangular mesh. Since we are interested in recovering the actual signal within the 3D point clouds, we skip the details of the surface reconstruction here. The triangulation surface representations can be coarse or fine and affect the construction of the triangulation. A very coarse triangulation could give a poor approximation to the object, while a very fine triangulation could introduce a more serious computation burden. Note that Assumption (A4') requires that the number of the tetrahedra, $N$, is larger than some minimum depending upon the degree of the spline, so theoretically, we can determine the fineness of a triangulation by setting $N \approx \lfloor c_1 n^{\gamma} \log(n)\rfloor +c_2$, in which $\lfloor\cdot\rfloor$ denotes the integer part, $c_1, c_2$ are tuning constants, and $1/(\ell+2) < \gamma < 1$. In practice, we can implement the proposed method in a coarse-to-fine resolution manner with similar criteria discussed in Section 5.2. Our extensive Monte Carlo simulation studies suggest that once a triangulation is fine enough, further refinement usually has little effect on the fitting accuracy. To avoid model over-fitting, we can stop the partition refinement when the model performance on the test set sits flat or even worsens.

\vskip .10in \noindent \textbf{5.4. Uncertainty studies} \vskip .10in
\label{SSEC:boot}

In this section, we adopt a bootstrap method to quantitatively estimate the uncertainty of the TPST estimator. A great advantage of the bootstrap method is its simplicity, which can be a straightforward way to derive estimates of standard deviations and confidence intervals for estimators of functions over complex domains. As shown in model (\ref{DEF:model}), the assumption of homoscedasticity if often invalid when a regression model is used to estimate the point clouds. The wild bootstrap \citep{Mammen:93, Hall:Horowitz:13} is specifically designed to work when the model is heteroscedastic. We conduct the following wild bootstrap procedure for estimating the standard errors.

\begin{enumerate}
\item[Step 1.] Based on the data $\{(\bs{p}_i, W_i)\}_{i=1}^{n}$, obtain the the TPST estimator $\widehat{m}(\bs{p}_i)$ described in Section 3 and the following residuals $\widehat{\epsilon}_i = W_i-\widehat{m}(\bs{p}_i)$, $i=1,\ldots,n$;
\item[Step 2.] Generate the bootstrap residuals $\{\epsilon^{\ast}_{i}\}_{i=1}^{n}$ by
$\epsilon_{i}^{\ast}=\delta_{i} \widehat{\epsilon}_i$, where $\delta_{i}=\frac{1\pm \sqrt{5}}{2}$ with probability $\frac{5\pm \sqrt{5}}{10}$, respectively. Define $W_i^{\ast} = \widehat{W}_i+\epsilon_{i}^{\ast}$.
\item[Step 3.] Apply the TPST estimator to the sample $\{(\bs{p}_i, W_i^{\ast})\}_{i=1}^{n}$, and obtain the estimated function $\widehat{m}^{\ast}(\cdot)$ over the entire domain;
\item[Step 4.] Repeat Steps 2 and 3 $B$ times and obtain a bootstrap sample of the TPST estimator as $\{\widehat{m}^{\ast}_{b}(\cdot)\}_{b=1}^{B}$. Then the standard deviation of $\widehat{m}(\bs{p})$, $\bs{p}\in \Omega$, is estimated by
\[
\left[\frac{1}{B}\sum_{b=1}^{B}\left\{\widehat{m}^{\ast}_{b}(\bs{p})-\overline{\widehat{m}}^{\ast}(\bs{p})\right\}^{2}\right]^{1/2},
\]
where $\overline{\widehat{m}}^{\ast} (\cdot)= B^{-1}\sum_{b=1}^B \widehat{m}^{\ast}_b(\cdot)$.
\end{enumerate}

\vskip .20in \noindent \textbf{6. Numerical studies} \vskip .10in
\renewcommand{\thetable}{6.\arabic{table}} \setcounter{table}{0} 
\renewcommand{\thefigure}{6.\arabic{figure}} \setcounter{figure}{0}
\renewcommand{\theequation}{6.\arabic{equation}} \setcounter{equation}{0} 
\label{SEC:simulations}

In this section, we conduct various simulation studies to assess the performance of the proposed TPST method. We use studies in Sections 6.1 and 6.2 to illustrate the capabilities of TPST in handling a variety of complex data structures, Section 6.3 to evaluate the uncertainties of the TPST estimator, and we use Section 6.4 to show the superiority of TPST in data reduction as well as signal enhancement, compared to other smoothing methods. More specifically, depending on the nature of the point clouds, we conduct experiments with unstructured (random design) and structured (fixed design) point clouds. The case of random design in Section 6.1 mimics the scenario that the point clouds are collections of 3D points distributed randomly in space. In contrast, the case of structured design is studied in Section 6.2. In this structured design setting, the point cloud is reconstructed from grid data; in other words, the locations of the points are deterministic.
The evaluation of standard error for the proposed TPST estimator is illustrated in Section 6.3.
In Section 6.4, we explore the practical performance of TPST on a structured point cloud by using PET scan data as a reference.

\vskip .10in \noindent \textbf{6.1. Unstructured complete point clouds} \vskip .10in
\label{SSEC:eg1}

In this example, we consider an unstructured (random design) point cloud in which the observations are randomly generated over the entire domain. We set the number of observations in each point cloud as $n = 20,000$ and $50,000$. To mimic some complicated scenarios in practice, we generate the point clouds from the following two domains: (i) a cuboid with a hole inside ($\Omega_1$); (ii) a 3D horseshoe ($\Omega_2$). These domains are illustrated in Figure \ref{FIG:Surface_Eg23}. To extract the underlying signal from point clouds, we consider the model in (\ref{DEF:model}), where the random noises, $\epsilon_i$'s, are assumed to be independent and identically distributed and follow a normal distribution $N(0,\sigma^2)$. For each domain, we consider two types of underlying functions with different degrees of variation: $m_1$ and $m_2$ for $\Omega_1$ (Figure \ref{FIG:Surface_Eg23} (c) and (e)), and $m_3$ and $m_4$ for $\Omega_2$ (Figure \ref{FIG:Surface_Eg23} (i) and (k)). The noises level, $\sigma$, is chosen according to the peak signal-to-noise ratios (PSNR) defined as
\[
\textrm{PSNR} = 20\log_{10} \left\{\max_i m(\bs{p}_i)/\sigma\right\}.
\]
In this study, we set $\mathrm{PSNR} = 5$ and $10$, representing scenarios of high and moderate noise levels, respectively. We use this study to investigate the effect of the size of the point cloud, degree of spline polynomial, and triangulation. We also compare the proposed method with the traditional tensor product spline method \citep{Stone:94} and thin plate spline smoothing \citep{Wood:03}. We implement the tensor product spline and the thin plate spline smoothing using the function \texttt{gam} in R package \texttt{mgcv}. For a fair comparison, we set the dimension of the basis used to represent the underlying function for the thin plate spline and tensor product spline to be very similar to the dimension we used for the TPST method. As a result, the dimension after the numerosity reduction is comparable among different methods.

\begin{figure}[htbp!]
\begin{center}
\begin{tabular}{ccccccc}
    \includegraphics[scale=0.28]{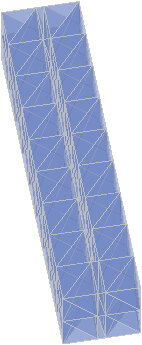} 
    &\includegraphics[scale=0.28]{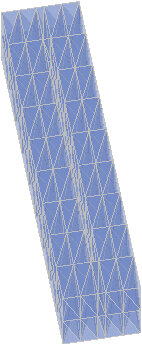}
    &\includegraphics[scale=0.14]{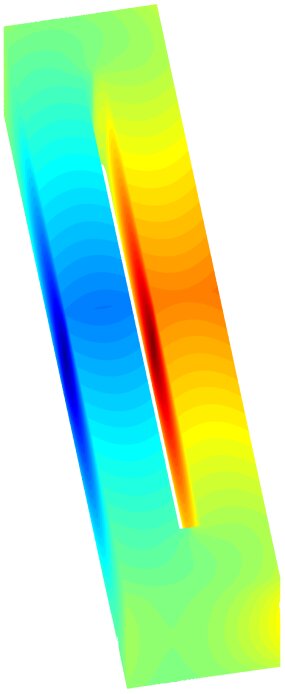} 
    &\includegraphics[scale=0.28]{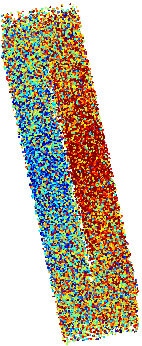}
    &\includegraphics[scale=0.14]{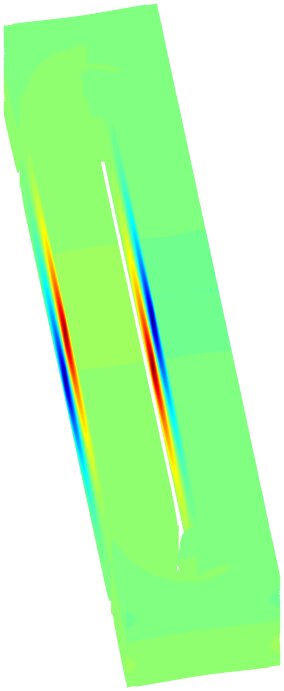}
    &\includegraphics[scale=0.28]{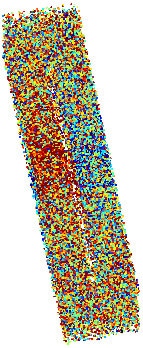}\\
    (a) &(b) &(c) &(d) &(e) &(f)\\
 
    \includegraphics[scale=0.22]{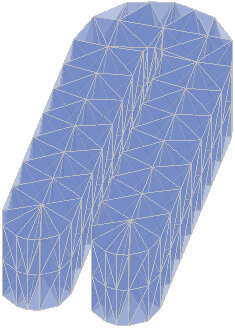} 
    &\includegraphics[scale=0.22]{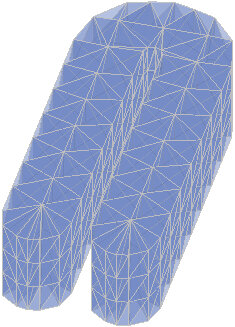}
    &\includegraphics[scale=0.11]{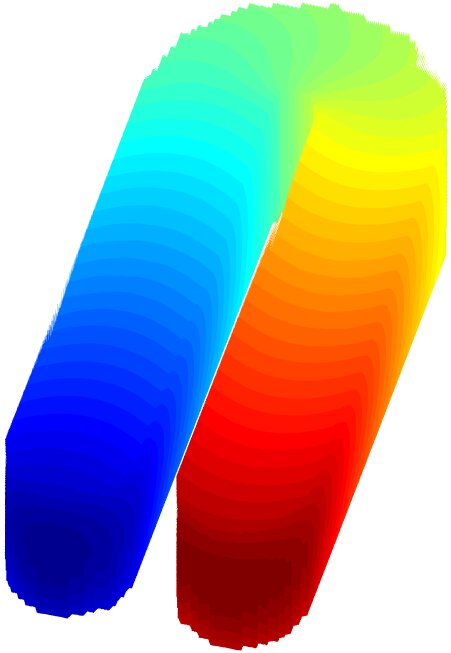} 
    &\includegraphics[scale=0.22]{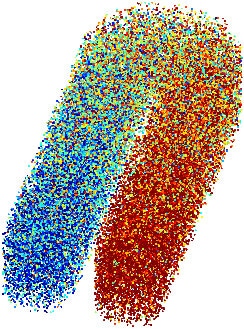}
    &\includegraphics[scale=0.11]{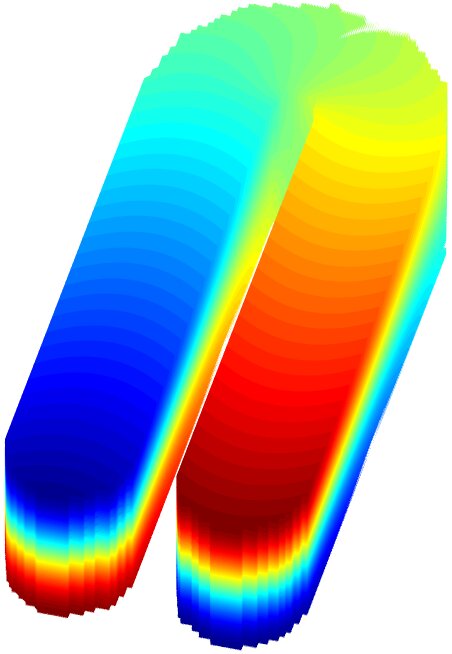}
    &\includegraphics[scale=0.22]{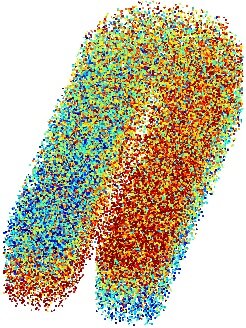}\\
    (g) &(h) &(i) &(j) &(k) &(l) 
\end{tabular} 
\end{center}
\caption{Triangulations, true functions and observed point clouds. Plots (a)--(f), are based on domain $\Omega_1$; Plots (g)--(l), are based on domain $\Omega_2$: (a) $\triangle_{11}$ for $\Omega_1$; (b) $\triangle_{12}$ for $\Omega_1$; (c) true function $m_1$; (d) observed point clouds with underlying function $m_1$; (e) true function $m_2$; (f) observed point clouds with underlying function $m_2$; (g) $\triangle_{21}$ for $\Omega_2$; (h) $\triangle_{22}$ for $\Omega_2$; (i) true function $m_3$; (j) observed point clouds with underlying function $m_3$; (k) true function $m_4$; (l) observed point clouds with underlying function $m_4$.}
\label{FIG:Surface_Eg23}
\end{figure}

To evaluate the estimation and prediction performance of each method, we calculate the out-of-sample mean integrated squared error (MISE). Figure \ref{FIG:scatterplot_random} presents the average of the MISEs over 200 replications for all different scenarios. Based on Figure \ref{FIG:scatterplot_random}, one can observe that as the size of the point cloud or the PSNR increases, the estimation and prediction accuracy improves for all the methods. Regardless of the scenarios, the proposed TPST method outperforms the other two traditional techniques with a similar dimension of numerosity reduction. These results indicate that the proposed TPST can better handle the ``leakage" problem over the irregular domain than conventional smoothing methods.

\begin{figure}[htbp!]
\begin{center}
\begin{tabular}{c}
	\includegraphics[scale = 0.46]{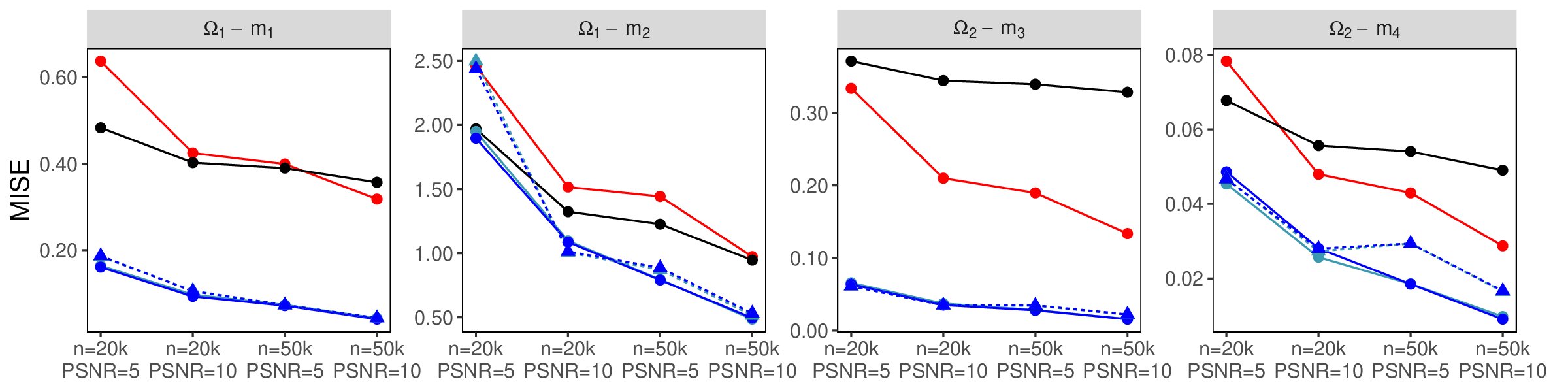}\\
	\includegraphics[scale = 0.8]{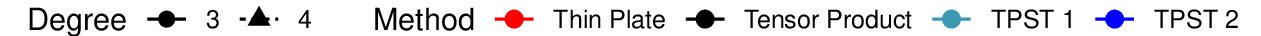}
\end{tabular} 
\end{center}
\caption{Average MISE plots for different methods on unstructured point clouds.}
\label{FIG:scatterplot_random}
\end{figure}

To evaluate the effect of triangulations on the TPST, we consider a relatively coarse mesh $\triangle_{11}$ for $\Omega_1$ and $\triangle_{21}$ for $\Omega_2$, and a fine mesh $\triangle_{12}$ for $\Omega_1$ and $\triangle_{22}$ for $\Omega_2$. An illustration of these partitions is given in the Figure \ref{FIG:Surface_Eg23} (a), (b), (g), and (h), and the number of vertices and the number of tetrahedra are summarized in Table \ref{TAB:tetrahedron}. From Figure \ref{FIG:scatterplot_random}, we can see that for each domain, the TPST estimator has a similar performance based on the two partitions. As discussed in Section 5.3, for TPST, when the number of the tetrahedra is sufficiently large to capture the pattern and features of the underlying function, more delicate triangulations will not benefit the estimation or the data reduction. For example, for domain $\Omega_2$, when we keep $d=3$, the remaining dimensions after the numerosity reduction are 323 and 340 for the two different triangulations, respectively.

\begin{table}[htbp!]
\centering
\caption{The number of tetrahedra (vertices) in each triangulation.}
\begin{tabular}{cc|cc}\hline\hline
$\Omega_1:~\triangle_{11}$ &$\Omega_1:~\triangle_{12}$ &$\Omega_2:~\triangle_{21}$ &$\Omega_2:~\triangle_{22}$\\ \hline
240 (120) &456 (180) &504 (207) & 756 (276)\\ \hline\hline
\end{tabular}
\label{TAB:tetrahedron}
\end{table}

Furthermore, based on Figure \ref{FIG:scatterplot_random}, one can see that the difference between estimates with $d = 3$ and $d = 4$ are relatively small. When we use a larger $d$, the estimators are usually less biased but with larger variance and more computationally intensive. In this example, $d = 3$ is preferred as they are more efficient in data reduction. In general, the choice of $d$ depends on the smoothness of the underlying function, the strength of signals, and computing resources.

As discussed previously, the spline methods are computationally efficient since they provide a global estimator. In this simulation example, the TPST method takes less than twenty seconds to fit the model for most of the simulation samples on a single Intel E5-2640 v3 core. The speed is comparable to tensor product and thin plate spline estimators.

\vskip .10in \noindent \textbf{6.2. Structured point clouds with missing data} \vskip .10in
\label{SSEC:eg2}

One of the critical aspects of analyzing point clouds data is handling the {unevenly distributed point clouds and/or} missing values. We use this example to investigate further the proposed method on a structured point cloud at different resolutions with different missing schemes. To be more specific, we mimic three types of missing schemes, including (i) complete data with no missing, (ii) missing at random, and (iii) missing in a contiguous block as well as at random. Figures \ref{FIG:MissTypes} illustrate different types of missing data we deal with in this example. Furthermore, we explore various missing rates under different missing mechanisms. For missing at random, we consider the missing rates ranging uniformly from 0 to 0.5, where 0 represents no missing voxels, and 0.5 means half of the voxels are missing. The contiguous block shown in Figure \ref{FIG:MissTypes} contains 12\% of the data.

\begin{figure}[htbp!]
\begin{center}
\begin{tabular}{ccc}
	\includegraphics[scale=0.28]{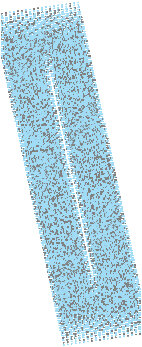}
	&\includegraphics[scale=0.28]{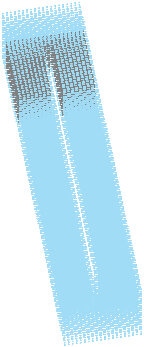}
	&\includegraphics[scale=0.28]{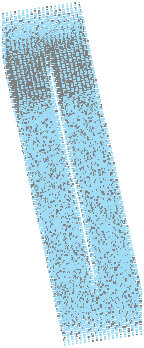}\\
	{(a) Missing at random} 
	& {(b) Missing in a block} 
	& {(c) Missing both} \\
\end{tabular} 
\end{center}
\caption{An illustration of three different missing patterns in point clouds. In (a) and (c), missing rate $ = 30\%$ and in (b), missing rate $ = 12\%$.}
\label{FIG:MissTypes}
\end{figure}

Based on the first simulation example in Section 6.1, we can see that the performance for various methods is relatively consistent. Thus, we only consider the first domain, $\Omega_1$, and try the same triangulations in this example. We consider two fixed resolutions here: the lower resolution/scale is $60\times20\times20$ with 22,160 voxels falling within the domain, and the higher resolution/scale is $75\times25\times25$ with 42,600 voxels inside.

Similar to Section 6.1, we calculate the average MISEs over 200 replications with different missing types and missing rates and illustrated in Figure \ref{FIG:scatterplot_fix}. Based on this Figure, one can see that the prediction accuracy improves for all of the methods as the missing rate decreases. The proposed methods outperform the two traditional methods regardless of the type of missing scheme and the missing rate. Furthermore, the missing type does not affect the proposed method very much. In contrast, the thin plate spline smoothing usually performs better for missing at random, while the tensor product spline is better when missing a contiguous block and at random. 

\begin{figure}[htbp!]
\begin{center}
\begin{tabular}{cccc}
	\multicolumn{4}{c}{Random Missing}\\ 
 	\multicolumn{4}{c}{\includegraphics[scale = 0.4]{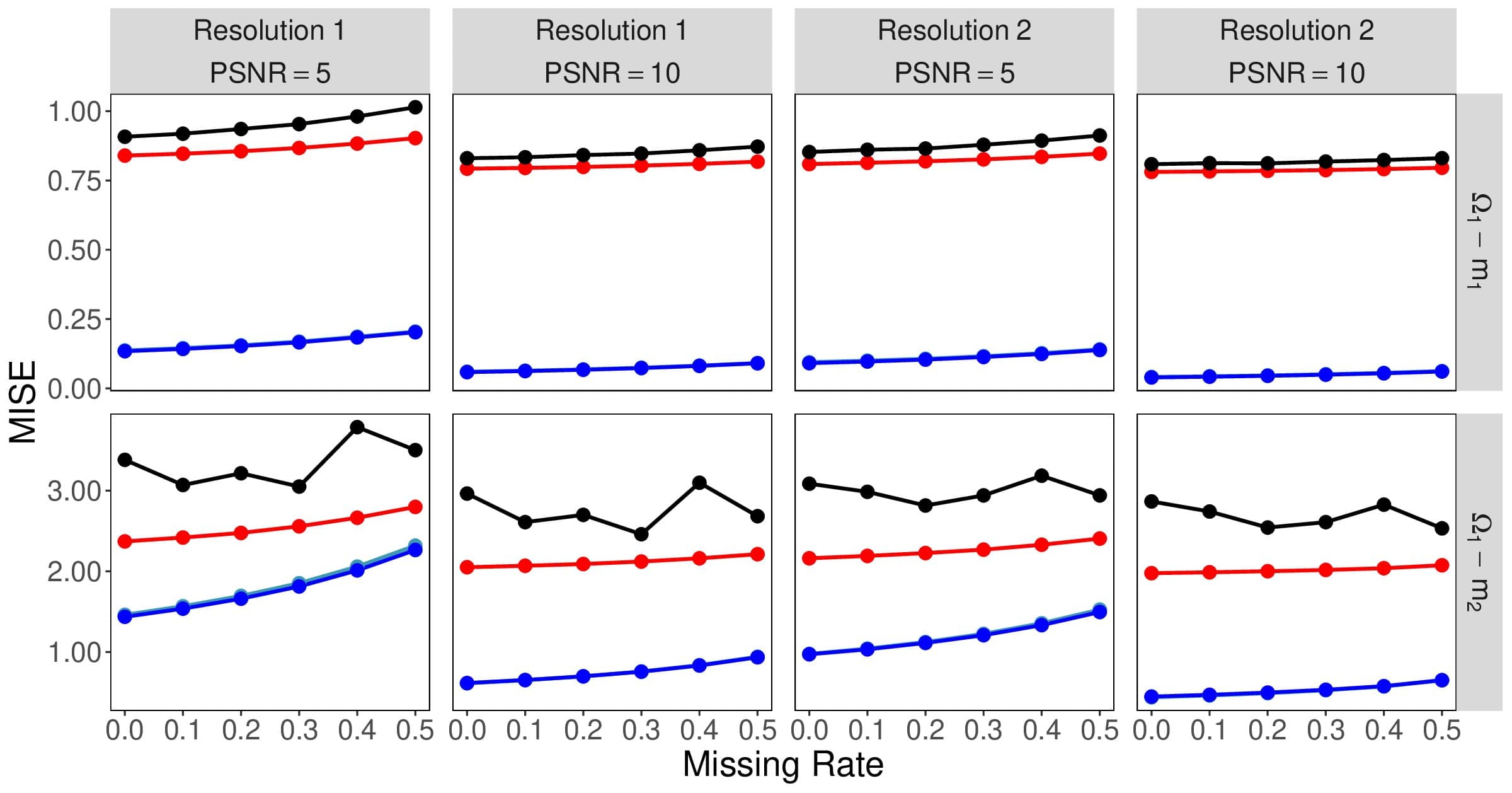}}\\
	\multicolumn{4}{c}{Tube+Random Missing}\\
 	\multicolumn{4}{c}{\includegraphics[scale = 0.4]{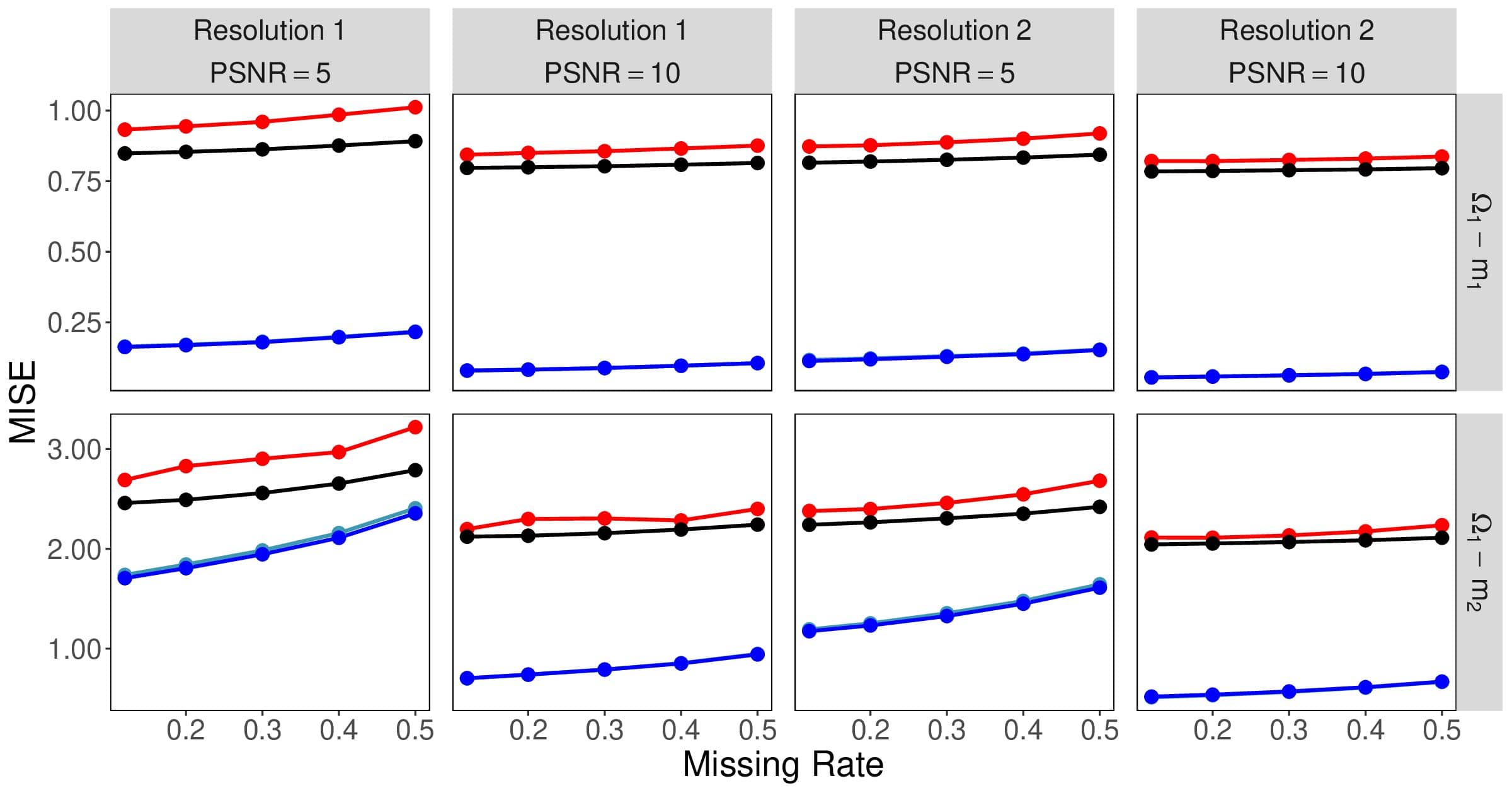}} \\
 \multicolumn{4}{c}{\includegraphics[scale = 0.65]{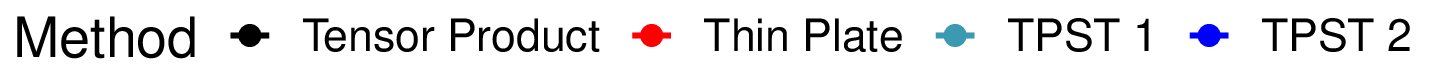}}
\end{tabular} 
\end{center} 
\caption{Average MISEs plots for different methods on structured point clouds with/without missing.}
\label{FIG:scatterplot_fix}
\end{figure}

\vskip .10in \noindent \textbf{6.3. Standard error evaluation} \vskip .10in
\label{SSEC:egBoot}

To evaluate the standard error (SE) for the proposed TPST estimator, we conduct another computational experiment using structured points clouds generated from the 3D horseshoe domain ($\Omega_2$). We consider the model in (\ref{DEF:model}), where the underlying signal is illustrated in Figure \ref{FIG:Surface_Eg23} (i), and the random noises, $\epsilon_i=\sigma(\bs{p}_i)\varepsilon_i$'s, are generated with $\sigma(\bs{p}_i)=6 - \{(x_i - 1.25)^2 + y_i^2 + z_i^2\}$ and $\varepsilon_i\sim N(0,1)$ iid. We consider a fixed solution/scale of $101\times65\times17$ with $93,449$ voxels falling within the domain.

To quantify the uncertainty of the estimator, we generate 200 replications as in Sections 6.1 and 6.2. For each replication, we calculate the bootstrap SEs over $100$ bootstrap samples using the wide bootstrap method introduced in Section 5.4. We then compute the the mean ($\mathrm{SE}_{\text{mean}}$) and median ($\mathrm{SE}_{\text{median}}$) of these bootstrap SEs across 200 replications. Additionally, we calculate the standard deviation of the TPST estimator based on the Monte Carlo samples ($\mathrm{SE}_{\text{mc}}$), which is served as the true value for SE. The results are displayed in Figure \ref{FIG:SE} (a) and (b) for $\mathrm{SE}_{\text{mean}}$ and $\mathrm{SE}_{\text{median}}$, respectively, and in Figure \ref{FIG:SE} (c) for $\mathrm{SE}_{\text{mc}}$. From these plots, one sees that the $\mathrm{SE}_{\text{mean}}$ and $\mathrm{SE}_{\text{median}}$ is very close to $\mathrm{SE}_{\text{mc}}$, which verifies the accuracy of the proposed bootstrap SE estimation method.

\begin{figure}[htbp!]
\begin{center}
\begin{tabular}{ccc}
    \includegraphics[width = 1in]{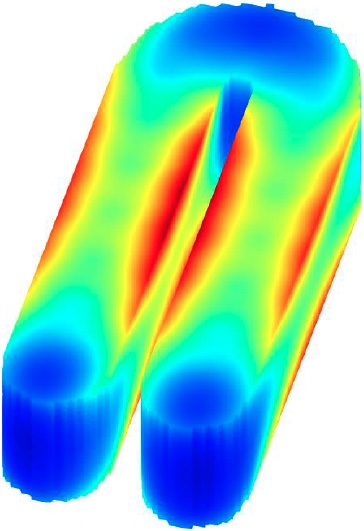} 
    &\includegraphics[width = 1in]{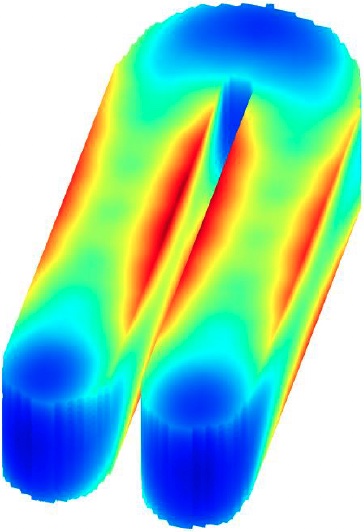}
    &\includegraphics[width = 1in]{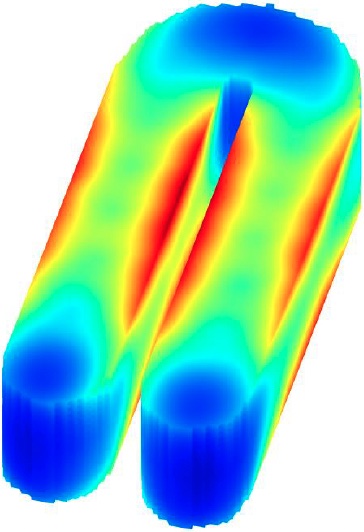}\\
    (a) $\mathrm{SE}_{\text{mean}}$ 
        & (b) $\mathrm{SE}_{\text{median}}$
        & (c) $\mathrm{SE}_{\text{mc}}$
\end{tabular} 
\end{center} 
\caption{{Standard error plot of the proposed TPST estimators: (a) and (b) are the mean ($\mathrm{SE}_{\text{mean}}$) and median ($\mathrm{SE}_{\text{median}}$) of these bootstrap SEs across 200 replications, respectively; and (c) shows the standard deviation of the TPST estimator based on the 200 Monte Carlo samples ($\mathrm{SE}_{\text{mc}}$).}}
\label{FIG:SE}
\end{figure}

\vskip .10in \noindent \textbf{6.4. Biomedical imaging analysis} \vskip .10in
\label{SSEC:eg3}

In our increasingly aging societies, Alzheimer's disease (AD) has become the most frequent cause of dementia. Much progress has been made in assisting the early diagnosis of AD with neuroimaging techniques. One widely used neuroimaging technique is PET imaging, which can also be considered as an example of structured point clouds. However, the traditional PET scanning technique limits the overall resolution of the brain image, and there is a lack of effective and efficient image reconstruction methods. In this example, we apply the proposed TPST method to denoise and enhance the resolution of an actual PET image while improving the storage efficiency. A visual representation of transverse, coronal, and sagittal planes is shown in Figure \ref{FIG:app} (a). {Based on this figure, we can observe that brain images are collected on a rectangular parallelepiped grid with a dimension of $79 \times 95 \times 66$. However, accurate signals are only present within the voxels of the human brain. As traditional methods, such as tensor product spline and thin plate spline, cannot handle complex domains, we train the model using all the voxels within the entire image. In contrast, our proposed TPST method can manage complex domains efficiently. Therefore, we use a brain mask to determine the domain of the human brain, resulting in 280,000 voxels within the human brain domain. These masks can be obtained by many methods, such as building a PET template by averaging normalized PET images or testing the quantile value of each voxel to exceed a given threshold.}

The results from Section 6.1 suggest that the choice of triangulation has a minimal effect on the performance of the TPST estimator as long as the number of tetrahedra is sufficient to capture the underlying pattern. In our implementation of the TPST method for brain imaging analysis, we use the triangulation depicted in Figure \ref{FIG:app} (b) that consists of 317 tetrahedra and 117 vertices. To effectively capture high-frequency oscillations present in human brain scans, we set the polynomial degree $d$ to at least 4. {Additionally, we investigate the impact of different levels of global smoothness by considering $r=0$ and $r=1$. Figure \ref{FIG:app} (c) shows an estimated PET image based on the settings $d=5$ and $r=1$, demonstrating the ability of the proposed TPST estimator to denoise the image while preserving the overall brain structure. To quantify the uncertainty of our estimations, we compute the bootstrap standard error as detailed in Section 5.4, and the logarithmic (base 10) transformation of the standard error map is presented in Figure \ref{FIG:app} (d).}

To compare the estimation accuracy, we consider four different measures: (i) root mean squared error, RMSE $= n^{-1/2}\{\sum_{i=1}^n (\widehat{W}_i-W_i)^2\}^{1/2}$, (ii) mean absolute error, MAE $= n^{-1}\sum_{i=1}^n |\widehat{W}_i - W_i|$, (iii) peak signal-to-noise ratio, PSNR $= 20\log_{10} \left\{\max (W_i)/\mathrm{RMSE}\right\}$, and (iv) mean relative absolute error, MRAE $= n^{-1}\sum_{i=1}^n |\widehat{W}_i - W_i|/W_i$. The results of the TPST method are compared to those of other smoothing techniques, such as Thin Plate smoothing and Tensor Product spline. Similar to simulation studies in previous sections, we implement the tensor product spline and the thin plate spline smoothing using the function \texttt{gam} in R package \texttt{mgcv}. Since the Thin Plate and Tensor Product methods could be affected by the noises outside the actual brain, we calculate all the measurements based on the voxels within the brain domain to ensure a fair comparison. The estimation and numerosity reduction results of different methods are shown in Table \ref{TAB:app}.

If a lower target for data reduction is acceptable, smoothing methods can achieve higher levels of accuracy. However, this comes at the cost of a larger basis expansion size. Conventional Thin Plate and Tensor Product splines, due to their specific functional form, may result in a high-dimensional representation of the data, particularly in the presence of significant noise or random fluctuations. This can increase the risk of overfitting, as the model may become overly complex and struggle to accurately capture the underlying patterns in the data. Additionally, the high dimensionality can result in memory constraints and make it difficult to process larger datasets. The TPST approach, however, addresses these limitations by providing better control over local variations in the data, resulting in a more compact representation compared to Thin Plate and Tensor Product splines. In this example, we cannot implement the Thin Plate and Tensor Product splines methods when the dimension of the spline basis exceeds 2500 on a regular PC.  As demonstrated in Table~\ref{TAB:app}, the TPST method, in contrast, is better equipped to deal with data with complex shapes or patterns due to its unique combination of sparsity, local control, adaptivity, and computational efficiency.

\begin{table}[htbp]
\caption{Estimation and numerosity reduction results.}
\label{TAB:app}
\begin{center}
\begin{tabular}{crrrrr}\hline \hline
Method &RMSE &MAE &PSNR &MRAE &Dimension\\ \hline
Thin Plate &{0.0906} &{0.0706} &{26.2514} &{0.0917} &2000\\ 
Tensor Product &{0.0945} &{0.0740} &{25.8799} &{0.0957} &2197\\ 
TPST ({$d = 5, r = 1$}) &0.0860 &0.0656 &26.7026 &0.0877 &2212\\ 
{TPST ($d = 4, r = 0$)} &{0.0647} &{0.0490} &{29.1734} &{0.0652} &{4247} \\
TPST ({$d = 6, r = 1$}) &0.0539 &0.0405 &30.7608 &0.0522 &4859\\ 
{TPST ($d = 5, r = 0$)} & {0.0461} & {0.0339} & {32.1176} & {0.0444} & {7941} \\
\hline\hline
\end{tabular}
\end{center}
\end{table}

\begin{figure}[htbp!]
\begin{center}
\begin{tabular}{cc}
    \includegraphics[scale = 0.25]{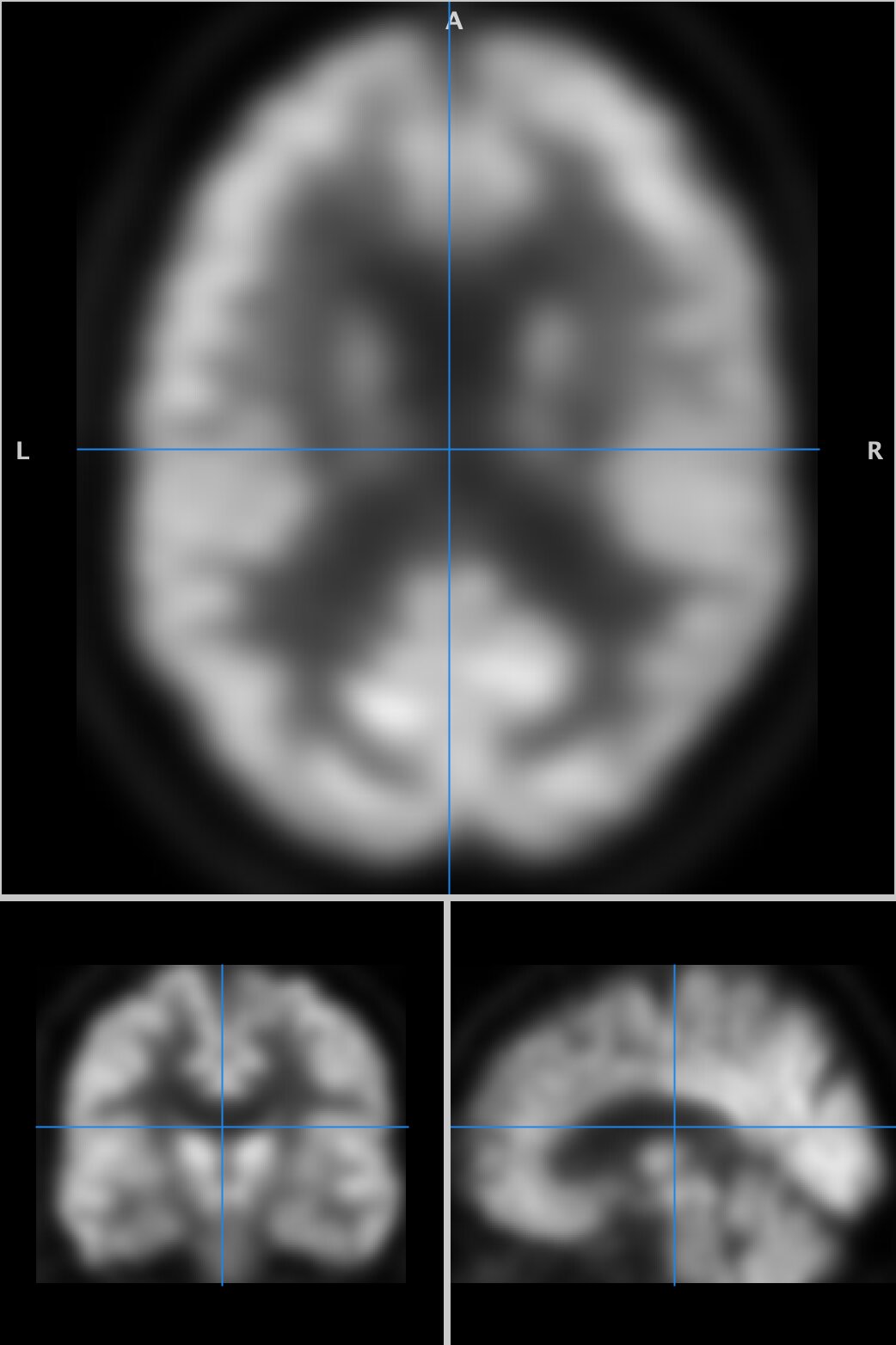} 
     & \includegraphics[scale = 0.25]{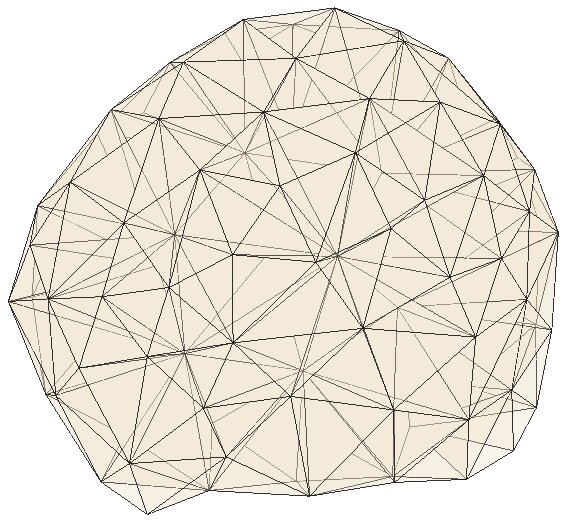} \\
     (a) Observed PET image & (b) Triangulation \\
     \includegraphics[scale = 0.1]{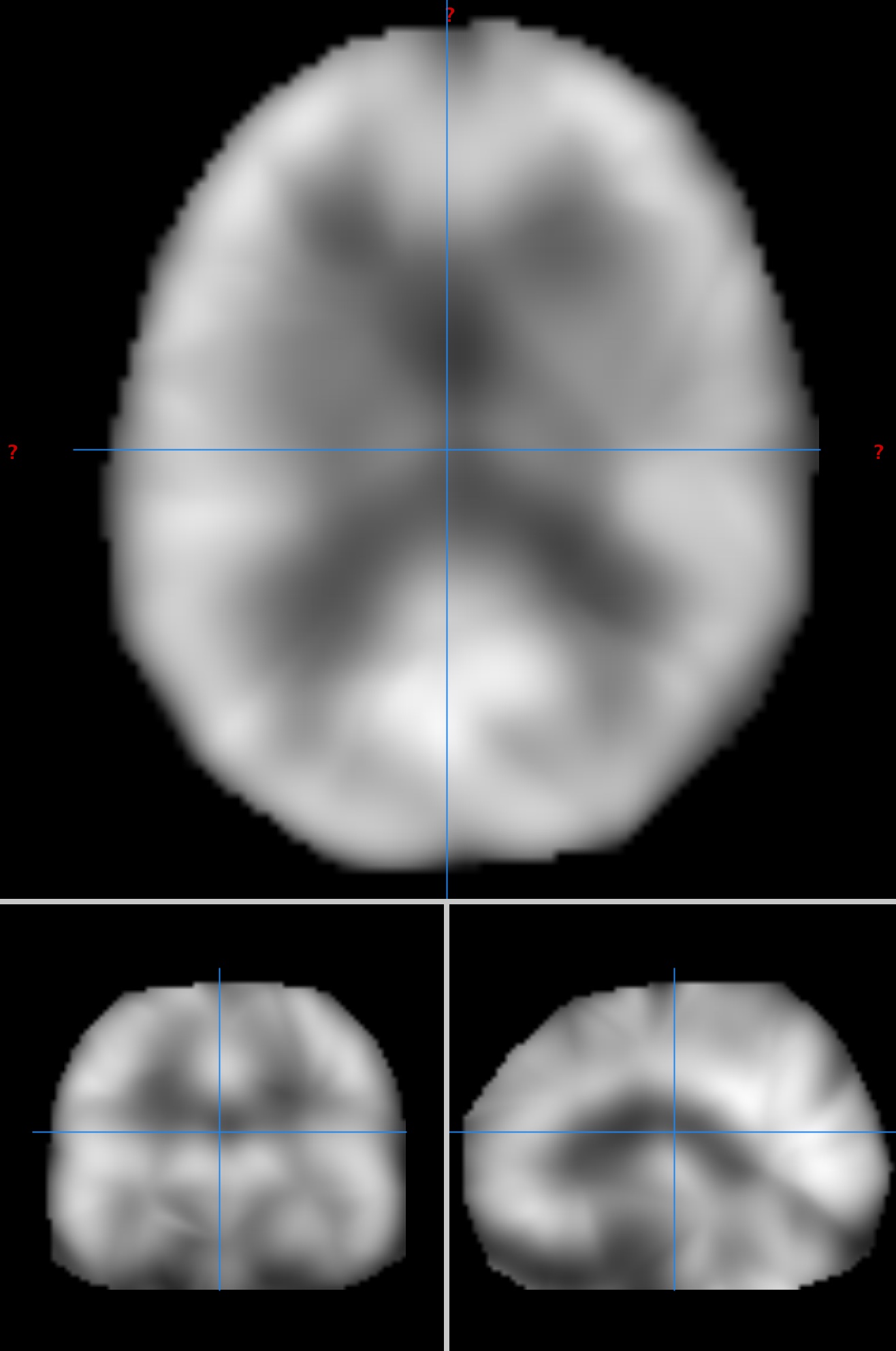}
     & \includegraphics[scale = 0.1]{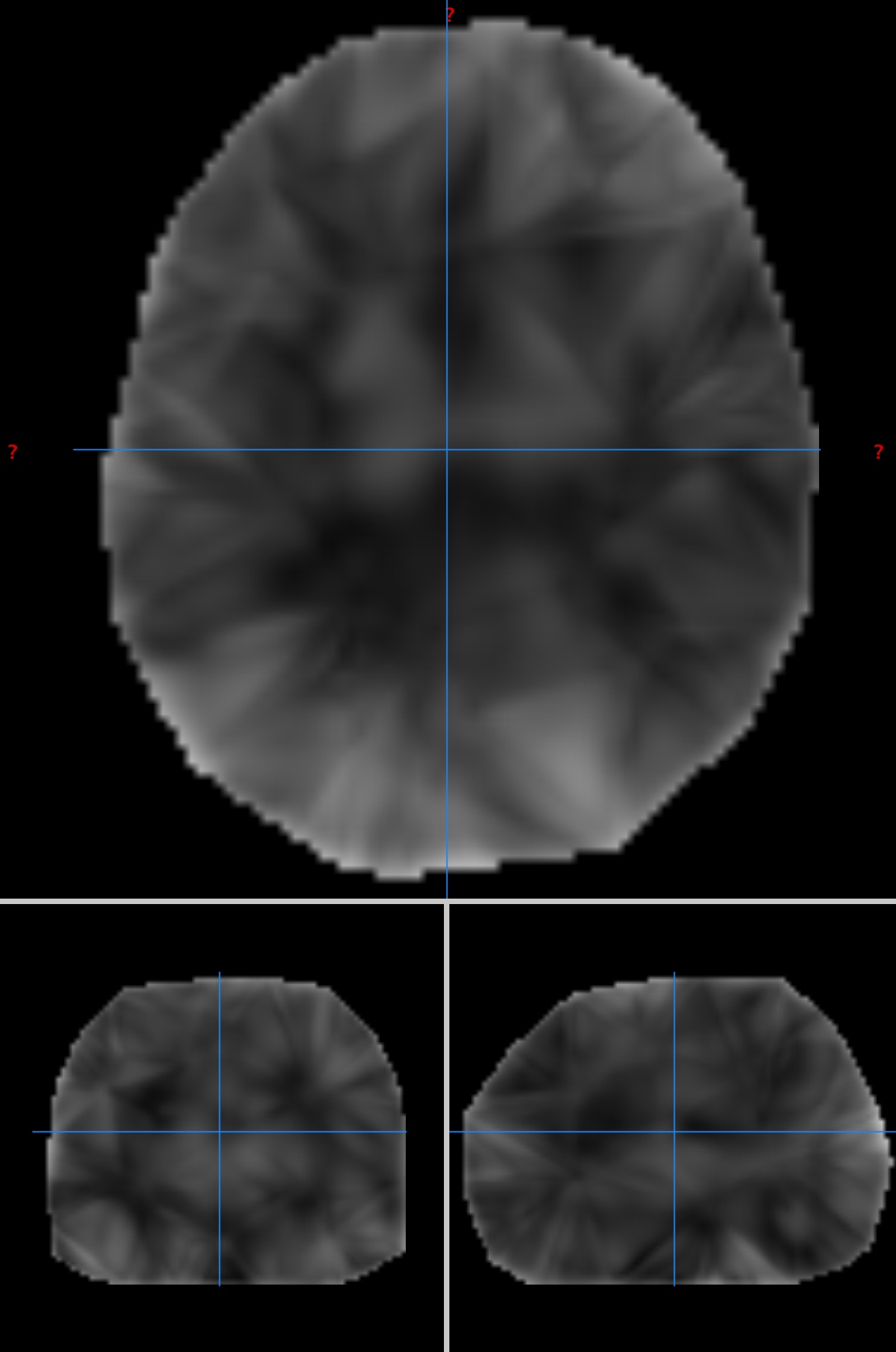}\\
	{(c) Estimated PET image}  & {(d) Bootstrap standard error map (logarithm with base 10)}
\end{tabular} 
\end{center} 
\caption{{An illustration of an observed PET image, triangulations of the brain domain, estimated PET image, and corresponding bootstrap standard error map.}}
\label{FIG:app}
\end{figure}

\vskip 0.1in  \noindent \textbf{7. Conclusions and discussion} \vskip 0.1in
\label{SEC:conclusions}

Challenges in handling and analyzing the irregularly-shaped 3D point cloud data motivated the advanced statistical methods to uncover the underlying trajectory over the point clouds. Unlike conventional smoothing methods, the proposed TPST methods are proved to perform well for complex objects. Moreover, the proposed approaches are able to handle the irregularly and regularly missing data problem, effectively denoise or deblur the data while preserving inherent geometric features or spatial structures, and providing multi-resolution reconstruction. The experimental results demonstrate the effectiveness of the proposed approaches compared to existing smoothing techniques.

Modern geoinformation technologies have provided a variety of feasible means for generating point clouds with billions of points. The enormous size of point clouds data always leads to serious time and memory consumption problems in data processing, storage, visualization, and transmission. Two approaches are currently considered as the most promising ones: parallel computing and data reduction. As a numerosity reduction method, the proposed method is an efficient data reduction technique which replaces the original big cloud point by a much smaller form of data representation. On the hand other, parallel computing is very appealing to address such ``big" data issues. Unlike tensor product spline or thin-plate spline, one unique feature of the proposed TPST is its great scalability in computing. Specifically, the spline basis function is generated restricted to each tetrahedron without any overlap, and the smoothness is achieved only by the constraints on the spline coefficients. Based on this feature, we plan to develop a parallel computing algorithm based on domain decomposition in future work.


\vskip 0.1in  \noindent \textbf{Acknowledgements} \vskip 0.1in

We express our sincere gratitude to the reviewers, the associate editor, and the action editor for their insightful comments and suggestions that greatly improved the quality of our paper. Research reported in this publication was partially supported by the National Institute Of General Medical Sciences of the National Institutes of Health under Award Number P20GM139769 (Xinyi Li), National Science Foundation awards DMS-2210658 (Xinyi Li) and DMS-2203207 (Li Wang), and Simons Foundation Mathematics and Physical Sciences-Collaboration Grant for Mathematicians \#864439 (Ming-Jun Lai) and \#963447 (Guannan Wang). The content is solely the responsibility of the authors and does not necessarily represent the official views of the National Institutes of Health.
The investigators within the ADNI contributed to the design and implementation of ADNI and/or provided data but did not participate in analysis or writing of this report. A complete listing of ADNI investigators can be found at \url{http://adni. loni.usc.edu/wp-content/uploads/how_to_apply/ADNI_Acknowledgement_List.pdf}.

\newpage

\vspace{0.8pc} \centerline{\large Supplemental Materials for \bf ``Nonparametric Regression for 3D Point Cloud Learning''}

\fontsize{12}{14pt plus.8pt minus .6pt}\selectfont
\vskip 0.1in  \noindent \textbf{A. Detailed introduction of trivariate splines on tetrahedra} \vskip 0.1in 
\renewcommand{\thetheorem}{{\sc A.\arabic{theorem}}}
\renewcommand{\thelemma}{{\sc A.\arabic{lemma}}} %
\renewcommand{\thecorollary}{{\sc A.\arabic{corollary}}}
\renewcommand{\theequation}{A.\arabic{equation}} 
\renewcommand{\theproposition}{A.\arabic{proposition}} 
\renewcommand{\thefigure}{A.\arabic{figure}} 
\renewcommand{\thetable}{A.\arabic{table}} 
\renewcommand{\thedefinition}{A.\arabic{definition}} 
\renewcommand{\theremark}{A.\arabic{remark}} 
\renewcommand{\theexample}{A.\arabic{example}} 
\renewcommand{\thesubsection}{A.\arabic{subsection}}
\setcounter{equation}{0}
\setcounter{theorem}{0}
\setcounter{lemma}{0}
\setcounter{figure}{0}
\setcounter{remark}{0}
\setcounter{proposition}{0}
\setcounter{subsection}{0}
\label{SEC:spline}

In this section, we give the detailed discussion of the trivariate splines on tetrahedra, along with the introduction for some theoretical properties. In Section A.1, we first introduce the barycentric coordinates associated with a tetrahedron and show that a trivariate polynomial can be written in a convenient form using the barycentric coordinates. In Section A.2, we describe the directional derivatives of a polynomial, and smoothness conditions for polynomials on adjoining triangular face. The introduction of trivariate splines on a triangulation is given in Section 2.3 in main part.

Section A.3 illustrates an example of the constraint matrix.
The technical details of the conclusions in Section A.2 were given in Section A.4.
Sections A.5 and A.6 give a detailed introduction and induction for the directional derivatives for basis functions and penalty matrix, respectively.

\vskip .10in \noindent \textbf{A.1. Barycentric coordinates and Bernstein basis polynomials} \vskip .10in
\label{SUBSEC:bbp}

Given a tetrahedron $T = \langle\bs{v}_1, \bs{v}_2, \bs{v}_3, \bs{v}_4\rangle$, any fixed point $\bs{p} := (x,y,z) \in \mathbb{R}^3$ has a unique representation in terms of $\langle\bs{v}_1,\bs{v}_2,\bs{v}_3,\bs{v}_4\rangle$,
\[
	\bs{p}=b_1\bs{v}_1 + b_2\bs{v}_2 + b_3\bs{v}_3 + b_4\bs{v}_4,~~ \text{ with } b_1+ b_2 + b_3 + b_4 = 1,
\]
where $(b_1,b_2,b_3,b_4)$ are called the \textsf{barycentric coordinates} of $\bs{p}$ \textsf{relative to the tetrahedron} $T$. When the point $\bs{p}$ is inside or on the faces of $T$, all $b_1$, $b_2$, $b_3$ and $b_4$ are nonnegative. By Cramer's rule, the barycentric coordinate corresponding to vertex $\bs{v}_{\imath}$, satisfies $b_{\imath} = {\mathrm{det}(\mathbf{M}_{\imath})}/{\mathrm{det}(\mathbf{M})}$, $\imath = 1, \ldots, 4$, where 
\[
	\mathbf{M}:= 
	\begin{pmatrix}
		1 & 1 & 1 & 1 \\
		x_1 & x_2 & x_3 & x_4 \\
		y_1 & y_2 & y_3 & y_4 \\
		z_1 & z_2 & z_3 & z_4 
	\end{pmatrix}
\] 
and $\mathbf{M}_{\imath}$ replaces $\mathbf{M}$'s $\imath$-th column with $(1\ x\ y\ z)^\top$.

The barycentric coordinates $(b_1,b_2,b_3,b_4)$ also have an interesting geometric interpretation. As shown in Figure \ref{Fig:Tri}, for any $\bs{p}\in T$, it divides the tetrahedron $T$ to four sub-tetrahedra, $T_1$, $T_2$, $T_3$ and $T_4$. Notice that $\mathrm{det}(\mathbf{M})=6V_T$, then $b_{\imath}=V_{T_{\imath}}/V_T$, where $T_{\imath}$ replace the vertex $\bs{v}_{\imath}$ in $V_T$ with $\bs{p}$, $\imath = 1,\ldots,4$.
\begin{figure}[htbp]
\begin{center}
\scalebox{0.85}{
\begin{tabular}{cc}
	\includegraphics[scale=0.08]{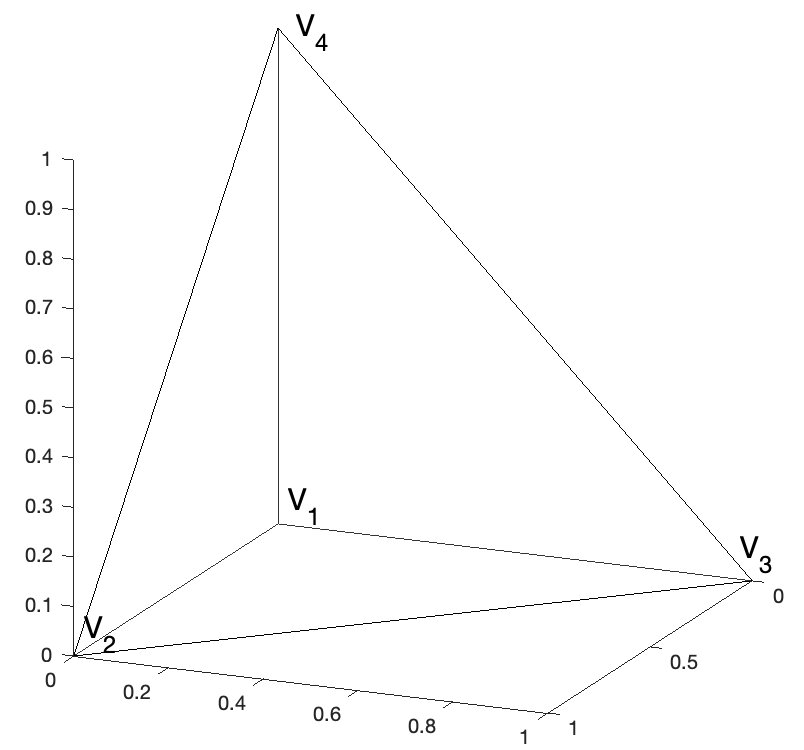}  ~~~&~~~
	\includegraphics[scale=0.2]{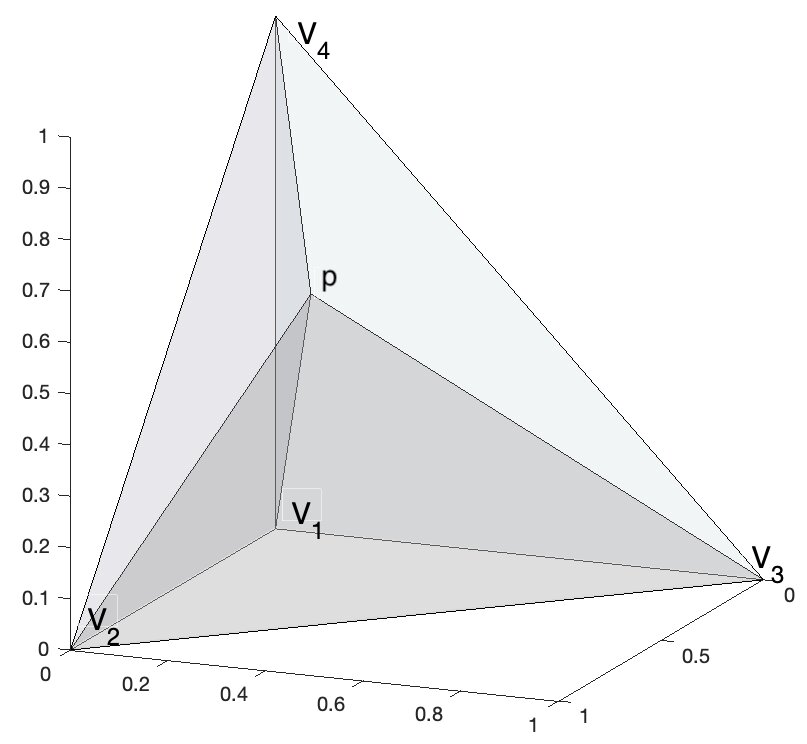} \\
\end{tabular}}
\end{center} 
\caption{An illustration of barycentric coordinates of a point $\bs{p}$ in a tetrahedron $T$. 
\label{Fig:Tri}}
\end{figure}

For a nondegenerate tetrahedron $T$ and a point $\bs{p}\in T$ with barycentric coordinates $(b_1, b_2, b_3, b_4)$, for nonnegative integers $i$, $j$, $k$, $l$ with $i+j+k+l=d$, define \textsf{trivariate Bernstein basis polynomial of degree} $d$ \textsf{relative to} $T$ as
\[
	B_{ijkl}^{d,T}(\bs{p}):=\frac{d!}{i!j!k!l!}b_1^i b_2^j b_3^k b_4^l,\text{ with } i + j + k + l = d.
\]
For any positive integer $d$ and tetrahedron $T$, let $\mathcal{P}_{d}(T)$ be the space of all trivariate polynomials defined on $T$ with degrees less than or equal to $d$.
Note that the dimension of $\mathcal{P}_{d}(T)$ is $\binom{d+3}{3}$.

According to Theorem 15.8 in \cite{Lai:Schumaker:07} and Lemma \ref{LEM:basis} in Appendix Section A.4.1, the set of Bernstein basis polynomials
\[
	\bs{B}^d_T(\bs{p}):=\{B_{ijkl}^{d,T}(\bs{p}): i,j,k,l\geq0, i+j+k+l=d\}
\]
forms a basis for the space of polynomials $\mathcal{P}_{d}(T)$. In addition, Bernstein basis functions $\{B_{ijkl}^{d, T}(\bs{p})\}_{i+j+k+l=d}$ have the following properties:
\begin{enumerate}
	\item $\{B_{ijkl}^{d, T}\}$ form a partition of unity, i.e., for all $\bs{p}\in T$, $\sum_{i+j+k+l=d}B_{ijkl}^{d, T}(\bs{p})=1$;
	\item  for all $\bs{p}\in T$, $0\leq B_{ijkl}^{d, T}(\bs{p})\leq1$;
	\item $B_{ijkl}^{d, T}$ has a unique maximum at the point $d^{-1}(i\bs{v}_1+j\bs{v}_2+k\bs{v}_3+l\bs{v}_4)$.
\end{enumerate}

\begin{remark}
Barycentric coordinates are invariant to linear transformations of Cartesian coordinates, that is, they do not depend on the orientation or location of the tetrahedra. Consequently, the trivariate splines based on Bernstein basis polynomials, which are constructed with barycentric coordinates, are also invariant to linear transformations.
\end{remark}

To further illustrate the Bernstein basis functions, we present an example of $\{B_{ijkl}^{d, T}(\bs{p})\}_{i+j+k+l=d}$ for $d=4$ in Figure \ref{Fig:Bd3}. In this example, there are $\binom{7}{3}=35$ basis functions in total. Note that all the function values vary between $0$ and $1$, and the colors scale to the quantiles of function values.

\begin{figure}[htbp]
\begin{center}
\scalebox{0.86}{
\begin{tabular}{cccccc}
	\includegraphics[scale=0.13]{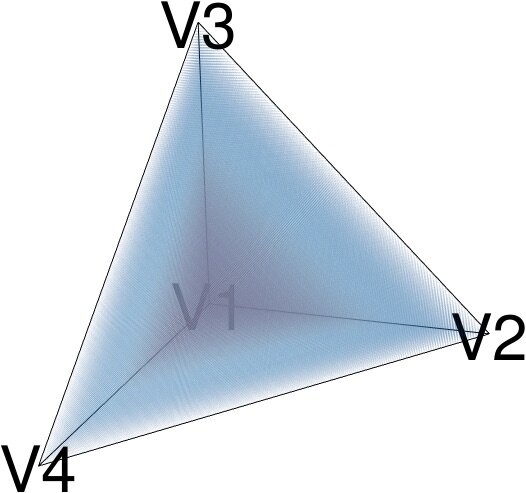} 
	\!\!\!\!&\!\!\!\!\includegraphics[scale=0.13]{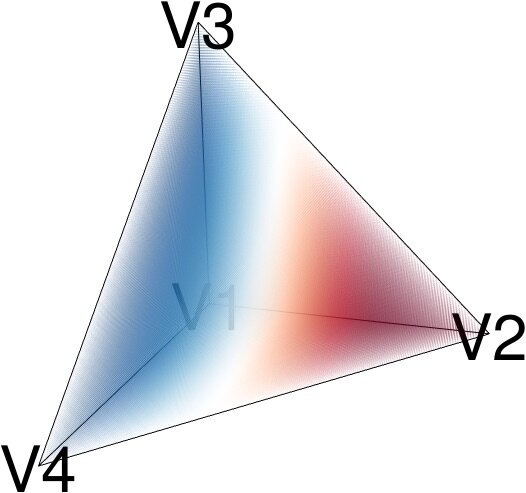} 
	\!\!\!\!&\!\!\!\!\includegraphics[scale=0.13]{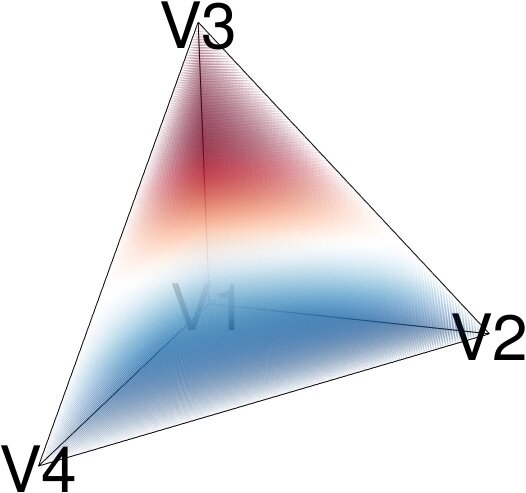} 
	\!\!\!\!&\!\!\!\!\includegraphics[scale=0.13]{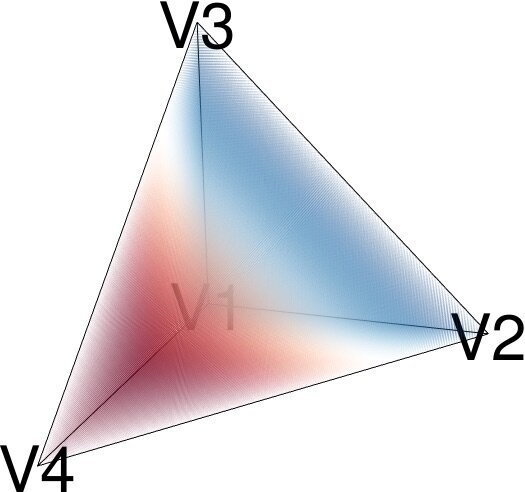} 
	\!\!\!\!&\!\!\!\!\includegraphics[scale=0.13]{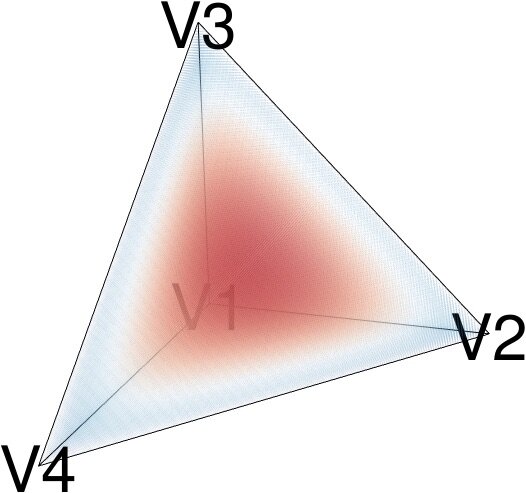}
	\!\!\!\!\!&\!\!\!\!\!\includegraphics[height=0.7in, width=0.65in]{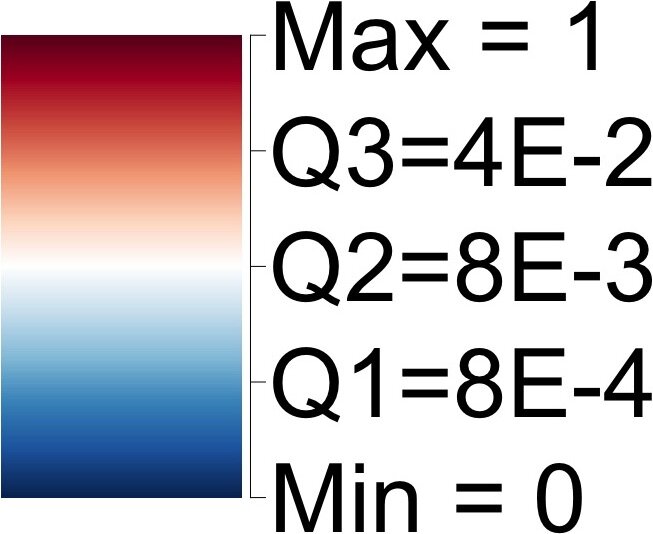}\\
	{\footnotesize $B_{4000}^{4, T}=b_1^4$} 
	\!\!\!\!\!&\!\!\!\!\! {\footnotesize $B_{0400}^{4, T}=b_2^4$}
	\!\!\!\!\!&\!\!\!\!\! {\footnotesize $B_{0040}^{4, T}=b_3^4$} 
	\!\!\!&\!\!\! {\footnotesize $B_{0004}^{4, T}=b_4^4$} 
	\!\!\!\!\!&\!\!\!\!\!\!\! {\footnotesize $B_{1111}^{4, T}=24b_1b_2b_3b_4$} 
	\!\!\!\!\!\!\!&\!\!\!\!\! {\footnotesize Color bar}  \\[3pt]
	\includegraphics[scale=0.13]{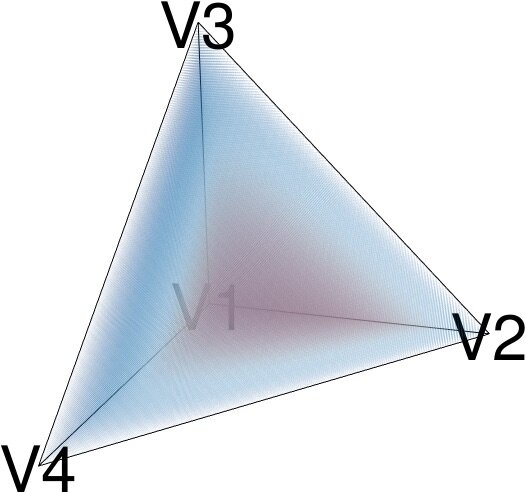} 
	\!\!\!\!&\!\!\!\!\!\includegraphics[scale=0.13]{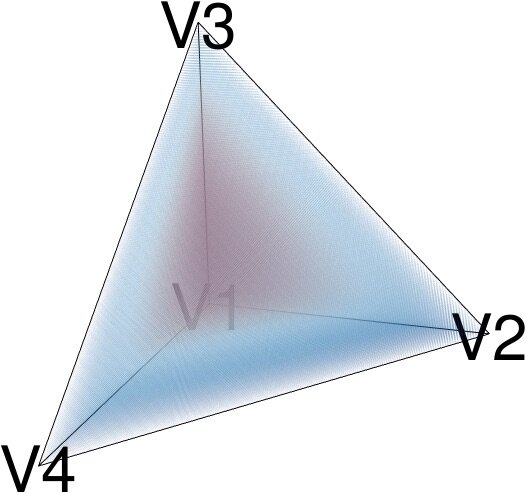} 
	\!\!\!\!\!&\!\!\!\!\!\includegraphics[scale=0.13]{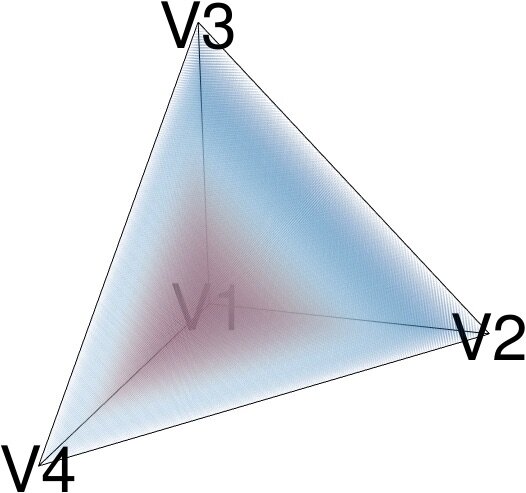}
	\!\!\!\!\!&\!\!\!\!\!\includegraphics[scale=0.13]{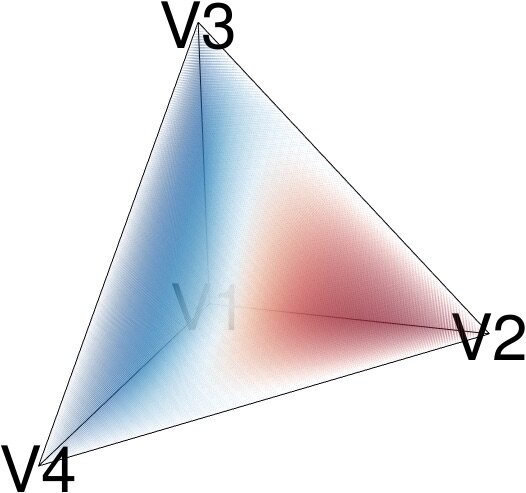}
	\!\!\!\!\!&\!\!\!\!\!\includegraphics[scale=0.13]{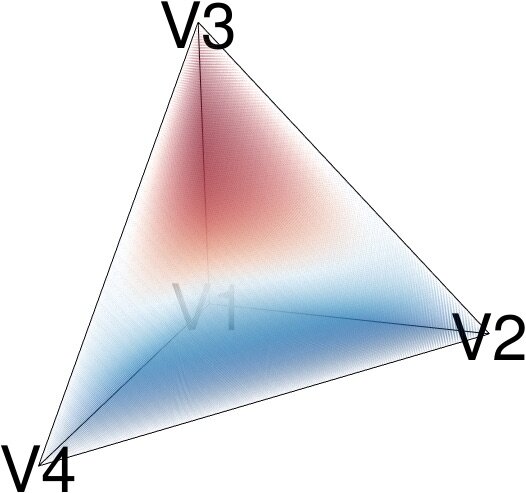}
	\!\!\!\!\!&\!\!\!\!\!\includegraphics[scale=0.13]{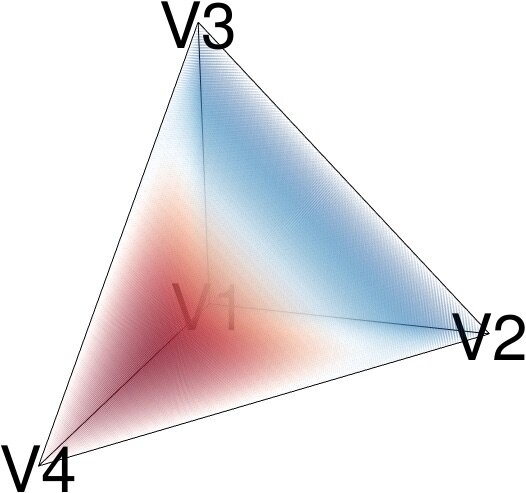}\\
	{\footnotesize $B_{3100}^{4, T}=4b_1^3b_2$} 
	\!\!\!\!\!&\!\!\!\!\! {\footnotesize $B_{3010}^{4, T}=4b_1^3b_3$} 
	\!\!\!\!\!&\!\!\!\!\! {\footnotesize $B_{3001}^{4, T}=4b_1^3b_4$} 
	\!\!\!\!\!&\!\!\!\!\! {\footnotesize $B_{1300}^{4, T}=4b_1b_2^3$} 
	\!\!\!\!\!&\!\!\!\!\! {\footnotesize $B_{1030}^{4, T}=4b_1b_3^3$} 
	\!\!\!\!\!&\!\!\!\!\! {\footnotesize $B_{1003}^{4, T}=4b_1b_4^3$}\\[3pt]
	\includegraphics[scale=0.13]{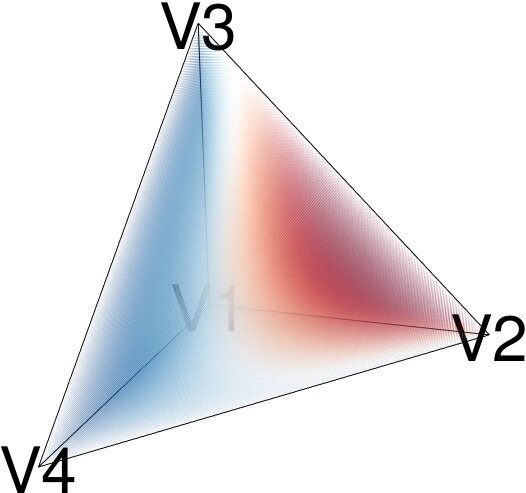} 
	\!\!\!\!&\!\!\!\!\!\includegraphics[scale=0.13]{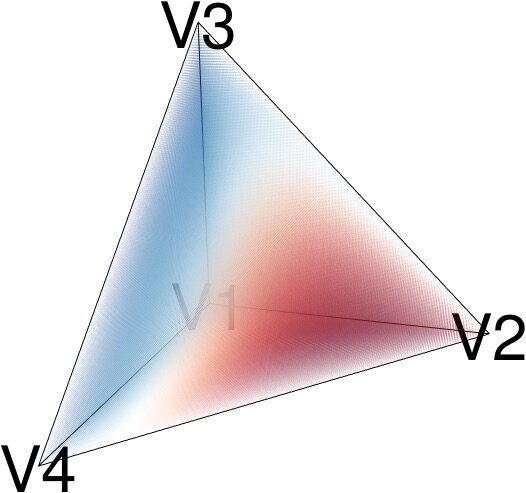} 
	\!\!\!\!\!&\!\!\!\!\!\includegraphics[scale=0.13]{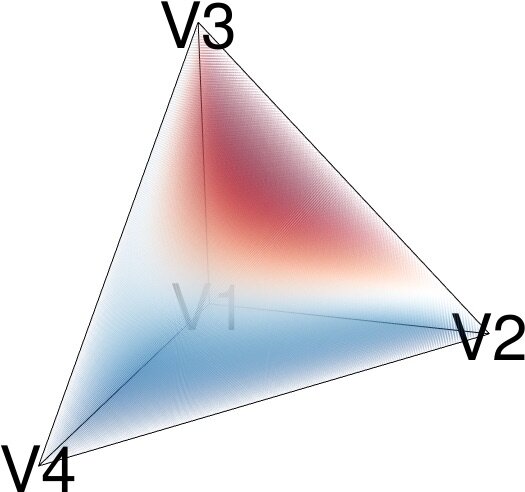} 
	\!\!\!\!\!&\!\!\!\!\!\includegraphics[scale=0.13]{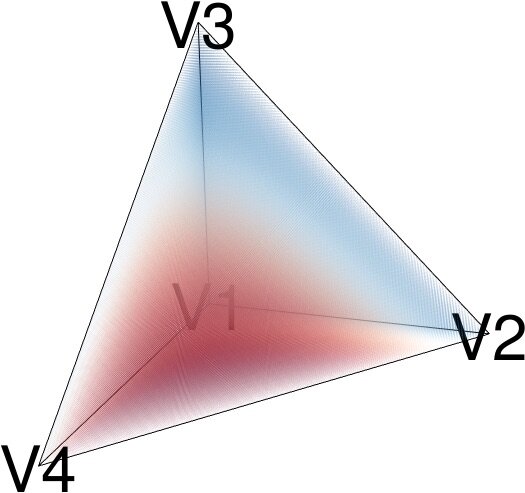} 
	\!\!\!\!\!&\!\!\!\!\!\includegraphics[scale=0.13]{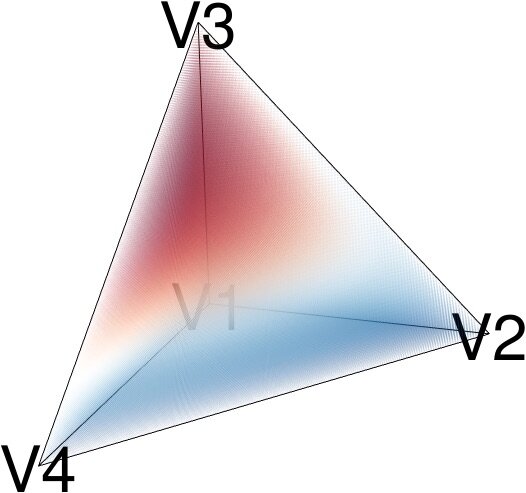}
	\!\!\!\!\!&\!\!\!\!\!\includegraphics[scale=0.13]{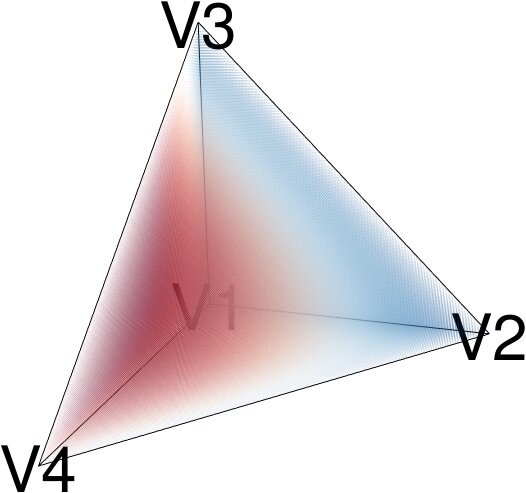}\\
	{\footnotesize $B_{0310}^{4, T}=4b_2^3b_3$} 
	\!\!\!\!\!&\!\!\!\!\! {\footnotesize $B_{0301}^{4, T}=4b_2^3b_4$} 
	\!\!\!\!\!&\!\!\!\!\! {\footnotesize $B_{0130}^{4, T}=4b_2b_3^3$} 
	\!\!\!\!\!&\!\!\!\!\! {\footnotesize $B_{0103}^{4, T}=4b_2b_4^3$} 
	\!\!\!\!\!&\!\!\!\!\! {\footnotesize $B_{0031}^{4, T}=4b_3^3b_4$} 
	\!\!\!\!\!&\!\!\!\!\! {\footnotesize $B_{0013}^{4, T}=4b_3b_4^3$}\\[3pt]
	\includegraphics[scale=0.13]{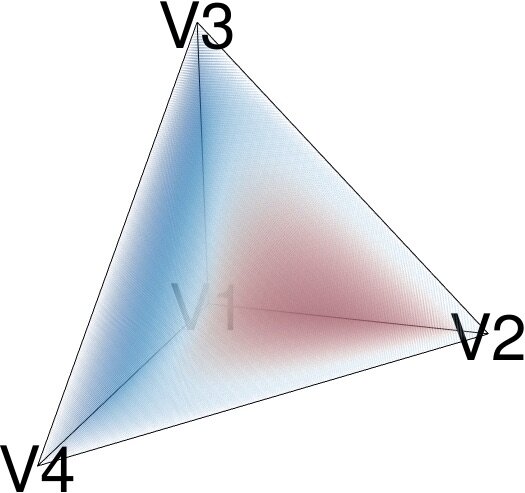} 
	\!\!\!\!&\!\!\!\!\!\includegraphics[scale=0.13]{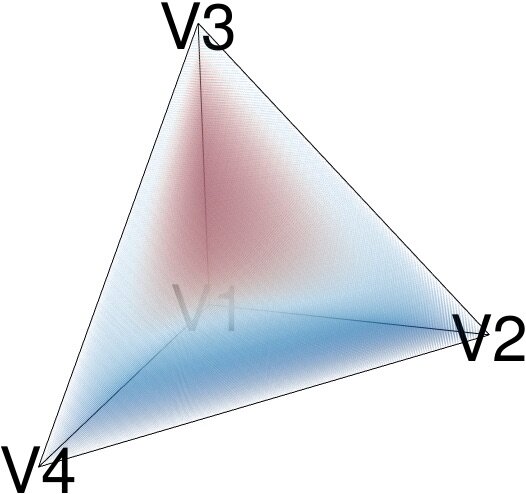} 
	\!\!\!\!\!&\!\!\!\!\!\includegraphics[scale=0.13]{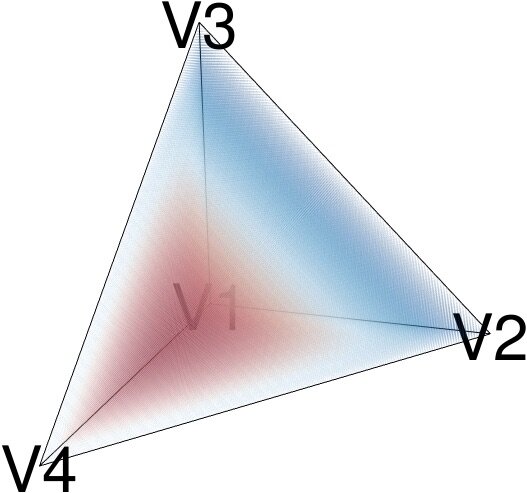} 
	\!\!\!\!\!&\!\!\!\!\!\includegraphics[scale=0.13]{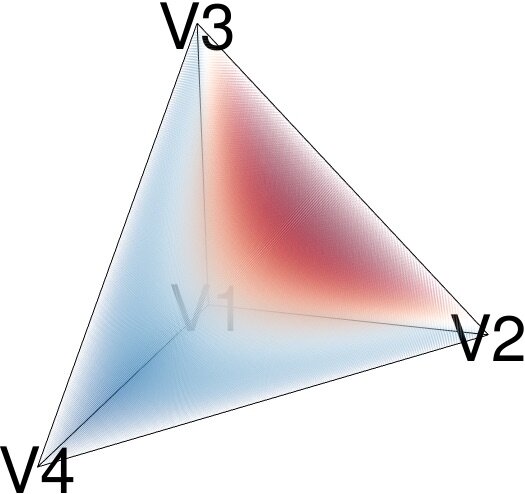} 
	\!\!\!\!\!&\!\!\!\!\!\includegraphics[scale=0.13]{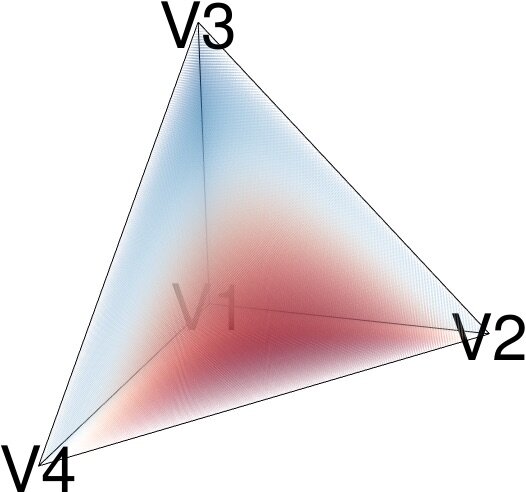}
	\!\!\!\!\!&\!\!\!\!\!\includegraphics[scale=0.13]{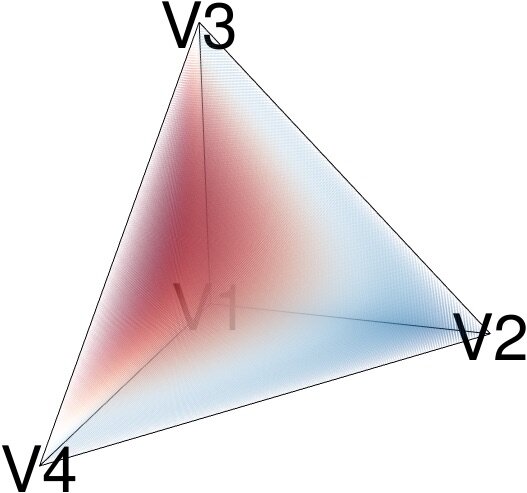}\\
	{\footnotesize $B_{2200}^{4, T}=6b_1^2b_2^2$} 
	\!\!\!\!\!&\!\!\!\!\! {\footnotesize $B_{2020}^{4, T}=6b_1^2b_3^2$} 
	\!\!\!\!\!&\!\!\!\!\! {\footnotesize $B_{2002}^{4, T}=6b_1^2b_4^2$} 
	\!\!\!\!\!&\!\!\!\!\! {\footnotesize $B_{0220}^{4, T}=6b_2^2b_3^2$} 
	\!\!\!\!\!&\!\!\!\!\! {\footnotesize $B_{0202}^{4, T}=6b_2^2b_4^2$} 
	\!\!\!\!\!&\!\!\!\!\! {\footnotesize $B_{0022}^{4, T}=6b_3^2b_4^2$}  \\[3pt]
	\includegraphics[scale=0.13]{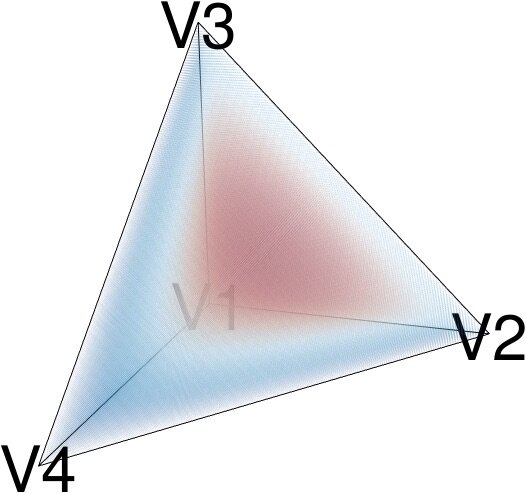} 
	\!\!\!\!&\!\!\!\!\!\includegraphics[scale=0.13]{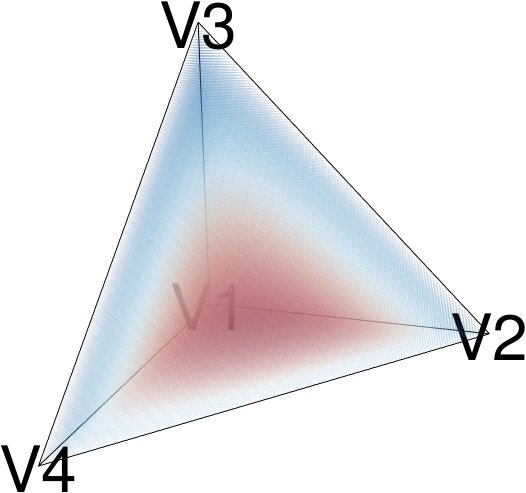} 
	\!\!\!\!\!&\!\!\!\!\!\includegraphics[scale=0.13]{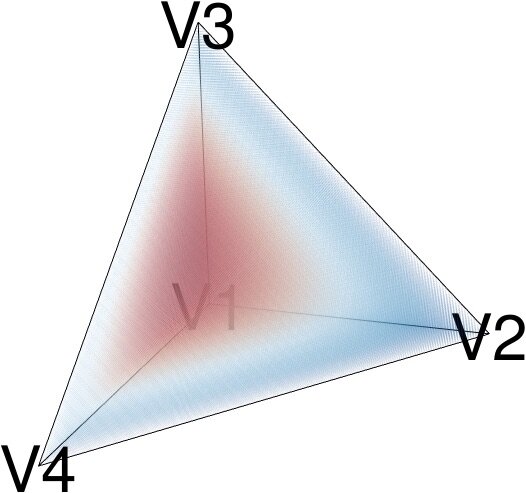}
	\!\!\!\!\!&\!\!\!\!\!\includegraphics[scale=0.13]{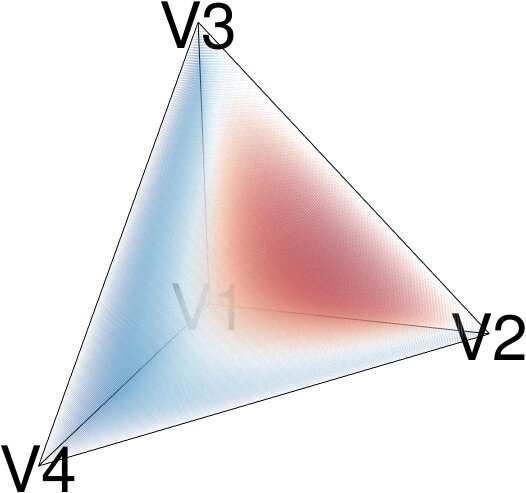}
	\!\!\!\!\!&\!\!\!\!\!\includegraphics[scale=0.13]{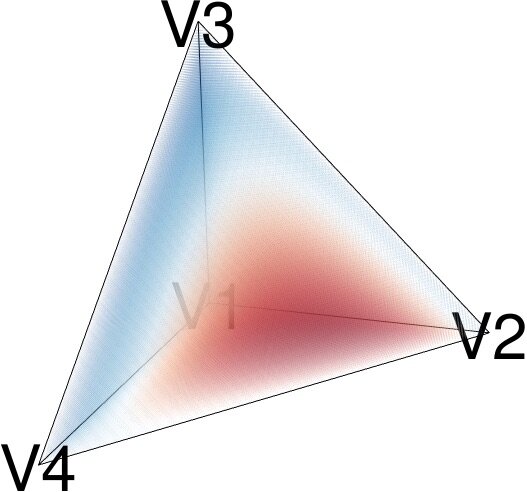}
	\!\!\!\!\!&\!\!\!\!\!\includegraphics[scale=0.13]{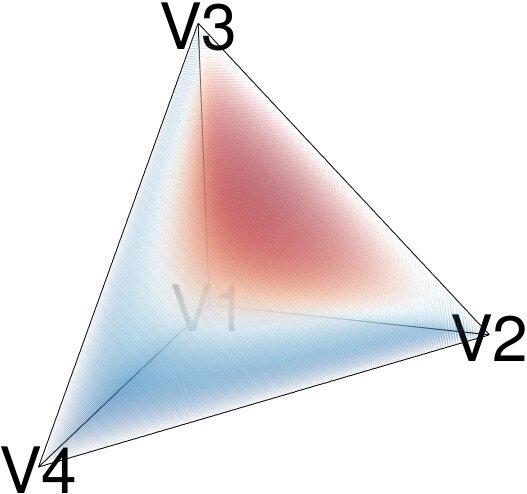}\\
	{\footnotesize $B_{2110}^{4, T}=12b_1^2b_2b_3$} 
	\!\!\!\!&\!\!\!\! {\footnotesize $B_{2101}^{4, T}=12b_1^2b_2b_4$} 
	\!\!\!\!&\!\!\!\! {\footnotesize $B_{2011}^{4, T}=12b_1^2b_3b_4$} 
	\!\!\!\!&\!\!\!\! {\footnotesize $B_{1210}^{4, T}=12b_1b_2^2b_3$} 
	\!\!\!\!&\!\!\!\! {\footnotesize $B_{1201}^{4, T}=12b_1b_2^2b_4$} 
	\!\!\!\!&\!\!\!\! {\footnotesize $B_{1120}^{4, T}=12b_1b_2b_3^2$}\\[3pt]
	\includegraphics[scale=0.13]{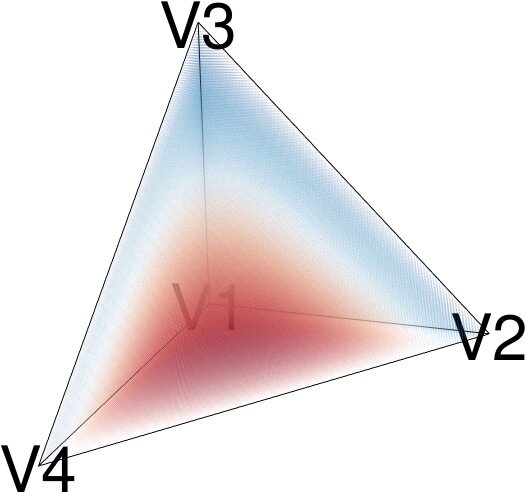} 
	\!\!\!\!&\!\!\!\!\!\includegraphics[scale=0.13]{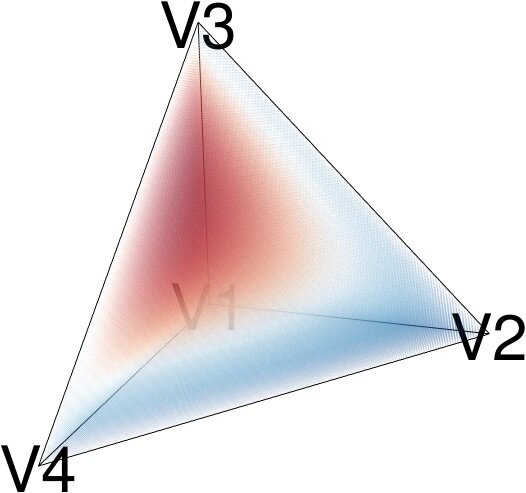} 
	\!\!\!\!\!&\!\!\!\!\!\includegraphics[scale=0.13]{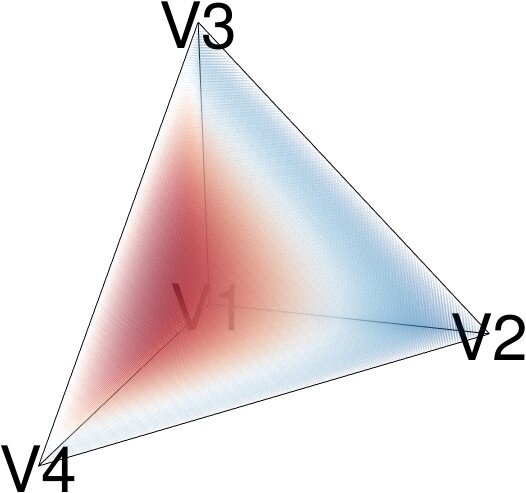}
	\!\!\!\!\!&\!\!\!\!\!\includegraphics[scale=0.13]{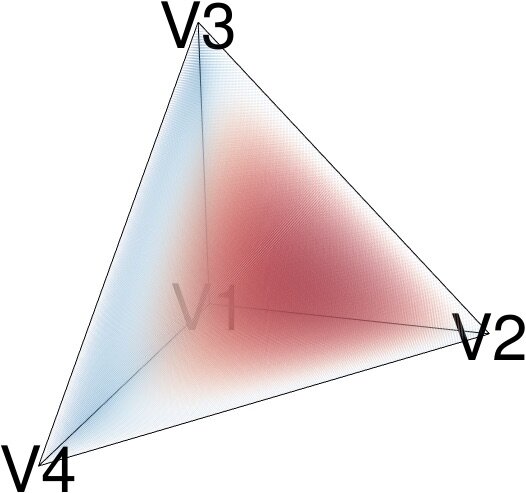}
	\!\!\!\!\!&\!\!\!\!\!\includegraphics[scale=0.13]{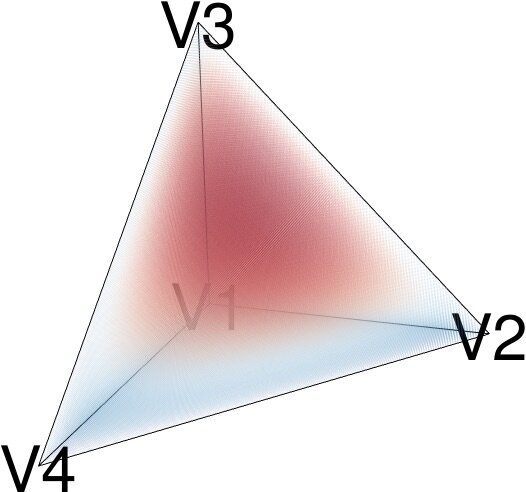}
	\!\!\!\!\!&\!\!\!\!\!\includegraphics[scale=0.13]{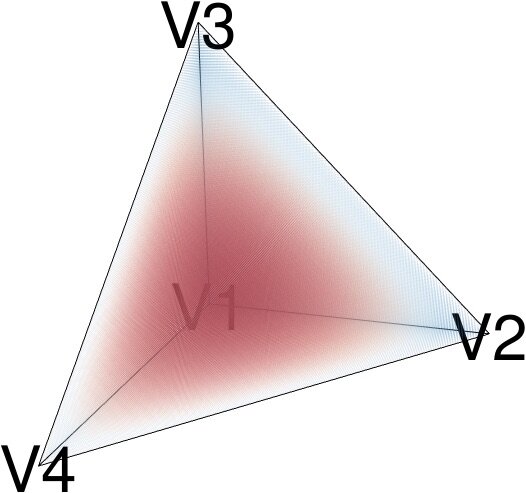}\\
	{\footnotesize $B_{1102}^{4, T}=12b_1b_2b_4^2$} 
	\!\!\!\!&\!\!\!\! {\footnotesize $B_{1021}^{4, T}=12b_1b_3^2b_4$} 
	\!\!\!\!&\!\!\!\! {\footnotesize $B_{1012}^{4, T}=12b_1b_3b_4^2$} 
	\!\!\!\!&\!\!\!\! {\footnotesize $B_{0211}^{4, T}=12b_2^2b_3b_4$} 
	\!\!\!\!&\!\!\!\! {\footnotesize $B_{0121}^{4, T}=12b_2b_3^2b_4$} 
	\!\!\!\!&\!\!\!\! {\footnotesize $B_{0112}^{4, T}=12b_2b_3b_4^2$}\\[3pt]
\end{tabular}} 
\end{center}
\caption{An illustration of Bernstein basis functions $\{B_{ijkl}^{4,T}(\bs{p})\}$. 
\label{Fig:Bd3}}
\end{figure}

Thus, given Bernstein basis functions $\{B_{ijkl}^{d, T}(\bs{p})\}_{i+j+k+l=d}$, any polynomial $\phi(\bs{p}) \in \mathcal{P}_{d}(T)$  can be written uniquely as \textsf{B-form}, 
\begin{equation}
\label{EQN:Bform}
    \phi(\bs{p}) |_{T}=\sum_{i+j+k+l=d} \gamma_{T;ijkl}B_{ijkl}^{d, T}(\bs{p})
    =\bs{B}^d_T(\bs{p})^{\top}\bs{\gamma}_{T},
\end{equation}
where the coefficients $\bs{\gamma}_{T}=\{\gamma_{T;ijkl}\}_{i+j+k+l=d}$ are called \textsf{B-coefficients} of $\phi$.
For the purpose of computer implementation, in this paper, we employ the lexicographical order for ordering of the coefficients $\bs{\gamma}_{T}$. To be specific, $\gamma_{T;ijkl}$ orders ahead of $\gamma_{T;i^{\prime}j^{\prime}k^{\prime}l^{\prime}}$ either
(i) $i > i^{\prime}$, or (ii) $i = i^{\prime}$ and $j > j^{\prime}$, or (iii) $i = i^{\prime}$, $j = j^{\prime}$ and $k > k^{\prime}$, or (iv) $i = i^{\prime}$, $j = j^{\prime}$, $k = k^{\prime}$ and $l > l^{\prime}$.
Consequently, we can express
\[
    \bs{B}^d_T(\bs{p}) = \left(B_{d,0,0,0}^{d, T}(\bs{p}), B_{d-1,1,0,0}^{d, T}(\bs{p}), B_{d-1,0,1,0}^{d, T}(\bs{p}),
    \ldots, B_{0,0,1,d-1}^{d, T}(\bs{p}), B_{0,0,0,d}^{d, T}(\bs{p})\right)^{\top}
\]
and
\begin{equation}
\label{EQN:gammaT}
    \bs{\gamma}_T = \left(\gamma_{T;d,0,0,0}, \gamma_{T;d-1,1,0,0}, \gamma_{T;d-1,0,1,0},
    \ldots, \gamma_{T;0,0,1,d-1}, \gamma_{T;0,0,0,d}\right)^{\top}.
\end{equation}
Accordingly, in (\ref{EQN:gammaT}), the index of the element $\gamma_{T;ijkl}$ in the vector $\bs{\gamma}_T$ is: 
\[
 	\sum_{m=0}^{d-i}\frac{(m+1)m}{2} + \sum_{n=0}^{d-i-j}(n+1) - k.
\]
Note that using a different ordering method will not affect the evaluation results for the trivariate polynomial functions.

It is convenient to derive conditions of continuous connection for polynomials defined on adjacent tetrahedra in using barycentric coordinates and Bernstein basis polynomials. 
We use the following Example \ref{EXM:TT} for illustration.
\begin{figure}[htbp!]
\begin{center}
\includegraphics[scale=.1]{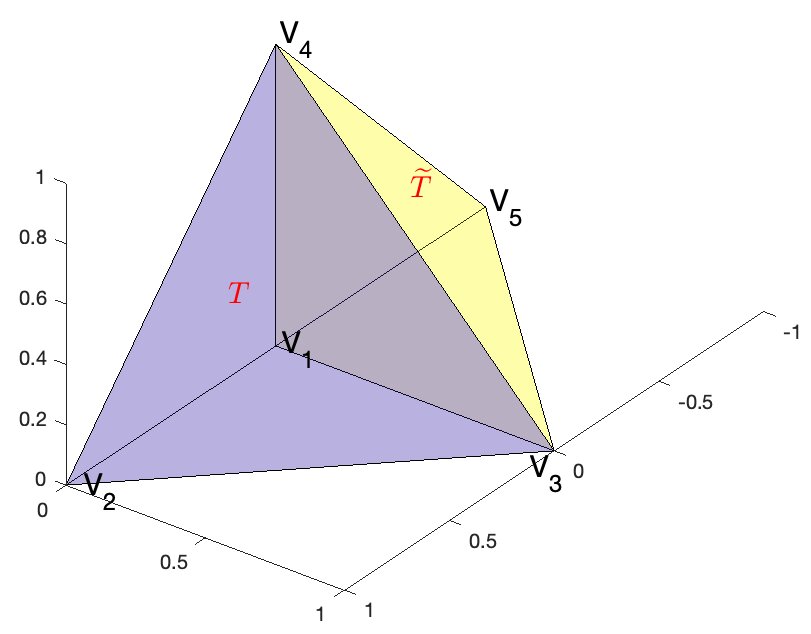}
\end{center}
\caption{An expample of two tetrahedra sharing a common face.}
\label{Fig:Tri_shape}
\end{figure}

\begin{example}
\label{EXM:TT}
Two adjacent tetrahedra $T=\langle \bs{v}_2, \bs{v}_1, \bs{v}_3, \bs{v}_4\rangle$ and $\widetilde{T}=\langle \bs{v}_5, \bs{v}_1, \bs{v}_4, \bs{v}_3\rangle$ share a common triangular face $F=\langle\bs{v}_1, \bs{v}_3, \bs{v}_4\rangle$, as illustrated in Figure \ref{Fig:Tri_shape}, where the Cartesian coordinates of the five vertices are $\bs{v}_1=(0,0,0)$, $\bs{v}_2=(1,0,0)$, $\bs{v}_3=(0,1,0)$, $\bs{v}_4=(0,0,1)$, $\bs{v}_5=(-1,0,0)$, respectively.
\end{example}

Assume two sets of Bernstein polynomial basis $\{B_{ijkl}^{d, T}(\bs{p})\}_{i+j+k+l=d}$ and $\{\widetilde{B}_{ijkl}^{d, \widetilde{T}}(\widetilde{\bs{p}})\}_{i+j+k+l=d}$ are defined on $T$ and $\widetilde{T}$ using the barycentric coordinates, respectively. Consider two degree-$d$ polynomials $\phi(\bs{p})$ and $\widetilde{\phi}(\widetilde{\bs{p}})$ defined on $T$ and $\widetilde{T}$, respectively, with B-forms
\begin{align*}
	\phi(\bs{p}) = \sum_{i+j+k+l=d}\gamma_{ijkl} B_{ijkl}^{d,T}(\bs{p}), ~~~~
	\widetilde{\phi}(\widetilde{\bs{p}}) = \sum_{i+j+k+l=d}\widetilde{\gamma}_{ijkl} \widetilde{B}_{ijkl}^{d, \widetilde{T}}(\widetilde{\bs{p}}).
\end{align*}
For point $\bs{p}_F\in F$, the barycentric coordinates with respect to $T$ and $\widetilde{T}$ are $(0, b_1, b_3, 1-b_1-b_3)$ and $(0, b_1, 1-b_1-b_3, b_3)$, respectively. Accordingly, we have
\begin{align}\nonumber
	\phi(\bs{p}_F) =& \sum_{j+k+l=d}\gamma_{0jkl} \frac{d!}{j!k!l!}b_1^j b_3^k (1-b_1-b_3)^l, \\
	\widetilde{\phi}(\bs{p}_F) =& \sum_{j+k+l=d}\widetilde{\gamma}_{0jkl} \frac{d!}{j!k!l!}b_1^j (1-b_1-b_3)^k b_3^l.
	\label{EQ:H_r0}
\end{align}
Therefore, $\phi$ and $\widetilde{\phi}$ are continuous on $F$ if and only if
\begin{equation}
    \gamma_{0jkl} = \widetilde{\gamma}_{0jlk}
\label{EQN:c0}
\end{equation}
for $j,k,l\geq0$ and $j+k+l=d$.

\vskip .10in \noindent \textbf{A.2. Directional derivatives and smoothness} \vskip .10in
\label{SUBSEC:deriv}

To generalize the smoothness restriction over the joint triangular face for two adjacent tetrahedra, we need to introduce the definitions of directional derivative first. Recall that for a general multivariate smooth function $\phi$, the directional derivative at point $\bs{p}$ with respect to direction $\bs{u}$ is defined as
\[
D_{\bs{u}}\phi(\bs{p}) := \frac{\partial}{\partial t}\phi(\bs{p} + t\bs{u}) \Big|_{t = 0} = \lim_{t\rightarrow0}\frac{\phi(\bs{p} + t\bs{u}) - \phi(\bs{p})}{t}.
\]
Accordingly, for vector $\bs{u}:=(u_x, u_y, u_z)\in\mathbb{R}^3$ and trivariate function $\phi$, the directional derivative at $\bs{p}=(x,y,z)$ is
\[
D_{\bs{u}}\phi(x,y,z) := \frac{\partial}{\partial t}\phi(x + tu_x, y + tu_y, z + tu_z) \Big|_{t = 0}.
\]
\begin{remark}
If $\phi$ is a polynomial of degree $d$, by calculus, $D_{\bs{u}}\phi(x,y,z) = u_x D_x \phi(x,y,z) + u_y D_y \phi(x,y,z) + u_z D_z \phi(x,y,z)$, so $D_{\bs{u}}\phi$ is a polynomial of degree $d - 1$.
\end{remark}

Consider direction $\bs{u} = \bs{p}_1 - \bs{p}_2$, where for $\jmath = 1,2$, $\bs{p}_{\jmath}\in\mathbb{R}^3$ have barycentric coordinates $(b_{\jmath1}, b_{\jmath2}, b_{\jmath3}, b_{\jmath4})$ with respect to $T$. Then $\bs{u}$ is uniquely described by the \textsf{directional coordinates} $\bs{a}=(a_1, a_2, a_3, a_4)=(b_{11}-b_{21}, b_{12}-b_{22}, b_{13}-b_{23}, b_{14}-b_{24})$. Obviously, $a_1+a_2+a_3+a_4=0$. Direct calculation gives the directional derivative of the Bernstein basis polynomial $B_{ijkl}^d$. 

\begin{lemma}
\citep[Lemma 15.12 in][]{Lai:Schumaker:07}
\label{LEM:ddb}
Consider direction $\bs{u}$ with directional coordinates $\bs{a}=(a_1, a_2, a_3, a_4)$. Then

\begin{align}
\label{EQN:ddb}
	D_{\bs{u}}B_{ijkl}^d(\bs{p}) 
	=& d\left\{a_1 B_{i-1,j,k,l}^{d-1}(\bs{p}) + a_2 B_{i,j-1,k,l}^{d-1}(\bs{p}) + a_3 B_{i,j,k-1,l}^{d-1}(\bs{p}) + a_4 B_{i,j,k,l-1}^{d-1}(\bs{p})\right\},
\end{align}
for any $\bs{p}\in T$ and $i+j+k+l=d$.
\end{lemma}

Consequently, one can obtain the directional derivative for any trivariate polynomial $\phi$.
\begin{theorem} 
\cite[Theorems 15.13 and 15.14 in][]{Lai:Schumaker:07}
\label{THM:Dup}
Consider direction $\bs{u}$ with directional coordinates $\bs{a}=(a_1, a_2, a_3, a_4)$. Then for any trivariate polynomial $\phi$
with B-form in (\ref{EQN:Bform}), the directional derivative is 
\begin{equation}
\label{EQU:Dup}
	D_{\bs{u}}\phi(\bs{p}) 
	= d \sum_{i+j+k+l=d-1} \gamma_{ijkl}^{(1)}(\bs{a}) B_{ijkl}^{d-1}(\bs{p}),
\end{equation}
where $\gamma_{ijkl}^{(1)}(\bs{a})=a_1\gamma_{i+1,j,k,l}+a_2\gamma_{i,j+1,k,l}+a_3\gamma_{i,j,k+1,l}+a_4\gamma_{i,j,k,l+1}$. 

In general, given $\bs{u}_1,\ldots,\bs{u}_m$ with associated directional coordinates $\bs{a}^{(\imath)}=(a_1^{(\imath)}, a_2^{(\imath)}, a_3^{(\imath)}, a_4^{(\imath)})$, $\imath = 1,\ldots,m$,
\begin{equation}
\label{EQU:derivative}
D_{\bs{u}_m}\cdots D_{\bs{u}_1}\phi(\bs{p}) 
	= \frac{d!}{(d-m)!} \sum_{i+j+k+l=d-m} \gamma_{ijkl}^{(m)}(\bs{a}^{(1)},\ldots,\bs{a}^{(m)}) B_{ijkl}^{d-m}(\bs{p}),
\end{equation}
where the coefficients are defined recursively as follows:
\begin{align*}
    \gamma_{ijkl}^{(m)}(\bs{a}^{(1)},\ldots,\bs{a}^{(m)})
    =&a_1^{(m)}\gamma_{i+1,j,k,l}^{(m-1)}(\bs{a}^{(1)},\ldots,\bs{a}^{(m-1)})
        +a_2^{(m)}\gamma_{i,j+1,k,l}^{(m-1)}(\bs{a}^{(1)},\ldots,\bs{a}^{(m-1)}) \\ 
    &+a_3^{(m)}\gamma_{i,j,k+1,l}^{(m-1)}(\bs{a}^{(1)},\ldots,\bs{a}^{(m-1)})
        +a_4^{(m)}\gamma_{i,j,k,l+1}^{(m-1)}(\bs{a}^{(1)},\ldots,\bs{a}^{(m-1)}),
\end{align*}
for $m=1,\ldots,d$, with $\gamma_{ijkl}^{(0)}(\bs{a})=\gamma_{ijkl}$.
\end{theorem}

In preparation for the discussion of trivariate spline and spline spaces, we need the following conditions for a smooth join between two polynomial on adjoining tetrahedra, like Example \ref{EXM:TT}
illustrated in Figure \ref{Fig:Tri_shape}.

\begin{theorem}
\label{THM:smoothjoin}
Suppose $\{\gamma_{ijkl}\}$ and $\{\widetilde{\gamma}_{ijkl}\}$ are B-coefficients of $\phi$ and $\widetilde{\phi}$ relative to two tetrahedra $T=\langle \bs{v}_2, \bs{v}_1, \bs{v}_3, \bs{v}_4\rangle$ and $\widetilde{T}=\langle \bs{v}_5, \bs{v}_1, \bs{v}_4, \bs{v}_3\rangle$, respectively, where $T$ and $\widetilde{T}$ share a common face $F=\langle\bs{v}_1, \bs{v}_3, \bs{v}_4\rangle$.
Then the following statements are equivalent:
\begin{enumerate}
\item $\phi$ and $\widetilde{\phi}$ join together with $\mathcal{C}^r$ continuity across the face $F$;
	\item For all $\bs{p}_F \in F$, $m=0,\ldots,r$ and for all directions $\bs{u}$,
	\[
		D_{\bs{u}}^m \phi(\bs{p}_F)=D_{\bs{u}}^m \widetilde{\phi}(\bs{p}_F);
	\]
\item For $\phi$ and $\widetilde{\phi}$ with B-forms in (\ref{EQ:H_r0}), for $i+j+k=d-m$, $m=0,\ldots,r$,
	\begin{equation}
		\widetilde{\gamma}_{mijk}
		=\sum_{\nu+\mu+\kappa+\delta=m}\gamma_{\nu,i+\mu,k+\kappa,j+\delta}
			B_{\nu\mu\kappa\delta}^m(\bs{v}_5).
	\label{EQ:H_r}
	\end{equation}
\end{enumerate}
\end{theorem}

Consider the case of $d=2$ piecewise polynomial. There are in total 10 Bernstein basis polynomials with coefficients 
\begin{align}
\{\gamma_{ijkl}\}&=(\gamma_{2000}, \gamma_{1100}, \gamma_{1010}, \gamma_{1001},\gamma_{0200}, \gamma_{0110}, \gamma_{0101}, \gamma_{0020}, \gamma_{0011}, \gamma_{0002})^{\top}, \nonumber\\ 
\{\widetilde{\gamma}_{ijkl}\}&=(\widetilde{\gamma}_{2000}, \widetilde{\gamma}_{1100}, \widetilde{\gamma}_{1010}, \widetilde{\gamma}_{1001},\widetilde{\gamma}_{0200}, \widetilde{\gamma}_{0110}, \widetilde{\gamma}_{0101}, \widetilde{\gamma}_{0020}, \widetilde{\gamma}_{0011}, \widetilde{\gamma}_{0002})^{\top}.\label{EQ:gamma}
\end{align}
Note that the barycentric coordinate of $\bs{v}_5$ with respect to $T$ is $(2, -1, 0, 0)$. 

If the trivariate polynomial is continuous over the whole region, then applying (\ref{EQ:H_r}) for $m=r=0$ generates 
\begin{equation}
    \widetilde{\gamma}_{0ijk}=\gamma_{0ikj}B_{0000}^0(\bs{v}_5)=\gamma_{0ikj},~~~~i+j+k=d,
\label{EQN:c0dd}
\end{equation}
which matches the conclusion in (\ref{EQN:c0}). 

Based on the de Casteljau Algorithm (Theorem \ref{THM:deCasteljau} in Appendix Section A.4.2), we propose the following computationally efficient algorithm in Algorithm \ref{ALGO:deCas} to calculate the derivatives of $\phi(\bs{p})$.
\normalem{
\begin{algorithm}
\footnotesize

\SetKwBlock{Begin}{}{}
\caption{Algorithm for the derivatives of $\phi(\bs{p})$.}

\textbf{Inputs:} Polynomial with B-form $\phi(\bs{p}) = \sum_{i+j+k+l=d}\gamma_{ijkl}B_{ijkl}^d(\bs{p})$, directions $\bs{u}_1, \ldots, \bs{u}_m$ with associated directional coordinates $\bs{a}^{(\imath)} = \left(a_1^{(\imath)}, a_2^{(\imath)}, a_3^{(\imath)}, a_4^{(\imath)}\right), \imath = 1, \ldots, m$

\textbf{Initialize:} $\imath:=0$, $\gamma_{ijkl}^{(0)}:=\gamma_{ijkl}$

\textbf{Outputs:} $D_{\bs{u}_m}\ldots D_{\bs{u}_1}\phi(\bs{p})$

\Begin
{
\For {$\imath = 1,\ldots,m$ }{
    \For{$i+j+k+l = d-\imath$}{
    \begin{align}
        \gamma_{ijkl}^{(\imath)}&(\bs{a}^{(1)},\ldots,\bs{a}^{(m)}) \nonumber\\
    =&a_1^{(\imath)}\gamma_{i+1,j,k,l}^{(\imath-1)}(\bs{a}^{(1)},\ldots,\bs{a}^{(\imath-1)})
        +a_2^{(\imath)}\gamma_{i,j+1,k,l}^{(\imath-1)}(\bs{a}^{(1)},\ldots,\bs{a}^{(\imath-1)}) \nonumber\\ 
    &+a_3^{(\imath)}\gamma_{i,j,k+1,l}^{(\imath-1)}(\bs{a}^{(1)},\ldots,\bs{a}^{(\imath-1)})
        +a_4^{(\imath)}\gamma_{i,j,k,l+1}^{(\imath-1)}(\bs{a}^{(1)},\ldots,\bs{a}^{(\imath-1)}), \nonumber
    \end{align}
    }
}
$D_{\bs{u}_m}\cdots D_{\bs{u}_1}\phi(\bs{p}) 
	= \frac{d!}{(d-m)!} \sum_{i+j+k+l=d-m} \gamma_{ijkl}^{(m)}(\bs{a}^{(1)},\ldots,\bs{a}^{(m)}) B_{ijkl}^{d-m}(\bs{p})$
}
\label{ALGO:deCas}
\end{algorithm}
\ULforem}

\vskip .10in \noindent \textbf{A.3. Example of the constraint matrix} \vskip .10in
\label{SUBSEC:egH}

We illustrate the construction of the constraint matrix $\mathbf{H}$ using the Example \ref{EXM:TT} presented in Figure \ref{Fig:Tri_shape}.

Continue our discussions in Section A.2 and consider the case of $d=2$ piecewise polynomial. Recall that there are in total 10 Bernstein basis polynomials, with coefficients $\{\gamma_{ijkl}\}$ and $\{\widetilde{\gamma}_{ijkl}\}$ in (\ref{EQ:gamma}),
and the barycentric coordinate of $\bs{v}_5$ with respect to $T$ is $(2, -1, 0, 0)$. 

As shown in (\ref{EQN:c0dd}), if the trivariate polynomial is continuous over the whole region, then $\widetilde{\gamma}_{0ijk}=\gamma_{0ikj}$ for $i+j+k=d$. In this case, we can write the constraint matrix $\mathbf{H}$ as

\begin{equation}
	\mathbf{H}=\left(
	\begin{array}{cccccccccccccccccccc}
		0 \!&\! 0 \!&\! 0 \!&\! 0 \!&\! 1 \!&\! 0 \!&\! 0 \!&\! 0 \!&\! 0 \!&\! 0 \!&\! 0 \!&\! 0 \!&\! 0 \!&\! 0 \!&\! -1 \!&\! 0 \!&\! 0 \!&\! 0 \!&\! 0 \!&\! 0 \\
		0 \!&\! 0 \!&\! 0 \!&\! 0 \!&\! 0 \!&\! 1 \!&\! 0 \!&\! 0 \!&\! 0 \!&\! 0 \!&\! 0 \!&\! 0 \!&\! 0 \!&\! 0 \!&\! 0 \!&\! 0 \!&\! -1 \!&\! 0 \!&\! 0 \!&\! 0 \\
		0 \!&\! 0 \!&\! 0 \!&\! 0 \!&\! 0 \!&\! 0 \!&\! 1 \!&\! 0 \!&\! 0 \!&\! 0 \!&\! 0 \!&\! 0 \!&\! 0 \!&\! 0 \!&\! 0 \!&\! -1 \!&\! 0 \!&\! 0 \!&\! 0 \!&\! 0 \\
		0 \!&\! 0 \!&\! 0 \!&\! 0 \!&\! 0 \!&\! 0 \!&\! 0 \!&\! 1 \!&\! 0 \!&\! 0 \!&\! 0 \!&\! 0 \!&\! 0 \!&\! 0 \!&\! 0 \!&\! 0 \!&\! 0 \!&\! 0 \!&\! 0 \!&\! -1 \\
		0 \!&\! 0 \!&\! 0 \!&\! 0 \!&\! 0 \!&\! 0 \!&\! 0 \!&\! 0 \!&\! 1 \!&\! 0 \!&\! 0 \!&\! 0 \!&\! 0 \!&\! 0 \!&\! 0 \!&\! 0 \!&\! 0 \!&\! 0 \!&\! -1 \!&\! 0 \\
		0 \!&\! 0 \!&\! 0 \!&\! 0 \!&\! 0 \!&\! 0 \!&\! 0 \!&\! 0 \!&\! 0 \!&\! 1 \!&\! 0 \!&\! 0 \!&\! 0 \!&\! 0 \!&\! 0 \!&\! 0 \!&\! 0 \!&\! -1 \!&\! 0 \!&\! 0 \\
	\end{array}\right). \label{EQ:H_r0_mtx}
\end{equation}

Furthermore, if the trivariate polynomial has continuous first derivatives over the whole region, then (\ref{EQ:H_r}) holds for both $m=r=0$ and $m=1$. Thus, in addition, we also need for any non-negative integers $i$, $j$, $k$ such that $i+j+k=1$, 	
\begin{align}
	\widetilde{\gamma}_{1ijk}
	=&\gamma_{1ikj}B_{1000}^1(\bs{v}_5)+\gamma_{0,i+1,j}B_{0100}^1(\bs{v}_5) +\gamma_{0,i,k+1,j}B_{0010}^1(\bs{v}_5)+\gamma_{0,i,k,j+1}B_{0001}^1(\bs{v}_5) \nonumber \\
	=&2\gamma_{1ikj}-\gamma_{0,i+1,k,j}.
\label{EQN:c1}
\end{align}
Therefore, for $r = 1$ where the trivariate spline has continuous first derivatives over the whole region, we will obtain the following $\mathbf{H}$
\[
	\mathbf{H}=\left(
	\begin{array}{cccccccccccccccccccc}
		0 \!&\! 0 \!&\! 0 \!&\! 0 \!&\! 1 \!&\! 0 \!&\! 0 \!&\! 0 \!&\! 0 \!&\! 0 \!&\! 0 \!&\! 0 \!&\! 0 \!&\! 0 \!&\!\! -1 \!\!&\! 0 \!&\! 0 \!&\! 0 \!&\! 0 \!&\! 0 \\
		0 \!&\! 0 \!&\! 0 \!&\! 0 \!&\! 0 \!&\! 1 \!&\! 0 \!&\! 0 \!&\! 0 \!&\! 0 \!&\! 0 \!&\! 0 \!&\! 0 \!&\! 0 \!&\! 0 \!&\! 0 \!&\!\! -1 \!\!&\! 0 \!&\! 0 \!&\! 0 \\
		0 \!&\! 0 \!&\! 0 \!&\! 0 \!&\! 0 \!&\! 0 \!&\! 1 \!&\! 0 \!&\! 0 \!&\! 0 \!&\! 0 \!&\! 0 \!&\! 0 \!&\! 0 \!&\! 0 \!&\!\! -1 \!\!&\! 0 \!&\! 0 \!&\! 0 \!&\! 0 \\
		0 \!&\! 0 \!&\! 0 \!&\! 0 \!&\! 0 \!&\! 0 \!&\! 0 \!&\! 1 \!&\! 0 \!&\! 0 \!&\! 0 \!&\! 0 \!&\! 0 \!&\! 0 \!&\! 0 \!&\! 0 \!&\! 0 \!&\! 0 \!&\! 0 \!&\! -1 \\
		0 \!&\! 0 \!&\! 0 \!&\! 0 \!&\! 0 \!&\! 0 \!&\! 0 \!&\! 0 \!&\! 1 \!&\! 0 \!&\! 0 \!&\! 0 \!&\! 0 \!&\! 0 \!&\! 0 \!&\! 0 \!&\! 0 \!&\! 0 \!&\!\! -1 \!\!&\! 0 \\
		0 \!&\! 0 \!&\! 0 \!&\! 0 \!&\! 0 \!&\! 0 \!&\! 0 \!&\! 0 \!&\! 0 \!&\! 1 \!&\! 0 \!&\! 0 \!&\! 0 \!&\! 0 \!&\! 0 \!&\! 0 \!&\! 0 \!&\!\! -1 \!\!&\! 0 \!&\! 0 \\
		0 \!&\! 2 \!&\! 0 \!&\! 0 \!&\!\! -1 \!\!&\! 0 \!&\! 0 \!&\! 0 \!&\! 0 \!&\! 0 \!&\! 0 \!&\!\! -1 \!\!&\! 0 \!&\! 0 \!&\! 0 \!&\! 0 \!&\! 0 \!&\! 0 \!&\! 0 \!&\! 0 \\
		0 \!&\! 0 \!&\! 2 \!&\! 0 \!&\! 0 \!&\!\! -1 \!\!&\! 0 \!&\! 0 \!&\! 0 \!&\! 0 \!&\! 0 \!&\! 0 \!&\!\! -1 \!\!&\! 0 \!&\! 0 \!&\! 0 \!&\! 0 \!&\! 0 \!&\! 0 \!&\! 0 \\
		0 \!&\! 0 \!&\! 0 \!&\! 2 \!&\! 0 \!&\! 0 \!&\!\! -1 \!\!&\! 0 \!&\! 0 \!&\! 0 \!&\! 0 \!&\! 0 \!&\! 0 \!&\!\! -1 \!\!&\! 0 \!&\! 0 \!&\! 0 \!&\! 0 \!&\! 0 \!&\! 0 \\
	\end{array}\right).
\]
Here, the first six rows match the form in (\ref{EQ:H_r0_mtx}), and the last three rows correspond to (\ref{EQN:c1}).

\vskip .10in \noindent \textbf{A.4. Proof of Theorems in Section A.2} \vskip .10in
\label{SUBSEC:proof}

\vskip .10in \noindent \textsc{A.4.1. Theoretical properties of Bernstein bases} \vskip .10in
\label{SUBSUBSEC:proBern}

\begin{lemma}
\label{LEM:basis}
For the Bernstein basis functions $\{B_{ijkl}^{d, T}(\bs{p})\}_{i+j+k+l=d}$, each $B_{ijkl}^{d, T}$ has a unique maximum at the point $d^{-1}(i\bs{v}_1+j\bs{v}_2+k\bs{v}_3+l\bs{v}_4)$.
\end{lemma}

To show Lemma \ref{LEM:basis}, we first need to show Lemma \ref{LEM:ddb}.

\begin{proof}[Proof of Lemma \ref{LEM:ddb}]
    Suppose $\bs{b}=(b_1,b_2,b_3,b_4)$ are the barycentric coordinates of $\bs{p}$. Then the barycentric coordinates of $\bs{p}+t\bs{u}$ are $ (b_1 + ta_1,b_2 + ta_2,b_3 + ta_3,b_4 + ta_4)$. Thus, for $i+j+k+l=d$,
    \[
        B_{ijkl}^d(\bs{p}+t\bs{u}) 
	    = \frac{d!}{i!j!k!l!}(b_1 + ta_1)^i (b_2 + ta_2)^j (b_3 + ta_3)^k (b_4 + ta_4)^l.
    \]
    Hence, 
    \begin{align*}
        &D_{\bs{u}}B_{ijkl}^d(\bs{p}) 
    	= \frac{\partial}{\partial t}\phi(\bs{p}+t\bs{u}) \Big|_{t=0} \\
    	=& \frac{d!}{i!j!k!l!}\left(ia_1b_1^{i-1}b_2^jb_3^kb_4^l + ja_2b_1^ib_2^{j-1}b_3^kb_4^l + ka_3b_1^ib_2^jb_3^{k-1}b_4^l + la_4b_1^ib_2^jb_3^kb_4^{l-1}\right) \\
    	=& d\left\{a_1 B_{i-1,j,k,l}^{d-1}(\bs{p}) + a_2 B_{i,j-1,k,l}^{d-1}(\bs{p}) + a_3 B_{i,j,k-1,l}^{d-1}(\bs{p}) + a_4 B_{i,j,k,l-1}^{d-1}(\bs{p})\right\}. 
    \end{align*}
\end{proof}

\begin{proof}[Proof of Lemma \ref{LEM:basis}]
    For $T = \langle\bs{v}_1, \bs{v}_2, \bs{v}_3, \bs{v}_4\rangle$, the barycentric coordinates for $\bs{v}_1$, $\bs{v}_2$, $\bs{v}_3$ and $\bs{v}_4$ are $(1,0,0,0)$, $(0,1,0,0)$, $(0,0,1,0)$ and $(0,0,0,1)$, respectively. 
    Accordingly, the directional coordinates of $\bs{u}_1=\bs{v}_1-\bs{v}_2$, $\bs{u}_2=\bs{v}_1-\bs{v}_3$ and $\bs{u}_3=\bs{v}_1-\bs{v}_4$ are $(1,-1,0,0)$, $(1,0,-1,0)$ and $(1,0,0,-1)$, respectively.
    For $\bs{p}\in T$ with barycentric coordinates $\bs{b}=(b_1,b_2,b_3,b_4)$, consider derivatives of $B_{ijkl}^d(\bs{p})$ with respect to directions $\bs{u}_1$, $\bs{u}_2$ and $\bs{u}_3$, then
    \begin{align*}
        D_{\bs{u}_1}B_{ijkl}^d(\bs{p})
        &=B_{ijkl}^d(\bs{p}) (ib_1^{-1} - jb_2^{-1}), ~~~~
        D_{\bs{u}_2}B_{ijkl}^d(\bs{p})
        =B_{ijkl}^d(\bs{p}) (ib_1^{-1} - kb_3^{-1}), \\
        D_{\bs{u}_3}B_{ijkl}^d(\bs{p})
        &=B_{ijkl}^d(\bs{p}) (ib_1^{-1} - lb_4^{-1}).
    \end{align*}
    Setting these equations to zero and combining with $b_1+b_2+b_3+b_4=1$ gives $(b_1,b_2,b_3,b_4)=d^{-1}(i,j,k,l)$.
\end{proof}
    
\vskip .10in \noindent \textsc{A.4.2. Proof of Theorems \ref{THM:Dup} and \ref{THM:smoothjoin}} \vskip .10in
\label{SUBSUBSEC:proofC}

\begin{proof}[Proof of Theorem \ref{THM:Dup}]
For $\phi(\bs{p})=\sum_{i+j+k+l=d} \gamma_{ijkl}B_{ijkl}^d(\bs{p})$, by Lemma \ref{LEM:ddb},
\begin{align*}
    &D_{\bs{u}}\phi(\bs{p})
    =\sum_{i+j+k+l=d} \gamma_{ijkl}D_{\bs{u}}B_{ijkl}^{d}(\bs{p}) \\
    =&\sum_{i+j+k+l=d} \gamma_{ijkl}d\left\{a_1 B_{i-1,j,k,l}^{d-1}(\bs{p}) + a_2 B_{i,j-1,k,l}^{d-1}(\bs{p}) + a_3 B_{i,j,k-1,l}^{d-1}(\bs{p}) + a_4 B_{i,j,k,l-1}^{d-1}(\bs{p})\right\} \\
    =&d \sum_{i+j+k+l=d-1} \left( a_1\gamma_{i+1,j,k,l}+a_2\gamma_{i,j+1,k,l}+a_3\gamma_{i,j,k+1,l}+a_4\gamma_{i,j,k,l+1}\right) B_{ijkl}^{d-1}(\bs{p}).
\end{align*}
Thus, (\ref{EQU:Dup}) follows. Consequently, one can obtain (\ref{EQU:derivative}) by repeatedly applying (\ref{EQU:Dup}) for directions $\bs{u}_1,\ldots,\bs{u}_m$. 
\end{proof}

To show Theorem \ref{THM:smoothjoin}, we need the following results.

\begin{theorem}
\citep[Theorems 15.10 in][]{Lai:Schumaker:07}
\label{THM:deCasteljau}
Suppose $\phi(\bs{p})$ is a trivariate polynomial with B-form $\phi(\bs{p})=\sum_{i+j+k+l=d} \gamma_{ijkl}B_{ijkl}^d(\bs{p})$. Define $\gamma_{ijkl}^{(0)}:=\gamma_{ijkl}$, $i+j+k+l=d$. Suppose $\bs{p}$ has barycentric coordinates $\bs{b}=(b_1,b_2,b_3,b_4)$. Then
\[
    \phi(\bs{p})=\sum_{i+j+k+l=d-m} \gamma_{ijkl}^{(m)}B_{ijkl}^{d-m}(\bs{p}),
\]
where for $m=1,\ldots,d$, $\gamma_{ijkl}^{(m)}$ are computed by the recursion 
\begin{equation}
\label{EQN:deCasteljau}
    \gamma_{ijkl}^{(m)}=b_1\gamma_{i+1,j,k,l}^{(m-1)} + b_2\gamma_{i,j+1,k,l}^{(m-1)} + b_3\gamma_{i,j,k+1,l}^{(m-1)} + b_4\gamma_{i,j,k,l+1}^{(m-1)},
\end{equation}
for $i+j+k+l=d-m$.
\end{theorem}
\begin{remark}
The recursive formula (\ref{EQN:deCasteljau}) is also referred as de$\,$Casteljau algorithm. See Section 15.6 in \cite{Lai:Schumaker:07} for more details.
\end{remark}

\begin{lemma}\label{LEM:gammam}
The coefficients in the recursive formula of de$\,$Casteljau algorithm (\ref{EQN:deCasteljau}) are given by 
\begin{equation}
\label{EQN:gammam}
    \gamma_{ijkl}^{(m)}=\sum_{i^{\prime}+j^{\prime}+k^{\prime}+l^{\prime}=m}\gamma_{i+i^{\prime},j+j^{\prime},k+k^{\prime},l+l^{\prime}}B_{i^{\prime}j^{\prime}k^{\prime}l^{\prime}}^{m}(\bs{p}),
\end{equation}
where $i+j+k+l=d-m$.
\end{lemma}
\begin{proof}
Define operators $E_1\gamma_{ijkl}=\gamma_{i+1,j,k,l}$, $E_2\gamma_{ijkl}=\gamma_{i,j+1,k,l}$, $E_3\gamma_{ijkl}=\gamma_{i,j,k+1,l}$ and $E_4\gamma_{ijkl}=\gamma_{i,j,k,l+1}$.
Thus, by (\ref{EQN:deCasteljau}), for $i+j+k+l=d-m$,
\begin{align*}
    \gamma_{ijkl}^{(m)}
    =&\left(b_1E_1+b_2E_2+b_3E_3+b_4E_4\right)\gamma_{ijkl}^{(m-1)}
    =\left(b_1E_1+b_2E_2+b_3E_3+b_4E_4\right)^m\gamma_{ijkl}\\
    =&\sum_{i^{\prime}+j^{\prime}+k^{\prime}+l^{\prime}=m} B_{i^{\prime}j^{\prime}k^{\prime}l^{\prime}}^{m}(\bs{p}) E_1^{i^{\prime}}E_2^{j^{\prime}}E_3^{k^{\prime}}E_4^{l^{\prime}}\gamma_{ijkl}
    = \sum_{i^{\prime}+j^{\prime}+k^{\prime}+l^{\prime}=m}\gamma_{i+i^{\prime},j+j^{\prime},k+k^{\prime},l+l^{\prime}}B_{i^{\prime}j^{\prime}k^{\prime}l^{\prime}}^{m}(\bs{p}).
\end{align*}
\end{proof}

\begin{lemma}
\label{LEM:gamman}
Suppose $\phi(\bs{p})$ is a trivariate polynomial with B-form $\phi(\bs{p})=\sum_{i+j+k+l=d} \gamma_{ijkl}B_{ijkl}^{d}(\bs{p})$. Then for any $1\leq n\leq d$, the $n$-th order directional derivative of $\phi$ with respect to the direction $\bs{u}=\bs{v}_4-\bs{v}_2$ is given by
\[
    D_{\bs{u}}^n\phi(\bs{p}) 
	= \frac{d!}{(d-n)!} \sum_{i+j+k+l=d-n} \gamma_{ijkl}^{(n)} B_{ijkl}^{d-n}(\bs{p})
\]
with
\[
    \gamma_{ijkl}^{(n)}=\sum_{m=0}^n\binom{n}{m}(-1)^m\gamma_{i+m,j,k,l+n-m}, ~~~~i+j+k+l=d-n.
\]
\end{lemma}
\begin{proof}
The directional coordinates of $\bs{u}=\bs{v}_4-\bs{v}_2$ is $\bs{a}=(-1,0,0,1)$. By (\ref{EQN:gammam}),
\begin{align*}
    \gamma_{ijkl}^{(n)}
    =&\sum_{i^{\prime}+j^{\prime}+k^{\prime}+l^{\prime}=n}\gamma_{i+i^{\prime},j+j^{\prime},k+k^{\prime},l+l^{\prime}}B_{i^{\prime}j^{\prime}k^{\prime}l^{\prime}}^{n}(\bs{p})
    =\sum_{i^{\prime}+l^{\prime}=n}\gamma_{i+i^{\prime},j,k,l+l^{\prime}}\frac{n!}{i^{\prime}!l^{\prime}!}(-1)^{i^{\prime}}\\
    =&\sum_{m=0}^n\binom{n}{m}(-1)^m\gamma_{i+m,j,k,l+n-m}, 
    ~~~~~~~~i+j+k+l=d-n.
\end{align*}
\end{proof}

\begin{proof}[Proof of Theorem \ref{THM:smoothjoin}]
The equivalence between Statements 1 and 2 is obvious by definition. Thus, we just show the equivalence between Statements 2 and 3.

i) We start to consider when $r=0$. It is equivalent to consider directions $\bs{\widetilde{u}}$ along $F$, where the directional coordinates are $(0,1-\widetilde{b}_3-\widetilde{b}_4,\widetilde{b}_3,\widetilde{b}_4)$ and $(0,1-\widetilde{b}_3-\widetilde{b}_4,\widetilde{b}_4,\widetilde{b}_3)$ with respect to the tetrahedron $T$ and $\widetilde{T}$, respectively. Thus, $\phi$ and $\widetilde{\phi}$ join continuously along $F$ if and only if
\[
    \sum_{j+k+l=d}\!\!\!\gamma_{0jkl}\frac{d!}{j!k!l!}(1-\widetilde{b}_3-\widetilde{b}_4)^j(\widetilde{b}_3)^k(\widetilde{b}_4)^l
    \!=\!\!\!\!\!\!\sum_{j+k+l=d}\!\!\!\widetilde{\gamma}_{0jkl}\frac{d!}{j!k!l!}(1-\widetilde{b}_3-\widetilde{b}_4)^j(\widetilde{b}_4)^k(\widetilde{b}_3)^l.
\]
That is, $\gamma_{0jkl}=\widetilde{\gamma}_{0jlk}$, with $j+k+l=d$, which matches our conclusion in (\ref{EQ:H_r}).

ii) Then we consider for $r>0$. First note that 
\begin{equation}
\label{EQN:Dp}
    D_{\bs{u}}^n\phi(\bs{p}_F) = D_{\bs{u}}^n\widetilde{\phi}(\bs{p}_F)
\end{equation}
holds for any $\bs{p}_F\in F$, $n=0,\ldots,r$, if and only if (\ref{EQN:Dp}) holds for the direction $\bs{u}=\bs{v}_5-\bs{v}_3$. It is a fact because by the argument in i), all derivatives of $\phi$ and $\widetilde{\phi}$ corresponding to the directions $\bs{v}_3-\bs{v}_1$ and $\bs{v}_4-\bs{v}_1$, agree at every point on $F$. And derivatives in all other directions can be written as linear combinations of $D_{\bs{u}}$, $D_{\bs{v}_3-\bs{v}_1}$ and $D_{\bs{v}_4-\bs{v}_1}$. 

Let $\bs{b}=(b_1,b_2,b_3,b_4)$ be the barycentric coordinates of $\bs{v}_5$ relative to the tetrahedron $T$. Correspondingly, the directional coordinates of $\bs{u}$ are $\bs{a}=(b_1,b_2,b_3-1,b_4)$ and $\bs{\widetilde{a}}=(1,0,0,-1)$ with respect to $T$ and $\widetilde{T}$, respectively.

By Theorem \ref{THM:Dup}, for each $0\leq n\leq r$,
\begin{align*}
    D_{\bs{u}}^n\phi(\bs{p}_F) |_F
	=& \frac{d!}{(d-n)!} \sum_{j+k+l=d-n} \gamma_{0jkl}^{(n)}(\bs{a}) B_{0jkl}^{d-n}(\bs{p}_F), \\
	D_{\bs{u}}^n\widetilde{\phi}(\bs{p}_F) |_F
	=& \frac{d!}{(d-n)!} \sum_{j+k+l=d-n} \gamma_{0jkl}^{(n)}(\bs{\widetilde{a}}) B_{0jkl}^{d-n}(\bs{p}_F).
\end{align*}
Since for points $\bs{p}_F\in F$, $\widetilde{B}_{0jkl}^{d-n}(\bs{p}_F)=B_{0jlk}^{d-n}(\bs{p}_F)$, it follows that (\ref{EQN:Dp}) holds if and only if for $j+k+l=d-n$, $n=0,\ldots,r$,
\begin{equation}\label{EQN:gamgam}
    \widetilde{\gamma}_{0jkl}^{(n)}(\bs{\widetilde{a}})
    =\gamma_{0jlk}^{(n)}(\bs{a}).
\end{equation}
By Lemma \ref{LEM:gamman}, for $j+k+l=d-n$, 
\begin{equation}\label{EQN:gamtil}
    \widetilde{\gamma}_{0jkl}^{(n)}(\bs{\widetilde{a}})
    =\sum_{m=0}^n(-1)^{n-m}\binom{n}{m}\widetilde{\gamma}_{m,j,k,d-m-j-k}.
\end{equation}
In another direction, following the proof of Lemma \ref{LEM:gammam}, for $j+l+k=d-n$,
\begin{align}
    \gamma_{0jlk}^{(n)}(\bs{a})
    =&\left\{b_1E_1+b_2E_2+(b_3-1)E_3+b_4E_4\right\}^n \gamma_{0jlk} \nonumber\\
    =&\left(b_1E_1+b_2E_2+b_3E_3+b_4E_4-E_3\right)^n \gamma_{0jlk} \nonumber\\
    =&\sum_{m=0}^n(-1)^{n-m}\binom{n}{m}\left(b_1E_1+b_2E_2+b_3E_3+b_4E_4\right)^n \gamma_{0,j,l+n-m,k} \nonumber\\
    =&\sum_{m=0}^n(-1)^{n-m}\binom{n}{m} \gamma_{0,j,d-j-k-m,k}^{(m)}(\bs{b}). \label{EQN:gamab}
\end{align}
Combining (\ref{EQN:gamtil}) and (\ref{EQN:gamab}), (\ref{EQN:gamgam}) holds if and only if for $j+k+l=d-n$ and $n=0,\ldots,r$, 
$\widetilde{\gamma}_{njkl}=\gamma_{0jlk}^{(n)}(\bs{b})$. By Lemma \ref{LEM:gammam}, 
\[
    \gamma_{0jlk}^{(n)}(\bs{b})
    =\sum_{i^{\prime}+j^{\prime}+k^{\prime}+l^{\prime}=n}\gamma_{i^{\prime},j+j^{\prime},k+k^{\prime},l+l^{\prime}}B_{i^{\prime}j^{\prime}k^{\prime}l^{\prime}}^{n}(\bs{v}_5).
\]
Thus, (\ref{EQN:Dp}) holds if and only if 
\[
    \widetilde{\gamma}_{njkl}
    =\sum_{i^{\prime}+j^{\prime}+k^{\prime}+l^{\prime}=n}\gamma_{i^{\prime},j+j^{\prime},k+k^{\prime},l+l^{\prime}}B_{i^{\prime}j^{\prime}k^{\prime}l^{\prime}}^{n}(\bs{v}_5),
\]
that is, (\ref{EQ:H_r}) follows. 
\end{proof}

\vskip .10in \noindent \textbf{A.5. Directional derivatives for basis functions} \vskip .10in
\label{SUBSEC:Cmatrix}

Based on the conclusion in (\ref{EQU:derivative}) in the Appendix A, the $m$th order directional derivatives for all the Bernstein basis functions with degree $d$ can be written as some linear combination of Bernstein basis functions with degree $d-m$. Specifically, for directions $\bs{u}_1, \ldots, \bs{u}_m$, there exists a $\binom{d+3}{3} \times \binom{d+3-m}{3}$ matrix $\mathbf{C}^{(m)}_d(\bs{u}_1, \ldots, \bs{u}_m)$ such that 
\[
D_{\bs{u}_m} \cdots D_{\bs{u}_1} \mathbf{B}_d(\bs{p})= \mathbf{C}^{(m)}_d(\bs{u}_1, \ldots, \bs{u}_m) \mathbf{B}_{d-m}(\bs{p}).
\]
Consequently, with $D_{\bs{u}_{m-1}} \cdots D_{\bs{u}_1} \mathbf{B}_d(\bs{p})= \mathbf{C}^{(m-1)}_d(\bs{u}_1, \ldots, \bs{u}_{m-1}) \mathbf{B}_{d-m+1}(\bs{p})$, we have
\begin{align*}
D_{\bs{u}_m} D_{\bs{u}_{m-1}} \cdots D_{\bs{u}_1} \mathbf{B}_d(\bs{p}) 
& = \mathbf{C}^{(m-1)}_d(\bs{u}_1, \ldots, \bs{u}_{m-1}) D_{\bs{u}_m} \mathbf{B}_{d-m+1}(\bs{p}) \\
& = \mathbf{C}^{(m-1)}_d(\bs{u}_1, \ldots, \bs{u}_{m-1})  \mathbf{C}^{(1)}_{d-m+1}(\bs{u}_{m}) \mathbf{B}_{d-m}(\bs{p}),
\end{align*}
which implies $\mathbf{C}^{(m)}_d(\bs{u}_1, \ldots, \bs{u}_m) = \mathbf{C}^{(m-1)}_d(\bs{u}_1, \ldots, \bs{u}_{m-1})\mathbf{C}^{(1)}_{d-m+1}(\bs{u}_m) $. Keep decomposing matrix $\mathbf{C}^{(m-1)}_d(\bs{u}_1, \ldots, \bs{u}_{m-1})$, we have $\mathbf{C}^{(m)}_d(\bs{u}_1, \ldots, \bs{u}_m) = \mathbf{C}^{(1)}_{d}(\bs{u}_1) \cdots  \mathbf{C}^{(1)}_{d-m+1}(\bs{u}_m)$.

Therefore, we can obtain the explicit form of the $m$th order directional derivatives for the Bernstein basis functions with degree $d$, given the explicit form of matrix $\mathbf{C}^{(1)}_{d}(\bs{u})$. Based on Lemma \ref{LEM:ddb}, 
\[
D_{\bs{u}}B_{ijkl}^d(\bs{p}) = d\left\{a_1 B_{i-1,j,k,l}^{d-1}(\bs{p}) + a_2 B_{i,j-1,k,l}^{d-1}(\bs{p}) + a_3 B_{i,j,k-1,l}^{d-1}(\bs{p}) + a_4 B_{i,j,k,l-1}^{d-1}(\bs{p})\right\},
\]
where $(a_1, a_2, a_3, a_4)$ is the barycentric coordinate of direction $\bs{u}$. Then, the $\{\sum_{m=0}^{d-i}(m+1)m/2 + \sum_{n=0}^{d-i-j}(n+1) - k\}$th row of matrix $\mathbf{C}^{(1)}_{d}(\bs{u})$ is $d(a_1 \bs{e}_{\mathcal{I}_1} + a_2 \bs{e}_{\mathcal{I}_2} + a_3 \bs{e}_{\mathcal{I}_3} + a_4 \bs{e}_{\mathcal{I}_4})$, where $\bs{e}_{\mathcal{I}}$ is unit vector with $\mathcal{I}$th element being one,  $\mathcal{I}_1$,  $\mathcal{I}_2$, $\mathcal{I}_3$ and $\mathcal{I}_4$ are indexes of basis functions $B_{i-1,j,k,l}^{d-1}(\bs{p})$, $B_{i,j-1,k,l}^{d-1}(\bs{p})$, $B_{i,j,k-1,l}^{d-1}(\bs{p})$, and $B_{i,j,k,l-1}^{d-1}(\bs{p})$. 

In the following, we provide a simple example to illustrate how to calculate the first and second order derivatives and therefore to construct the penalty matrix $\mathbf{P}_T$ in (\ref{EQN:EsTP}) in Section A.6. Since the penalty function is calculated through the second-order derivatives in terms of the three-dimensional Cartesian coordinate system $x=(1,0,0)$, $y=(0,1,0)$ and $z=(0,0,1)$, in the following example, we only consider directional derivatives for $x, y, z$.

Consider the basis functions on the tetrahedron $T$ in the Figure \ref{Fig:Tri} with $d=3$. The horizontal axis $x=(1,0,0)$ can be written as $x=\bs{v}_1 - \bs{v}_2$, therefore, the barycentric coordinates of $x$ is $(1,-1,0,0)$. Then, the first and second order derivatives of all the Bernstein basis functions $\mathbf{B}_3(\bs{p})$ are   $D^{(1,0,0)} \mathbf{B}_3(\bs{p})= \mathbf{C}^{(1)}_3(1,0,0) \mathbf{B}_{2}(\bs{p})$ and $D^{(2,0,0)} \mathbf{B}_3(\bs{p})= \mathbf{C}^{(2)}_3(x, x) \mathbf{B}_1(\bs{p}) = \mathbf{C}^{(1)}_3(x) \mathbf{C}^{(1)}_2(x)\mathbf{B}_1(\bs{p})$, 
where matrices $\mathbf{C}^{(1)}_2(x)$ and $\mathbf{C}^{(1)}_{3}(x)$ are
\[
\begingroup 
\setlength\arraycolsep{2pt}
\mathbf{C}^{(1)}_2(x) = \begin{pmatrix}
2 & 0  & 0 & 0\\
-2 & 2  & 0 & 0\\
0 & 0  & 2 & 0\\
0 & 0  & 0 & 2\\
0 & -2  & 0 & 0\\
0 & 0  & -2 & 0\\
0 & 0  & 0 & -2\\
0 & 0  & 0 & 0\\
0 & 0  & 0 & 0\\
0 & 0  & 0 & 0\\
\end{pmatrix}
\endgroup,
~
\begingroup 
\setlength\arraycolsep{2pt}
\mathbf{C}^{(1)}_3(x) = \begin{pmatrix}
3 & 0 & 0 & 0 & 0 & 0 & 0 & 0 & 0 & 0 \\
-3 & 3 & 0 & 0 & 0 & 0 & 0 & 0 & 0 & 0 \\
0 & 0 & 3 & 0 & 0 & 0 & 0 & 0 & 0 & 0 \\
0 & 0 & 0 & 3 & 0 & 0 & 0 & 0 & 0 & 0 \\
0 & -3 & 0 & 0 & 3 & 0 & 0 & 0 & 0 & 0 \\
0 & 0 & -3 & 0 & 0 & 3 & 0 & 0 & 0 & 0 \\
0 & 0 & 0 & -3 & 0 & 0 & 3 & 0 & 0 & 0 \\
0 & 0 & 0 & 0 & 0 & 0 & 0 & 3 & 0 & 0 \\
0 & 0 & 0 & 0 & 0 & 0 & 0 & 0 & 3 & 0 \\
0 & 0 & 0 & 0 & 0 & 0 & 0 & 0 & 0 & 3 \\
0 & 0 & 0 & 0 & -3 & 0 & 0 & 0 & 0 & 0 \\
0 & 0 & 0 & 0 & 0 & -3 & 0 & 0 & 0 & 0 \\
0 & 0 & 0 & 0 & 0 & 0 & -3 & 0 & 0 & 0 \\
0 & 0 & 0 & 0 & 0 & 0 & 0 & -3 & 0 & 0 \\
0 & 0 & 0 & 0 & 0 & 0 & 0 & 0 & -3 & 0 \\
0 & 0 & 0 & 0 & 0 & 0 & 0 & 0 & 0 & -3 \\
0 & 0 & 0 & 0 & 0 & 0 & 0 & 0 & 0 & 0 \\
0 & 0 & 0 & 0 & 0 & 0 & 0 & 0 & 0 & 0 \\
0 & 0 & 0 & 0 & 0 & 0 & 0 & 0 & 0 & 0 \\
0 & 0 & 0 & 0 & 0 & 0 & 0 & 0 & 0 & 0 \\
\end{pmatrix}.
\endgroup
\]

\vskip .10in \noindent \textbf{A.6. Details of constructing penalty matrix} \vskip .10in
\label{SUBSEC:penalty}

Next, we introduce the details of constructing penalty $\mathbf{P}$ matrix. Following the definition of penalty matrix $\mathbf{P}_T$ in the Section 3, we can further implement as
\begin{align}
\mathcal{E}(s_T)
&= \sum_{|\bs{\alpha}| = 2}{2 \choose \alpha_1}{2 - \alpha_1 \choose \alpha_2}\int_T \bigg\{\sum_{i + j + k + l = d} \gamma_{T; ijkl} D^{\bs{\alpha}} B_{ijkl}^{d,T}(\bs{p})\bigg\}^2 \mathrm{d}\bs{p} \nonumber \\
\notag
&= \sum_{|\bs{\alpha}| = 2} {2 \choose \alpha_1}{2 - \alpha_1 \choose \alpha_2} \sum_{i + j + k + l = d}\sum_{i^{\prime} + j^{\prime} + k^{\prime} + l^{\prime} = d}  \gamma_{T; ijkl} \gamma_{T; i^{\prime}j^{\prime}k^{\prime}l^{\prime}} \times \\
& \quad \quad \int_T \left\{D^{\bs{\alpha}} B_{ijkl}^{d,T}(\bs{p})\right\} \left\{D^{\bs{\alpha}} B_{i^{\prime} j^{\prime} k^{\prime} l^{\prime}}^{d,T}(\bs{p})\right\} \mathrm{d} \bs{p} \nonumber \\
&= \sum_{|\bs{\alpha}| = 2} \bs{\gamma}^{\top}_T \mathbf{P}_T^{\bs{\alpha}} \bs{\gamma}_T = \bs{\gamma}^{\top}_T \mathbf{P}_T \bs{\gamma}_T, \label{EQN:EsTP}
\end{align}
where each $\mathbf{P}_T^{\bs{\alpha}}$  is a $\binom{d+3}{3} \times \binom{d+3}{3}$ matrix with entries $\int_T \{D^{\bs{\alpha}} B_{ijkl}^{d,T}(\bs{p})\}$ $\{D^{\bs{\alpha}} B_{i^{\prime} j^{\prime} k^{\prime} l^{\prime}}^{d,T}(\bs{p})\} \mathrm{d} \bs{p}$ for $\bs{\alpha}$ satisfying $|\bs{\alpha}| = 2$. Applying the results of directional derivatives in Section A.5, the second order derivative of $\mathbf{B}_d(\bs{p})$ at directions $\bs{u}_1$ and $\bs{u}_2$ are
\[D_{\bs{u}_2} D_{\bs{u}_1}\mathbf{B}_d(\bs{p}) =  \mathbf{C}^{(2)}_d(\bs{u}_1, \bs{u}_2)\mathbf{B}_{d-2}(\bs{p}).
\]

For notation simplicity, we denote $\mathbf{C}_d^{\bs{\alpha}}$ the derivative coefficient matrix for $D^{\bs{\alpha}}\mathbf{B}_d(\bs{p})$ with respect to direction(s) of Cartesian coordinate system, i.e., $D^{\bs{\alpha}}\mathbf{B}_d(\bs{p}) = \mathbf{C}_d^{\bs{\alpha}}\mathbf{B}_{d-|\bs{\alpha}|}(\bs{p})$. Therefore, for $\bs{\alpha}$ with $|\bs{\alpha}| = 2$, we have 
\begin{align*}
\int_T \left\{D^{\bs{\alpha}}\mathbf{B}_d(\bs{p})\right\} \left\{D^{\bs{\alpha}}\mathbf{B}_d(\bs{p})\right\}^{\top} \mathrm{d}\bs{p} = \mathbf{C}_d^{\bs{\alpha}} \left\{\int_T  \mathbf{B}_{d-2}(\bs{p}) \mathbf{B}_{d-2}(\bs{p})^{\top} \mathrm{d}\bs{p}\right\} (\mathbf{C}_d^{\bs{\alpha}})^{\top}.
\end{align*}
By the Lemma 15.29 in \cite{Lai:Schumaker:07}, we have 
\begin{align*}
\int_TB_{\nu \mu \kappa \delta}^{d-2, T}(\bs{p}) B_{\nu^{\prime} \mu^{\prime} \kappa^{\prime} \delta^{\prime}}^{d-2, T}(\bs{p}) \mathrm{d}\bs{p} 
&=\frac{{\nu+\nu^{\prime} \choose \nu}{\mu+\mu^{\prime} \choose \mu}{\kappa+\kappa^{\prime} \choose \kappa} {\delta+\delta^{\prime} \choose \delta}}{{2d-4 \choose d-2}} \int_{T}B_{\nu+\nu^{\prime}, \mu+\mu^{\prime}, \kappa+\kappa^{\prime}, \delta+\delta^{\prime}}^{2d-4, T}(\bs{p}) \mathrm{d}\bs{p}   \\
&= \frac{{\nu+\nu^{\prime} \choose \nu}{\mu+\mu^{\prime} \choose \mu}{\kappa+\kappa^{\prime} \choose \kappa} {\delta+\delta^{\prime} \choose \delta}}{{2d-4 \choose d-2}{2d-1 \choose 3}}V_T,
\end{align*}
recall that $V_T$ is the volume of tetrahedron $T$. Let $\mathbf{L}_T^{d-2} = \int_T  \mathbf{B}_{d-2}(\bs{p}) \mathbf{B}_{d-2}(\bs{p})^{\top} \mathrm{d}\bs{p}$ be the $\binom{d+1}{3} \times \binom{d+1}{3}$ matrix with entries $\int_{T}B_{\nu \mu \kappa \delta}^{d-2,T}(\bs{p}) B_{\nu^{\prime} \mu^{\prime} \kappa^{\prime} \delta^{\prime}}^{d-2,T} (\bs{p})\mathrm{d}\bs{p}$. We finally obtain $\mathbf{P}_{T}^{\bs{\alpha}}=\mathbf{C}_d^{\bs{\alpha}}\mathbf{L}_{T}^{d-2}(\mathbf{C}_d^{\bs{\alpha}})^{\top}$ and $\mathbf{P}_{T}=\sum_{|\bs{\alpha}| = 2} \mathbf{P}_{T}^{\bs{\alpha}}$. Consequently, $\mathcal{E}(s) = \bs{\gamma}^{\top} \mathbf{P} \bs{\gamma}$.

Following the same example in Section A.5, similarly, 
the barycentric coordinates of $y =\bs{v}_3 - \bs{v}_2$ and $z = \bs{v}_4 - \bs{v}_2$ are $(0,-1,1,0)$ and $(0,-1,0,1)$, respectively. Following similar inductions, the second directional derivative of the Bernstein basis polynomials are the linear combination of $\{B_{1000}^{1,T}(\bs{p}), B_{0100}^{1,T}(\bs{p}), B_{0010}^{1,T}(\bs{p}), B_{0001}^{1,T}(\bs{p})\}$. For example, given the formulae of $\mathbf{C}^{(1)}_2$ and $\mathbf{C}^{(1)}_3$ shown above, we have $D^{(1,0,0)}B_{2100}^{3,T}(\bs{p})=3B_{1100}^{2,T}(\bs{p}) - 3B_{2000}^{2,T}(\bs{p}),$ and $D^{(2,0,0)}B_{2100}^{3,T}(\bs{p})
=-12B_{1000}^{1,T}(\bs{p})+6B_{0100}^{1,T}(\bs{p})$. According to the results in Sections A.5 and A.6, the coefficient matrices for second order derivatives are $\mathbf{C}_3(2,0,0)$, $\mathbf{C}_3(0,2,0)$,  $\mathbf{C}_3(0,0,2)$, $\mathbf{C}_3(1,1,0)$, $\mathbf{C}_3(1,0,1)$, and $\mathbf{C}_3(0,1,1)$:
\[
\begingroup 
\setlength\arraycolsep{2pt}
\stackrel{\mbox{$\mathbf{C}_3(2,0,0)$}}{
\begin{pmatrix}
6 & 0  & 0 & 0\\
-12 & 6 & 0 & 0\\
0 & 0  & 6 & 0\\
0 & 0  & 0 & 6\\
6 & -12  & 0 & 0\\
0 & 0  & -12 & 0\\
0 & 0  & 0 & -12\\
0 & 0  & 0 & 0\\
0 & 0  & 0 & 0\\
0 & 0  & 0 & 0\\
0 & 6  & 0 & 0\\
0 & 0  & 6 & 0\\
0 & 0  & 0 & 6\\
0 & 0  & 0 & 0\\
0 & 0  & 0 & 0\\
0 & 0  & 0 & 0\\
0 & 0  & 0 & 0\\
0 & 0  & 0 & 0\\
0 & 0  & 0 & 0\\
0 & 0  & 0 & 0\\
\end{pmatrix}},
\endgroup
\begingroup 
\setlength\arraycolsep{2pt}
\stackrel{\mbox{$\mathbf{C}_3(0,2,0)$}}{
\begin{pmatrix}
0 & 0  & 0 & 0\\
0 & 0  & 0 & 0\\
0 & 0  & 0 & 0\\
0 & 0  & 0 & 0\\
6 & 0  & 0 & 0\\
-12 & 0  & 0 & 0\\
0 & 0  & 0 & 0\\
6 & 0  & 0 & 0\\
0 & 0  & 0 & 0\\
0 & 0  & 0 & 0\\
0 & 6  & 0 & 0\\
0 & -12  & 6 & 0\\
0 & 0  & 0 & 6\\
0 & 6  & -12 & 0\\
0 & 0  & 0 & -12\\
0 & 0  & 0 & 0\\
0 & 0  & 6 & 0\\
0 & 0  & 0 & 6\\
0 & 0  & 0 & 0\\
0 & 0  & 0 & 0\\
\end{pmatrix}},
\endgroup
\begingroup 
\setlength\arraycolsep{2pt}
\stackrel{\mbox{$\mathbf{C}_3(0,0,2)$}}{
\begin{pmatrix}
0 & 0  & 0 & 0\\
0 & 0  & 0 & 0\\
0 & 0  & 0 & 0\\
0 & 0  & 0 & 0\\
6 & 0  & 0 & 0\\
0 & 0  & 0 & 0\\
-12 & 0  & 0 & -12\\
0 & 0  & 0 & 0\\
0 & 0  & 0 & 0\\
6 & 0  & 0 & 0\\
0 & 6  & 0 & 0\\
0 & 0  & 6 & 0\\
0 & -12  & 0 & 6\\
0 & 0  & 0 & 0\\
0 & 0  & -12 & 0\\
0 & 6  & 0 & -12\\
0 & 0  & 0 & 0\\
0 & 0  & 6 & 0\\
0 & 0  & 0 & 6 \\
0 & 0  & 0 & 0\\
\end{pmatrix}},
\endgroup
\]

\[
\begingroup 
\setlength\arraycolsep{2pt}
\stackrel{\mbox{$\mathbf{C}_3(1,1,0)$}}{
\begin{pmatrix}
0 & 0  & 0 & 0\\
-6 & 0  & 0 & 0\\
6 & 0  & 0 & 0\\
0 & 0  & 0 & 0\\
6 & -6  & 0 & 0\\
-6 & 6  & -6 & 0\\
0 & 0  & 0 & -6\\
0 & 0  & 6 & 0\\
0 & 0  & 0 & 6\\
0 & 0  & 0 & 0\\
0 & 6 & 0 & 0\\
0 & -6  & 6 & 0\\
0 & 0  & 0 & 6\\
0 & 0  & -6 & 0\\
0 & 0  & 0 & -6\\
0 & 0  & 0 & 0\\
0 & 0  & 0 & 0\\
0 & 0  & 0 & 0\\
0 & 0  & 0 & 0\\
0 & 0  & 0 & 0\\
\end{pmatrix}},
\endgroup
~
\begingroup 
\setlength\arraycolsep{2pt}
\stackrel{\mbox{$\mathbf{C}_3(1,0,1)$}}{
\begin{pmatrix}
0 & 0  & 0 & 0\\
-6 & 0  & 0 & 0\\
0 & 0  & 0 & 0\\
6 & 0  & 0 & 0\\
6 & -6  & 0 & 0\\
0 & 0  & -6 & 0\\
-6 & 6  & 0 & -6\\
0 & 0  & 0 & 0\\
0 & 0  & 6 & 0\\
0 & 0  & 0 & 6\\
0 & 6  & 0 & 0\\
0 & 0  & 6 & 0\\
0 & -6  & 0 & 6\\
0 & 0  & 0 & 0\\
0 & 0  & -6 & 0\\
0 & 0  & 0 & -6\\
0 & 0  & 0 & 0\\
0 & 0  & 0 & 0\\
0 & 0  & 0 & 0\\
0 & 0  & 0 & 0\\
\end{pmatrix}},
\endgroup
~
\begingroup 
\setlength\arraycolsep{2pt}
\stackrel{\mbox{$\mathbf{C}_3(0,1,1)$}}{
\begin{pmatrix}
0 & 0  & 0 & 0\\
0 & 0  & 0 & 0\\
0 & 0  & 0 & 0\\
0 & 0  & 0 & 0\\
6 & 0  & 0 & 0\\
-6 & 0  & 0 & 0\\
-6 & 0  & 0 & 0\\
0 & 0  & 0 & 0\\
6 & 0  & 0 & 0\\
0 & 0  & 0 & 0\\
0 & 6  & 0 & 0\\
0 & -6  & 6 & 0\\
0 & -6 & 0 & 6\\
0 & 0  & -6 & 0\\
0 & 6  & -6 & -6\\
0 & 0  & 0 & -6\\
0 & 0  & 0 & 0\\
0 & 0  & 0 & 0\\
0 & 0  & 0 & 0\\
0 & 0  & 0 & 0\\
\end{pmatrix}}.
\endgroup
\]
Therefore, we have
\[
    \mathbf{L}_T^1 =\int_T  \mathbf{B}_{1}(\bs{p}) \mathbf{B}_{1}(\bs{p})^{\top} \mathrm{d}\bs{p}
    = \frac{1}{12}\left(
    \begin{array}{cccc}
        2 & 1 & 1 & 1 \\
        1 & 2 & 1 & 1 \\
        1 & 1 & 2 & 1 \\
        1 & 1 & 1 & 2 \\
    \end{array}
    \right),
    \text{ and }
    \mathbf{P}_T 
    = \sum_{|\bs{\alpha}| = 2} \mathbf{C}_3(\bs{\alpha}) \mathbf{L}_{T}^{1}\{\mathbf{C}_3(\bs{\alpha})\}^{\top}.
\]

\fontsize{12}{14pt plus.8pt minus .6pt}\selectfont
\vskip 0.1in  \noindent \textbf{B. Proof of theoretical results} \vskip 0.1in 
\renewcommand{\theequation}{B.\arabic{equation}}
\renewcommand{\thesubsection}{B.\arabic{subsection}}
\renewcommand{\thetheorem}{B.\arabic{theorem}}
\renewcommand{\thelemma}{{\rm B.\arabic{lemma}}}
\renewcommand{\theproposition}{B.\arabic{proposition}}
\renewcommand{\thecorollary}{B.\arabic{corollary}}
\renewcommand{\thefigure}{B.\arabic{figure}}
\renewcommand{\thetable}{B.\arabic{table}}
\renewcommand{\theremark}{B.\arabic{remark}}
\setcounter{equation}{0}
\setcounter{theorem}{0}
\setcounter{lemma}{0}
\setcounter{figure}{0}
\setcounter{remark}{0}
\setcounter{proposition}{0}
\setcounter{subsection}{0}
\setcounter{subsubsection}{0}

\label{SEC:proof_main}

\vskip .10in \noindent \textbf{B.1. Preliminaries} \vskip .10in
\label{SSEC:Pre}

In this section, we start with the stability condition for trivariate spline bases over triangulations, and establish the uniform rate of the approximation of empirical inner product to the theoretical inner product based on the stability property. 

\begin{lemma}
\label{LEM:normequity}
Let $\{B_{\xi}\}_{\xi \in \mathcal{M}}$ be the basis for $\mathcal{S}_d^r(\triangle)$ constructed in \cite{Lai:Schumaker:07}, where $\mathcal{M}$ stands for the index set of spline bases.  Then there exist positive constants $ C_1$, $C_2$ depending on degree $d$ and partition quasi-uniform parameter $\beta$ such that
\begin{eqnarray*}
	C_1|\triangle|^3\sum_{\xi \in \mathcal{M}}|c_{\xi}|^2
	\le\left\Vert \sum_{\xi \in \mathcal{M}}c_{\xi}B_{\xi}\right\Vert _{L^2(\Omega)}^2
	\le C_2|\triangle|^3\sum_{\xi \in \mathcal{M}}|c_{\xi}|^2, 
\end{eqnarray*}
for all $c_{\xi},\xi \in \mathcal{M}$.
\end{lemma}
\begin{proof} 
The proof of Lemma \ref{LEM:normequity} follows directly from the Theorem 17.18 in \cite{Lai:Schumaker:07}.
\end{proof}
We illustrate the connection between the theoretical norm and $L^2$ norm in Lemma \ref{LEM:L2TheoNorm}.

\begin{lemma}
\label{LEM:L2TheoNorm} 
Let $\{B_{\xi}\}_{\xi \in \mathcal{M}}$ be the basis for $\mathcal{S}_d^r(\triangle)$ as in Lemma \ref{LEM:normequity}.
Under Assumption (A3), if $d\ge 6r+3$ and
$\triangle$ is a $\beta$-quasi-uniform triangulation, there exist positive constants $C_1$, $C_2$ such that
\begin{eqnarray*}
	C_1\left\Vert \sum_{\xi \in \mathcal{M}}c_{\xi}B_{\xi}\right\Vert _{L^2(\Omega)}^2
	\le \left\Vert \sum_{\xi \in \mathcal{M}}c_{\xi}B_{\xi}\right\Vert _{\Omega}^2
	\le C_2\left\Vert \sum_{\xi \in \mathcal{M}}c_{\xi}B_{\xi}\right\Vert _{L^2(\Omega)}^2, 
\end{eqnarray*}
for all $c_{\xi},\xi \in \mathcal{M}$.
\end{lemma}
\begin{proof} 
It is straightforward to obtain upper and lower bounds of $\Vert \sum_{\xi \in \mathcal{M}}c_{\xi}B_{\xi}\Vert _{\Omega}^2$ by Assumption (A3).
\end{proof}

\begin{proof}[Proof of Lemma \ref{LEM:Rnorder}] 
According to the definition of empirical inner product and induced norm of the theoretical $L^2$ inner product in Section 3,
\[
	\left\langle g_1,g_2\right\rangle _{n,\Omega}
	=\frac{1}{n}\sum_{i=1}^{n}\left\{ \sum_{\xi \in \mathcal{M}}c_{\xi}B_{\xi}\left(\bs{p}_i\right) \right\} \left\{ \sum_{\zeta \in \mathcal{M}}\tilde{c}_{\zeta }B_{\zeta }\left( \bs{p}_i\right) \right\}
	=\sum_{\xi,\zeta \in \mathcal{M}}c_{\xi}\tilde{c}_{\zeta}\left\langle B_{\xi},B_{\zeta }\right\rangle _{n,\Omega},
\]
\[
\left\| g_1\right\| _{\Omega}^2=\sum_{\xi,\xi^{\prime }\in \mathcal{%
M}}c_{\xi}c_{\xi^{\prime }}\left\langle B_{\xi},B_{\xi^{\prime
}}\right\rangle _{\Omega},\hbox{ and }\left\| g_2\right\|
_{\Omega}^2=\sum_{\zeta,\zeta^{\prime }\in \mathcal{M}}\tilde{c}_{\zeta }\tilde{c%
}_{\zeta^{\prime }}\left\langle B_{\zeta },B_{\zeta^{\prime
}}\right\rangle _{\Omega}.
\]
By Lemma \ref{LEM:normequity}, we have for $g_1$ and $g_2$, 
$C_1|\triangle|^3\sum_{\xi \in \mathcal{M}}|c_{\xi}|^2 
		\le \Vert g_1\Vert _{L^2(\Omega)}^2
		\le C_2|\triangle|^3\sum_{\xi \in \mathcal{M}}|c_{\xi}|^2$
and $C_1|\triangle|^3\sum_{\zeta \in\mathcal{M}}|\tilde{c}_{\zeta }|^2 
		\le \Vert g_2\Vert_{L^2(\Omega)}^2
		\le C_2|\triangle|^3\sum_{\zeta \in\mathcal{M}}|\tilde{c}_{\zeta }|^2$.
Consequently, by Lemma \ref{LEM:L2TheoNorm}, 
\[
	C_1|\triangle|^3\left\{ \sum_{\xi \in \mathcal{M}}|c_{\xi}|^2\sum_{\zeta \in \mathcal{M}}|\tilde{c}_{\zeta }|^2\right\}^{1/2}
	\leq \left\| g_1\right\| _{\Omega}\left\| g_2\right\|_{\Omega}
	\leq C_2|\triangle|^3\left\{ \sum_{\xi \in \mathcal{M}}|c_{\xi}|^2\sum_{\zeta \in \mathcal{M}}|\tilde{c}_{\zeta}|^2\right\}^{1/2}.
\]
Thus,
\begin{align}
	R_n
	&\leq \frac{\sum_{\xi,\zeta \in \mathcal{M}}|c_{\xi}\tilde{c}_{\zeta}|}
		{C_1|\triangle|^3\left\{ \sum_{\xi \in \mathcal{M}}|c_{\xi}|^2\sum_{\zeta \in \mathcal{M}}|\tilde{c}_{\zeta }|^2\right\}^{1/2}}
		\max_{\xi,\zeta \in \mathcal{M}}\left| \left\langle B_{\xi},B_{\zeta }\right\rangle _{n,\Omega}-\left\langle B_{\xi},B_{\zeta }\right\rangle _{\Omega}\right| \nonumber \\
	&\leq C_1^{-1}\left|\triangle\right|^{-3}
		\max_{\xi,\zeta \in \mathcal{M}}
		\left|\left\langle B_{\xi},B_{\zeta }\right\rangle _{n,\Omega}
		-\left\langle B_{\xi},B_{\zeta }\right\rangle _{\Omega}\right|.
\label{EQ:R_n}
\end{align}
To obtain the conclusion, it suffices to show that with probability one,
\begin{eqnarray}
	\max_{\xi,\zeta \in \mathcal{M}}\left| \left\langle B_{\xi},B_{\zeta }\right\rangle _{n,\Omega}-\left\langle B_{\xi},B_{\zeta }\right\rangle _{\Omega}\right| 
	=O_P\left\{(\log n)^{1/2}(nN)^{-1/2}\right\}. \label{EQ:B_n-B_2}
\end{eqnarray}
Let 
$ R_{\xi,\zeta,i}
	=B_{\xi}\left(\bs{p}_i\right) B_{\zeta }\left( \bs{p}_i\right)
		-\mathrm{E}B_{\xi}\left( \bs{p}_i\right) B_{\zeta }\left(\bs{p}_i\right)
$. Then we have the second moment
	$\mathrm{E}R_{\xi \mathbf{,}\zeta,i}^2
	=\mathrm{E}\{B_{\xi}^2\left(\bs{p}_i\right)B_{\zeta }^2\left(\bs{p}_i\right)\}
		-\left\{\mathrm{E}B_{\xi}\left(\bs{p}_i\right)B_{\zeta }\left(\bs{p}_i\right)\right\}^2$,
where 
$\mathrm{E}\{B_{\xi}^2\left( \bs{p}_i\right) B_{\zeta }^2\left( \bs{p}_i\right)\} 
	\sim \left|\triangle\right|^3$,
$\left\{\mathrm{E}B_{\xi}\left( \bs{p}_i\right) B_{\zeta}\left(\bs{p}_i\right) \right\}^2
	\sim \left| \triangle \right|^6$. Hence, 
$\mathrm{E}R_{\xi,\zeta,i}^2\sim \left| \triangle \right|^3$. Note that the $k$-th moment is
$\mathrm{E}| R_{\xi,\zeta,i}|^k 
	\leq 2^{k-1}\{\mathrm{E}|B_{\xi}(\bs{p}_i)B_{\zeta}(\bs{p}_i)|^k
		+|\mathrm{E}B_{\xi}(\bs{p}_i)B_{\zeta }(\bs{p}_i)|^k\}$,
where 
$\mathrm{E}\left| B_{\xi}\left( \bs{p}_i\right) B_{\zeta }\left(\bs{p}_i\right) \right|^k
	\sim \left| \triangle \right|^3$,
$\left| \mathrm{E}B_{\xi}\left( \bs{p}_i\right) B_{\zeta }\left( \bs{p}_i\right) \right|^k
	\sim \left| \triangle \right|^{3k}$. 
Then there exists a constant $C>0$ such that $\mathrm{E}\left|R_{\xi \mathbf{,}\zeta,i}\right|^k\leq C2^{k-1}k!\mathrm{E}R_{\xi\mathbf{,}\zeta,i}^2$. Thus, $\{R_{\xi\mathbf{,}\zeta,i}\}_{i=1}^n$ satisfying Cramer's condition with some constant $C\sim O(1)$. 
By Bernstein's inequality in \cite{Bosq:98}, for $\delta >0$ large enough,
\begin{eqnarray}
	P\left( \frac{1}{n}\left| \sum_{i=1}^{n}R_{\xi,\zeta,i}\right| \ge \delta \sqrt{\frac{\log n}{nN}}\right) 
	\le 
	2\exp \left( \frac{-\delta^2\log n}{4+2c\delta \sqrt{N\log n/n}}\right) 
	\leq 2n^{-4}. 
\label{EQ:Bernstein}
\end{eqnarray}
It is easy to see that the cardinality of $\mathcal{M}$ is $(d+1)(d+2)(d+3)N/6$ as discussed in Section 2.3.
Thus, for the $\delta>0$ in (\ref{EQ:Bernstein}),
\[
	\sum_{n=1}^{\infty }P\left( \max_{\xi,\zeta \in \mathcal{M}}\left| \frac{1}{n}\sum_{i=1}^{n}R_{\xi,\zeta,i}\right| \geq \delta \sqrt{\frac{\log n}{nN}}\right)
	\leq
	\frac{1}{6}\sum_{n=1}^{\infty}\frac{\left\{(d+1)(d+2)(d+3)N\right\}^2}{n^4}
	\leq C\sum_{n=1}^{\infty}n^{-2}<\infty,
\]
Borel-Cantelli Lemma entails that
$\max_{\xi,\zeta \in \mathcal{M}}| n^{-1}\sum_{l=1}^{n}R_{\xi,\zeta,i}|
	=O_P\{(\log n)^{1/2}(nN)^{-1/2}\}$.
The desired result follows from (\ref{EQ:R_n}) and (\ref{EQ:B_n-B_2}). 
\end{proof}

As a direct result of Lemma \ref{LEM:Rnorder}, we have
\begin{eqnarray}
	\sup_{g\in \mathcal{S}_d^r(\triangle)}\left| \left. \left\| g\right\| _{n,\Omega}^2\right/ \Vert g\Vert _{\Omega}^2-1\right| 
	=O_P\left\{(N\log n)^{1/2}{n}^{-1/2}\right\} .  
\label{EQ:normratio}
\end{eqnarray}

\vskip .10in \noindent \textbf{B.2. Size of the bias and noise terms (Propositions \ref{PROP:bias} -- \ref{PROP:maxnorm})} \vskip .10in
~\label{SSEC:BiasNoise}

In this section, we measure the size of the bias term and show the proofs of Proposition \ref{PROP:bias} -- \ref{PROP:maxnorm}. We firstly proof Lemma \ref{LEM:Vnorder}.

\begin{proof}[Proof of Lemma \ref{LEM:Vnorder}] 
According to Lemma 15.2 of \cite{Lai:Schumaker:07}, for any $g\in \mathcal{S}_d^r(\triangle)$,
\[
	\left\| g\right\| _{\infty }
	=\left\| g\right\| _{\infty,T}
	\leq K_1V_T^{-1/2}\left\| g\right\|_{L^2(T)}
	\leq \frac{K}{\varrho_{\triangle}^{3/2}}\left\| g\right\| _{L^2(T)}
	\leq \frac{K}{\varrho_{\triangle}^{3/2}}\left\| g\right\| _{L^2(\Omega)}
	\leq\frac{K_{\beta }}{\left|\triangle\right|^{3/2}}\left\| g\right\|_{L^2(\Omega)},
\] 
where $K_\beta$ is a positive constant dependent on $\beta$ and recall that $V_T$ is the volume of tetrahedron $T$. 
Consequently, by the
Markov's inequality Theorem 15.28 in \cite{Lai:Schumaker:07}, for any $g\in \mathcal{S}_d^r(\triangle)$, 
$\| g\| _{\infty,\Omega}\leq K_{\beta }| \triangle|^{-3/2}\| g\| _{L^2(\Omega)}$ and $\| g\|_{\mathcal{E}}\leq K_{\beta }| \triangle |^{-2}\| g\| _{L^2(\Omega)}$.
Equation (\ref{EQ:normratio}) implies that
\[
	\sup_{g\in\mathcal{S}_d^r(\triangle)}\left\{ \left. \left\| g\right\| _{n,\Omega}\right/\left\| g\right\| _{L^2(\Omega)}\right\} 
	\geq \left[1-O_P\left\{(N\log n)^{1/2}n^{-1/2}\right\}\right]^{1/2}.
\]
Thus, we have
\begin{eqnarray*}
	V_n &\leq &K_{\beta }\left|\triangle\right|^{-3/2}\left[1-O_P\left\{(N\log n)^{1/2}n^{-1/2}\right\} \right]^{-1/2}
	=O_P\left(\left|\triangle\right|^{-3/2}\right), \\
	\overline{V}_n &\leq &K_{\beta }\left|\triangle\right|^{-2}\left[1-O_P\left\{(N\log n)^{1/2}n^{-1/2}\right\}\right]^{-1/2}
	=O_P\left(\left|\triangle\right|^{-2}\right).
\end{eqnarray*}
\end{proof}

\begin{proof}[Proof of Proposition \ref{PROP:bias}] 
By triangle inequality and the definition of $V_n$, 
\begin{eqnarray}
    \left\| m-s_{\rho_n,m}\right\| _{\infty,\Omega}
    &\leq& \left\|m-s_{0,m}\right\| _{\infty,\Omega}+\left\| s_{0,m}-s_{\rho_n
,m}\right\| _{\infty,\Omega} \nonumber \\ 
    &\leq& \left\|m-s_{0,m}\right\| _{\infty,\Omega}+V_n\left\| s_{0,m}-s_{\rho_n,m}\right\| _{n,\Omega}. \label{EQ:pls-ps-sup}
\end{eqnarray}
By the definition of $s_{\rho_n,m}$, we have $\forall u\in \mathcal{S}_d^r(\triangle)$, $t\in\mathbb{R}_{+}$,
\begin{eqnarray*}
    &&n\left\|m-s_{\rho_n,m}\right\|_{n,\Omega}^2 + \rho_n\|s_{\rho_n,m}\|_{\mathcal{E}}^2
    \leq n\left\|m-s_{\rho_n,m}-tu\right\|_{n,\Omega}^2 + \rho_n\|s_{\rho_n,m}+tu\|_{\mathcal{E}}^2 \\
    &\Rightarrow& n\left\langle m-s_{\rho_n, m},u\right\rangle _{n,\Omega}
	-\rho_n\left\langle s_{\rho_n, m},u\right\rangle _{\mathcal{E}}
	\leq \frac{t}{2} \left(\left\|u\right\|_{n,\Omega}^2+ \rho_n\|u\|_{\mathcal{E}}^2\right).
\end{eqnarray*}
Similarly, take $t\in\mathbb{R}_{-}$ and induct in a similar fashion, we thus obtain 
\begin{equation}
    \left|n\left\langle m-s_{\rho_n, m},u\right\rangle _{n,\Omega}
	-\rho_n\left\langle s_{\rho_n, m},u\right\rangle _{\mathcal{E}}\right|
	\leq \frac{t}{2} \left(\left\|u\right\|_{n,\Omega}^2+ \rho_n\|u\|_{\mathcal{E}}^2\right).
\label{EQN:nOmegaE}
\end{equation}
Since (\ref{EQN:nOmegaE}) holds for any $t\in\mathbb{R}$, we can conclude that the penalized spline $s_{\rho_n,m}$ of $m$ is characterized by the orthogonality relations
\begin{eqnarray}
	n\left\langle m-s_{\rho_n, m},u\right\rangle _{n,\Omega}
	=\rho_n\left\langle s_{\rho_n, m},u\right\rangle _{\mathcal{E}},
	\quad \textrm{for all }u\in \mathcal{S}_d^r(\triangle).
\label{EQ:pls}
\end{eqnarray}
When $\rho_n=0$, we have for $s_{0, m}$,
\begin{eqnarray}
\left\langle m -s_{0, m},u\right\rangle _{n,\Omega}=0,\quad \textrm{for all }u\in \mathcal{S}_d^r(\triangle).  \label{EQ:ls}
\end{eqnarray}
Combining (\ref{EQ:pls}) and (\ref{EQ:ls}), we obtain 
$ n\left\langle s_{0,m}-s_{\rho_n, m},u\right\rangle _{n,\Omega}
	=\rho_n \left\langle s_{\rho_n,m},u\right\rangle _{\mathcal{E}}$, 
for all $u\in \mathcal{S}_d^r(\triangle)$. Inserting $u=s_{0, m}-s_{\rho_n, m}$ yields
that
\begin{eqnarray}
	n\left\| s_{0, m}-s_{\rho_n, m}\right\| _{n,\Omega}^2
	=\rho_n\left\langle s_{\rho_n, m},s_{0, m}-s_{\rho_n, m}\right\rangle_{\mathcal{E}}.  
\label{EQ:pls-ls}
\end{eqnarray}
Thus, by Cauchy-Schwarz inequality and the definition of $\overline{V}_n$,
\[
	n\left\| s_{0,m}-s_{\rho_n,m}\right\| _{n,\Omega}^2
	\leq\rho_n \left\| s_{\rho_n,m}\right\|_{\mathcal{E}}\left\|s_{0,m}-s_{\rho_n,m}\right\| _{\mathcal{E}}
	\leq \overline{V}_n\rho_n \left\| s_{\rho_n,m}\right\|_{\mathcal{E}}\left\| s_{0,m}-s_{\rho_n,m}\right\|_{n,\Omega}.
\]
Similarly, using (\ref{EQ:pls-ls}), we have
\[
	n\left\| s_{0,m}-s_{\rho_n, m}\right\| _{n,\Omega}^2
	=\rho_n 
		\left\{\left\langle s_{\rho_n, m},s_{0, m}\right\rangle_{\mathcal{E}}
		-\left\langle s_{\rho_n, m},s_{\rho_n, m}\right\rangle _{\mathcal{E}}\right\} 
	\geq 0.
\]
Thus, by Cauchy-Schwarz inequality, 
$\| s_{\rho_n, m}\|_{\mathcal{E}}^2\leq \langle s_{\rho_n,g},s_{0, m}\rangle _{\mathcal{E}}
	\leq \|s_{\rho_n, m}\| _{\mathcal{E}}\| s_{0, m}\| _{\mathcal{E}}$,
which implies that $\left\| s_{\rho_n, m}\right\|
_{\mathcal{E}}\leq \left\| s_{0, m}\right\|
_{\mathcal{E}}$. Therefore
\begin{eqnarray}
\left\| s_{0, m}-s_{\rho_n, m}\right\| _{n,\Omega}\leq
n^{-1}\overline{V}_n \rho_n \left\| s_{0, m}\right\|
_{\mathcal{E}}. \label{EQ:pls-ps2}
\end{eqnarray}
Combining (\ref{EQ:pls-ps-sup}), (\ref{EQ:pls-ps2}) and Lemma \ref{eqGS} in Section 2.3 in main part yields that
\begin{eqnarray*}
    \left\| s_{0, m}-s_{\rho_n, m}\right\| _{\infty,\Omega}
    &\leq& V_n \left\| s_{0, m}-s_{\rho_n, m}\right\| _{n,\Omega}
    \leq n^{-1}V_n \overline{V}_n \rho_n \left\| s_{0, m}\right\|_{\mathcal{E}} \nonumber \\
    &=& n^{-1}V_n \overline{V}_n \rho_n C_1\left(\left| m\right| _{2,\infty,\Omega}+\sum_{|\bs{\alpha}|=2}
		\left\|D^{\bs{\alpha}}
		\left(m - s_{0, m}\right) \right\| _{\infty,\Omega}\right) \nonumber \\
	&\leq& n^{-1}V_n\overline{V}_n\rho_nC_2 
		\left(\left|m\right| _{2,\infty,\Omega}
		+\left| \triangle \right|^{\ell -1}\left| m\right|_{\ell+1,\infty,\Omega}\right). 
\end{eqnarray*}
Plugging the orders of $V_n$ and $\overline{V}_n$, therefore, 
\[
	\left\| s_{0, m}-s_{\rho_n, m}\right\| _{\infty,\Omega}
	=O_P\left\{\frac{\rho_n}{n\left|\triangle\right|^{7/2}}\left(\left|m\right|_{2,\infty,\Omega}
		+\left|\triangle\right|^{\ell -1}\left| m\right| _{\ell +1,\infty,\Omega}\right) \right\}.
\]
Hence by (\ref{EQ:pls-ps-sup}),
\begin{eqnarray*}
	\left\| m-s_{\rho_n, m}\right\| _{\infty,\Omega}
	&\leq& C_1\left| \triangle \right|^{\ell +1}\left| m\right|_{\ell+1,\infty,\Omega}
	+O_P\left\{\frac{\rho_n }{n\left|\triangle \right|^{7/2}}\left(\left|m\right| _{2,\infty,\Omega}
		+\left|\triangle\right|^{\ell-1}\left|m\right|_{\ell+1,\infty,\Omega}\right)\right\} .
\end{eqnarray*}
Therefore, Proposition \ref{PROP:bias} is established. 
\end{proof}

Next we give the proof of Proposition \ref{PROP:variance}.

\begin{proof}[Proof of Proposition \ref{PROP:variance}] 
It is known from \cite{Lai:Schumaker:07} that there is a locally support basis $B_{\xi}$, $\xi\in\mathcal{M}$ for $\mathcal{S}_d^r(\triangle)$. We write $s_{\rho_n,\epsilon}(\bs{p})=\sum_{\xi\in\mathcal{M}}c_{\rho_n,\xi}B_{\xi}(\bs{p})$ for some coefficients $c_{\rho_n,\xi}$. It is easy to obtain the orthogonal relations for the penalized spline $s_{\rho_n,\epsilon}$
$n\langle s_{\rho_n,\epsilon}-\varepsilon,u\rangle _{n,\Omega}
	+\rho_n \langle s_{\rho_n,\epsilon},u\rangle_{\mathcal{E}}=0$, 
for all $u\in \mathcal{S}_d^r(\triangle)$. Consequently
\[
	\sum_{i=1}^{n}s_{\rho_n,\epsilon}\left( \bs{p}_i\right)B_{\xi}\left( \bs{p}_i\right) 
		+\rho_n \langle s_{\rho_n,\epsilon},B_{\xi}\rangle _{\mathcal{E}}
	=\sum_{i=1}^nB_{\xi}\left(\bs{p}_i\right)\sigma\left(\bs{p}_i\right)\varepsilon_i
\]
for all $\xi \in \mathcal{M}$. 
Multiply $c_{\rho_n, \xi}$ to both sides
and take the summation of $\xi \in \mathcal{M}$ and use Cauchy-Schwarz's inequality, 
\begin{eqnarray*}
	\left\| s_{\rho_n,\epsilon}\right\| _{n,\Omega}^2 
	&\leq &\left\|s_{\rho_n,\epsilon}\right\| _{n,\Omega}^2+\frac{\rho_n }{n}\mathcal{E}(s_{\rho_n,\epsilon})
	= \frac{1}{n}\sum_{\xi \in \mathcal{M}}c_{\rho_n, \xi }
		\sum_{i=1}^nB_{\xi}\left(\bs{p}_i\right)\sigma\left(\bs{p}_i\right)\varepsilon_i \\
	&\leq& \left(\sum_{\xi \in \mathcal{M}}|c_{\rho_n, \xi }|^2\right)^{1/2}
		\left[ \sum_{\xi \in \mathcal{M}}\left\{ \frac{1}{n}\sum_{i=1}^{n}
		B_{\xi}\left( \bs{p}_i\right)\sigma \left( \bs{p}_i\right) \varepsilon_i\right\}^2\right]^{1/2} \\
	&\leq& \frac{1}{K_1|\triangle|^{3/2}}\Vert s_{\rho_n,\epsilon}\Vert _{L^2(\Omega)}
		\left[ \sum_{\xi \in \mathcal{M}}\left\{\frac{1}{n}\sum_{i=1}^{n}
		B_{\xi}\left(\bs{p}_i\right)\sigma\left(\bs{p}_i\right)\varepsilon_i\right\}^2\right]^{1/2}.
\end{eqnarray*}
It follows from (\ref{EQ:normratio}) that
\begin{align*}
	&\Vert s_{\rho_n,\epsilon}\Vert _{L^2(\Omega)}^2 
	\le \left\|s_{\rho_n,\epsilon}\right\| _{n,\Omega}^2
		+O_P\left\{n^{-1/2}(\log n)^{1/2}N^{1/2}\Vert s_{\rho_n,\epsilon}\Vert _{L^2(\Omega)}^2\right\} \\
	\le& \frac{\Vert s_{\rho_n,\epsilon}\Vert_{L^2(\Omega)}}{K_1|\triangle|^{3/2}}
		\left[ \sum_{\xi \in \mathcal{M}}\left\{\frac{1}{n}\sum_{i=1}^{n}B_{\xi}\left( \bs{p}_i\right)\sigma \left( \bs{p}_i\right) \varepsilon_i\right\}^2\right]^{1/2}
		+O_P\left\{\frac{(N\log n)^{1/2}}{n^{1/2}}\Vert s_{\rho_n,\epsilon}\Vert _{L^2(\Omega)}^2\right\}.
\end{align*}
Next we note that $V_{\Omega}\le C_1|\triangle|^3N$.
That is, $|\triangle|\ge (C_1N)^{-1/3}V_{\Omega}^{1/3}$. We thus have
\begin{eqnarray*}
	\Vert s_{\rho_n,\epsilon}\Vert _{L^2(\Omega)}
	&\leq&\frac{(C_1N)^{1/2}}{V_{\Omega}^{1/2}K_1}\left[ \sum_{\xi \in \mathcal{M}}\left\{ \frac{1}{n}\sum_{i=1}^{n}B_{\xi}\left( \bs{p}_i\right) \sigma \left(\bs{p}_i\right) \varepsilon_i\right\}^2\right]^{1/2}\\
	&&+O_P\left\{n^{-1/2}(N\log n)^{1/2}\Vert s_{\rho_n,\epsilon}\Vert _{L^2(\Omega)}\right\}
\end{eqnarray*}
Thus, we have
\begin{eqnarray}
	&&\left[1-O_P\left\{\frac{(N\log n)^{1/2}}{n^{1/2}}\right\}\right]\Vert s_{\rho_n,\epsilon}\Vert _{L^2(\Omega)}\nonumber \\
	&&\quad\quad\quad\quad\quad\quad\quad\quad\quad\quad\quad\quad \le \frac{CN^{1/2}}{V_{\Omega}^{1/2}}\left[ \sum_{\xi \in \mathcal{M}}\left\{ \frac{1}{n}\sum_{i=1}^{n}B_{\xi}\left( \bs{p}_i\right) \sigma \left(\bs{p}_i\right) \varepsilon_i\right\}^2\right]^{1/2}
\label{EQ:L2noise}
\end{eqnarray}
for a constant $C>0$. Observing that the two random variables $\varepsilon_i$ and $\bs{p}_i$ are independent, we have
$\mathrm{E}\{n^{-1}\sum_{i=1}^{n}\varepsilon_iB_{\xi}(\bs{p}_i) \sigma ( \bs{p}_i) \}^2
	=n^{-1}\mathrm{E}[B_{\xi}^2(\bs{p}_i)\sigma^2 ( \bs{p}_i)]$,
where according to Lemma 15.2 in \cite{Lai:Schumaker:07}, the expectation of $B_{\xi}^2\left(
\bs{p}_i\right)\sigma^2 \left( \bs{p}_i\right)$ can be
estimated as follows:
\begin{eqnarray}
	\!\!\!\!\!\!\!\!\mathrm{E}\left[B_{\xi}^2\left( \bs{p}_i\right)\sigma^2 \left(\bs{p}_i\right)\right]
	&\leq&  C_{\sigma}^2\left\| B_{\xi}\right\|^2_{L^2(\Omega)} \nonumber \\
	&\le& c(\beta,\sigma)|\triangle|^3\Vert B_{\xi}\Vert^2_{\infty,\Omega}
	\le C(\beta,\sigma)N^{-1}V_{\Omega} \label{EQ:error-rate}
\end{eqnarray}
for positive constants $c(\beta,\sigma)$ and $C(\beta,\sigma)$ which depend only on $\beta$ and $\sigma$. 
Thus, Assumption (A2) and (\ref{EQ:error-rate}) imply that
$n^{-1}\sum_{i=1}^{n}B_{\xi}( \bs{p}_i)\sigma ( \bs{p}_i)\varepsilon_i=O_P(n^{-1/2}N^{-1/2})$.
Therefore, we have
$\sum_{\xi \in \mathcal{M}}\{n^{-1}\sum_{i=1}^{n}B_{\xi}( \bs{p}_i)\sigma ( \bs{p}_i)\varepsilon_i\}^2=O_P(n^{-1})$.
Combining (\ref{EQ:L2noise}), we obtain that 
$\Vert s_{\rho_n,\epsilon}\Vert _{L^2(\Omega)}=O_P\left(N^{1/2}n^{-1/2}\right)$. 
So Proposition \ref{PROP:variance} is established. 
\end{proof}

Let $\bs{\Gamma}_{\rho_n}$ be the symmetric positive definite matrix
\begin{eqnarray}
	\left[ \frac{1}{n}\sum_{i=1}^{n}B_{\xi}\left(\bs{p}_i\right) B_{\zeta }\left( \bs{p}_i\right)+\frac{\rho_n }{n}\langle B_{\xi},B_{\zeta }\rangle_{\mathcal{E}}\right] _{\xi,\zeta \in \mathcal{M}}.
\label{DEF:Gamma_lambda}
\end{eqnarray}
We firstly show the bounds of $\bs{\Gamma}_{\rho_n}$.

\begin{lemma}
\label{LEM:Gamma_rho_sup}
Suppose (A3) holds, $d\geq 6r+3$, $\triangle $ is a $\beta $-quasi-uniform triangulation and $n^{-1}N\log n \rightarrow 0$ as $n \rightarrow \infty$. For $\bs{\Gamma}_{\rho_n}$ defined in (\ref{DEF:Gamma_lambda}), we have the following asymptotic properties:
\begin{itemize}
	\item[(i)] As $n\rightarrow \infty$, for some constants $0 < c_{\rho} < C_{\rho} < \infty$, with probability approaching one, 
	\[
		c_{\rho}|\triangle|^{3} 
		\leq \rho_{\min}(\bs{\Gamma}_{\rho_n}) 
		\leq \rho_{\max}(\bs{\Gamma}_{\rho_n}) 
		\leq  C_{\rho}\left(|\triangle|^{3}+\frac{\rho_{n}}{n|\triangle|}\right).
	\]
	\item[(ii)] There exists a constant $M_d > 0$ such that $\| \bs{\Gamma}_{\rho_n}^{-1} \|_{\infty}\leq M_d |\triangle|^{-3}$. 
	\item[(iii)] For every vector $\bs{a} = (a_1, \ldots, a_n)^{\top}$, there exists a constant $C_{d} > 0$ such that 
	\[
		\left\| \mathbf{B}^{\top}(\bs{p})\ \bs{\Gamma}_{\rho_n}^{-1}\frac{1}{n}\sum_{i=1}^{n} \mathbf{B}(\bs{p}_i)a_i \right\|_{\infty} 
		\leq C_{d} \|\bs{a}\|_{\infty}.
	\]
\end{itemize}
\end{lemma}

\begin{proof} 
(i) Let's randomly pick $\bs{\theta}\in\mathbb{R}^{|\mathcal{M}|}$, then for $g(\bs{p}) = \mathbf{B}^{\top}(\bs{p}) \bs{\theta}$, we have $g \in \mathcal{S}_d^r(\triangle)$. Then
\begin{equation*}
	\bs{\theta}^{\top}\bs{\Gamma}_{\rho_n}\bs{\theta} 
	=\bs{\theta}^{\top}\frac{1}{n}\sum_{i=1}^{n} \mathbf{B}(\bs{p}_i)\mathbf{B}^{\top}(\bs{p}_i)\bs{\theta}
		+ \bs{\theta}^{\top}\frac{\rho_n}{n} [\langle B_{m},B_{m^{\prime}}\rangle_{\mathcal{E}}]_{m,m^{\prime}\in \mathcal{M}} \bs{\theta}
		=\|g\|_{n,\Omega}^2 +\frac{\rho_n}{n} \|g\|_{\mathcal{E}}^2.
\end{equation*}
By (\ref{EQ:normratio}), 
	$\big| \|g\|^2_{n,\Omega}/\|g\|^2_{L^2(\Omega)} -1\big| \leq R_n$, 
and combined with Lemma \ref{LEM:normequity}, we have 
\[
	c (1-R_n) |\triangle|^{3}\Vert \bs{\theta}\Vert^{2}   
	\leq (1-R_n) \Vert g\Vert_{L^2(\Omega)}^{2}
	\leq \Vert g \Vert_{n,\Omega}^{2} 
	\leq (1+R_n) \Vert g \Vert_{L^2(\Omega)}^{2}
	\leq C (1+R_n)|\triangle|^{3}\Vert \bs{\theta}\Vert^{2}.
\]
Thus, $\rho_{\min}(\bs{\Gamma}_{\rho_n}) \geq c_{\rho} |\triangle|^{3}$ for some positive constant $c_{\rho}$. 

In the other side, as shown in the proof of Proposition \ref{PROP:bias} and by Lemma \ref{LEM:normequity}, 
$
	\|g\|_{\mathcal{E}}^2
	\leq C|\triangle|^{-4}\Vert g\Vert_{L^2(\Omega)}^{2}
	\leq C|\triangle|^{-1}\Vert\bs{\theta}\Vert ^{2}.
$
Thus, 
\[
	\rho_{\max}(\bs{\Gamma}_{\rho_n}) 
	\leq C\left\{(1+R_n)|\triangle|^{3}+\frac{\rho_n}{n}\frac{1}{|\triangle|}\right\} 
	\leq C_{\rho} \left(|\triangle|^{3}+\frac{\rho_n}{n|\triangle|}\right),
\]
for some positive constant $C_{\rho}$.

(ii) By (i), $\bs{\Gamma}_{\rho_n}$ is an invertible symmetric matrix, and its condition number $c_{d} = \rho_{\max}(\bs{\Gamma}_{\rho_n})/\rho_{\min}(\bs{\Gamma}_{\rho_n})$ satisfies $1 < c_{d} \leq c_{\rho}^{-1}C_{\rho}$. 
	
	According to the definition of banded matrix in \cite{DeVore:Lorentz:93}, a matrix $\mathbf{A} = (a_{ij})$ is said banded with bandwidth $b$ if $a_{ij} = 0, |i-j| \geq b$, and if $b$ is the smallest integer with this property. Based on the construction of trivariate splines, $\bs{\Gamma}_{\rho_n}$ is a banded matrix with bandwidth $b = {d + 3 \choose 3}$.
By Theorem 13.4.3 in \cite{DeVore:Lorentz:93},
$\|\bs{\Gamma}_{\rho_n}^{-1}\|_{\infty} 
		\leq 2 \tau^{-2b} \|\bs{\Gamma}_{\rho_n}^{-1}\|_{2} (1 - \tau)^{-1}$,
where 
	$\tau =  ( {c_{d}^2 - 1}/{c_{d}^2 + 1} )^{1/{4b}} < 1$.
	Therefore, there exist some positive constant $M_{d}$ such that $\|\bs{\Gamma}_{\rho_n}^{-1}\|_{\infty}  \leq M_{d} |\triangle|^{-3}$.
	
(iii) Combining (i) and (ii), it is straightforward to have
\begin{align*}
	\left\|\mathbf{B}^{\top}(\bs{p})\bs{\Gamma}_{\rho_n}^{-1}\frac{1}{n}\sum_{i=1}^{n} \mathbf{B}(\bs{p}_{j})a_{j} \right\|_{\infty}  
	&\leq \left\|\mathbf{B}^{\top}\right\|_{\infty} 
		\left\|\bs{\Gamma}_{\rho_n}^{-1} \right\|_{\infty} 
		\left\|\frac{1}{n}\sum_{i=1}^{n} \mathbf{B}(\bs{p}_i)\right\|_{\infty} 
		\left\|\bs{a}\right\|_{\infty} 
	\leq  C_{d} \|\bs{a}\|_{\infty}.
\end{align*}
\end{proof}

\begin{lemma}
\label{LEM:epshatorder} Under Assumptions (A2) and (A4)
\begin{equation}
	\left\|s_{0,\epsilon} \right\|_{\infty,\Omega}
	=O_P\left\{\frac{(\log n)^{1/2}}{\sqrt{n}|\triangle|^{3/2}}\right\}.
\label{EQ:epshatunifsize}
\end{equation}
\end{lemma}

\begin{proof} 
Note that $s_{0,\epsilon}(\bs{p})=\sum_{\xi \in \mathcal{M}}\widehat{c}_{0,\xi}B_{\xi}(\bs{p})$ for some coefficients $\widehat{c}_{0,\xi}$, so the order of $s_{\rho_n,\epsilon}(\bs{p})$ is related to that of $\widehat{c}_{0,\xi}$. In fact
\[
	\left\|s_{0,\epsilon}\right\|_{\infty,\Omega} 
	\leq  c \left\|\widehat{\mathbf{c}}_0\right\|_{\infty}
	= \left\| \mathbf{\Gamma}_{0}^{-1} \left[n^{-1} \sum_{i=1}^{n}B_{\xi}\left( \bs{p}_i\right) \sigma \left( \bs{p}_i\right) \varepsilon_i\right] _{\xi \in \mathcal{M}}\right\| _{\infty,\Omega},
\]
where $\widehat{\mathbf{c}}_{0}=(\widehat{c}_{0,\xi})_{\xi \in \mathcal{M}}$ and $\mathbf{\Gamma} _{0}$ is the symmetric positive definite matrix
\[
\left[ \frac{1}{n}\sum_{i=1}^{n}B_{\xi}\left(
\bs{p}_i\right) B_{\zeta }\left( \bs{p}_i\right)
\right] _{\xi,\zeta \in \mathcal{M}}
\]
defined in (\ref{DEF:Gamma_lambda}). Thus,
\[
	\left\|s_{0,\epsilon}\right\|_{\infty,\Omega} 
	\leq  CN \max_{\xi \in \mathcal{M}}\left| n^{-1}\sum_{i=1}^{n}B_{\xi}\left( \bs{p}_i\right) \sigma \left( \bs{p}_i\right) \varepsilon_i\right|,\ \ a.s.
\]
Next we show that with probability $1$
\begin{equation}
	\max_{\xi \in \mathcal{M}}\left|n^{-1}\sum_{i=1}^{n}
		B_{\xi}\left(\bs{p}_i\right) \sigma\left(\bs{p}_i\right)\varepsilon_i\right| 
	=O\left\{(\log n)^{1/2}(nN)^{-1/2}\right\}.  \label{EQ:BEsupnorm}
\end{equation}
To prove (\ref{EQ:BEsupnorm}), we decompose the noise variable $\varepsilon_i$ into a truncated part and a tail part 
$\varepsilon_i=\varepsilon_{i,1}^{D_n}+\varepsilon _{i,2}^{D_n}+\mu^{D_n}$, 
where 
$D_n=n^{\alpha}$ with {$1/{(2+\eta)}\leq\alpha<(1-\gamma^{-1})/2$}, 
$\varepsilon_{i,1}^{D_n}=\varepsilon_iI\left\{ \left| \varepsilon_i\right|>D_n\right\} $,
\[
	\varepsilon _{i,2}^{D_n}=\varepsilon_iI\left\{ \left| \varepsilon_i\right| \leq D_n\right\} -\mu^{D_n},\quad 
	\mu^{D_n}=\mathrm{E}\left[\varepsilon_iI\left\{ \left| \varepsilon_i\right| \leq D_n\right\} \right].
\]
It is straightforward to verify that the mean of the truncated part is
uniformly bounded by $D_n^{-2}$, so the boundedness of trivariate spline basis and
of the function $\sigma^2$ entail that
\begin{align}
	\left| B_{\xi}\left( \bs{p}_i\right)\sigma \left( \bs{p}_i\right) \mu^{D_n}\right| 
	=O\left( D_n^{-2}\right)
	=o\left( n^{-2/3}\right) .
\label{EQN:muo}
\end{align}
Next we show that tail part vanishes almost surely. Recall that $\mathrm{E}|\varepsilon_n^{2+\eta}|\leq\upsilon_{\eta}$ in Assumption (A2), so
\begin{equation}
	\sum_{n=1}^{\infty }P\left\{ \left| \varepsilon_n\right| >D_n\right\} 
	\leq\sum_{n=1}^{\infty}\frac{E\left|\varepsilon_n\right|^{2+\eta}}{D_n^{2+\eta}}
	\leq \upsilon_{\eta}\sum_{n=1}^{\infty}D_n^{-\left(2+\eta \right) }<\infty .  
\label{keyH}
\end{equation}
By the Borel-Cantelli Lemma, we have
\[
P\left\{ \omega |\exists N\left( \omega \right),\left| \varepsilon _n\left(
\omega \right) \right| \leq D_n \textrm{\ for\ } n>N\left( \omega \right)
\right\} =1.
\]
Let $\upsilon _{\varepsilon }=\max \{\left| \varepsilon _1\right|,\left|
\varepsilon _2\right|,...,\left| \varepsilon _{N\left( \omega \right) }\right|
\}$ and there exists $N_1\left( \omega \right) >N\left( \omega \right) $, $%
D_{N_1\left( \omega \right) }>\upsilon _{\varepsilon }$. Since $D_n=n^{\alpha }$ is an increasing function, we have $D_n>D_{N_1\left( \omega \right) }>\upsilon _{\varepsilon }$, for $n>N_1\left( \omega \right)$. Thus,
\[
P\left\{ \omega |\exists N\left( \omega \right),\left| \varepsilon_i\left(
\omega \right) \right| \leq D_n,\ 1\leq i\leq n\textrm{,\ for\ }%
n>N\left( \omega \right) \right\} =1,
\]
which implies that
\[
P\left\{ \omega |\exists N\left( \omega \right),\left| \varepsilon
_{i,1}^{D_n}\right| =0,\ 1\leq i\leq n\textrm{,\ for\ } n>N\left( \omega
\right) \right\} =1.
\]
The boundedness of the spline basis implies that
\begin{align}
	\left| \frac{1}{n}\sum_{i=1}^{n}\sigma \left( \bs{p}_i\right)
		\varepsilon _{i,1}^{D_n}B_{\xi}\left(\bs{p}_i\right) \right| 
	=O_{a.s.}\left( n^{-k}\right),\textrm{\ for\ any\ } k>0.
\label{EQN:epsi1o}
\end{align}
Next let $Z_i=n^{-1}\varepsilon _{i,2}^{D_n}\sigma \left( \bs{p}_i\right) B_{\xi}\left( \bs{p}_i\right) $. Since $\mathrm{E}\left( \varepsilon _{i,2}^{D_n}\right) =0$, we have
\[
	\mathrm{Var}\left( \varepsilon _{i,2}^{D_n}\right)  
	=\mathrm{E}\left( \varepsilon_i^2\right) -\mathrm{E}\left[ \varepsilon_i^2I\left\{ \left| \varepsilon_i\right| >D_n\right\} \right] -\left( \mu^{D_n}\right)^2 \\
	=1+O_P\left\{ D_n^{-\eta }+D_n^{-2\left( 1+\eta \right) }\right\} .
\]
According to (\ref{EQ:error-rate}), 
$\mathrm{E}\left[B_{\xi}^2\left(\bs{p}_i\right)\sigma^2 \left(\bs{p}_i\right)\right]
	\le C(\beta,\sigma)N^{-1}V_{\Omega}$. 
Note the independence between $\varepsilon_{i,2}^{D_n}$ and $B_{\xi}\left(\bs{p}_i\right)$ and the independence of $\varepsilon _{i,2}^{D_n},i=1,\cdots,n$, we have $V_n^2=\mathrm{Var}\left( \sum_{i=1}^{n}Z_i\right) =c\left( nN\right)
^{-1}$ for some $c>0$. By the fact that $| \varepsilon _{i,2}^{D_n}|<2D_n$, we have 
\[
	\mathrm{E}\left| \varepsilon _{i,2}^{D_n}\right|^k 
	\leq 2^{k-2}D_n^{k-2}\mathrm{E}\left| \varepsilon _{i,2}^{D_n}\right|^2,
	\quad k\geq 2.
\]
Note that
\begin{align*}
	\mathrm{E}\left| Z_i\right|^k 
	&= n^{-k}\mathrm{E}\left|\varepsilon_{i,2}^{D_n}\right|^k 
		\mathrm{E}\left|\sigma \left( \bs{p}_i\right) B_{\xi}\left( \bs{p}_i\right)\right|^k \\
	&\leq n^{-k}\mathrm{E}\left|\varepsilon_{i,2}^{D_n}\right|^k 
		C_{\sigma}^{k-2} \left\|B_{\xi}\right\|_{\infty,\Omega}^{k-2}
		\mathrm{E}\left|\sigma \left( \bs{p}_i\right)B_{\xi}\left( \bs{p}_i\right)\right|^2 
	\leq (2D_nn^{-1})^{k-2}k!\mathrm{E}(Z_i^2).
\end{align*}
Thus, $\{Z_i\}_{i=1}^n$ satisfies the Cramer condition with constant $c^{*}=2n^{-1}D_n$.
By the Bernstein inequality, for any $\delta >0$ large enough,
\[
	P\left(\left|\sum_{i=1}^{n}Z_i\right|\ge\delta\sqrt{\frac{\log n}{nN}}\right)
	\!\leq \! 2\exp \left( \frac{-\delta^2\frac{\log n}{nN}}{4V_n^2+2c^{*}\delta \sqrt{\frac{\log n}{nN}}}\right) 
	\! =2 \! \exp \left(\frac{-\delta^2\log n}{4c+4D_n\delta \sqrt{\frac{N\log n}{n}}}\right)
	\! \leq 2n^{-3}.
\]
Therefore,
\[
	\sum_{n=1}^{\infty }P\left( \max_{\xi \in \mathcal{M}}\left| \frac{1}{n}\sum_{i=1}^{n}\varepsilon_{i,2}^{D_n}\sigma\left(\bs{p}_i\right)B_{\xi}\left( \bs{p}_i\right) \right| \ge \delta \sqrt{\frac{\log n}{nN}}\right) 
	\leq \frac{N}{3}(d+1)(d+2)(d+3)\sum_{n=1}^{\infty }n^{-3}
	<\infty
\]
for such $\delta >0$. Thus, the conclusion follows by combining Borel-Cantelli's lemma with (\ref{EQN:muo}) and (\ref{EQN:epsi1o}).  
\end{proof}

\begin{proof}[Proof of Proposition \ref{PROP:maxnorm}] 
Note that the penalized spline $s_{\rho_n,\epsilon}$ of $\varepsilon$ is characterized by the orthogonality relations
\begin{eqnarray}
	n\left\langle \varepsilon -s_{\rho_n, \varepsilon},u\right\rangle _{n,\Omega}
	=\rho_n\left\langle s_{\rho_n, \varepsilon},u\right\rangle_{\mathcal{E}},
	\quad \textrm{for all }u\in \mathcal{S}_d^r(\triangle). 
\label{EQ:pls2}
\end{eqnarray}
In particular, $s_{0, \varepsilon}$ is characterized by
\begin{eqnarray}
	\left\langle \varepsilon -s_{0, \varepsilon},u\right\rangle _{n,\Omega}=0,
	\quad \textrm{for all }u\in \mathcal{S}_d^r(\triangle).  
\label{EQ:ls2}
\end{eqnarray}
Inserting $u=s_{0, \varepsilon}-s_{\rho_n,\epsilon}$ in (\ref{EQ:pls2}) and using (\ref{EQ:ls2}) with this $u$ yield that
\begin{eqnarray}
	n\left\| s_{0, \varepsilon}-s_{\rho_n, \varepsilon}\right\| _{n,\Omega}^2
	&=&\rho_n\left\langle s_{\rho_n, \varepsilon},s_{0, \varepsilon}-s_{\rho_n, \varepsilon}\right\rangle_{\mathcal{E}}  \nonumber \\
	&=&\rho_n (\langle s_{\rho_n,\epsilon}, s_{0,\epsilon}\rangle_{\mathcal{E}} -\langle s_{\rho_n,\epsilon}, s_{\rho_n,\epsilon}\rangle_{\mathcal{E}}) .  
\end{eqnarray}
It follows, by Cauchy-Schwarz inequality, that
\[
\left\| s_{\rho_n, \varepsilon}\right\|
_{\mathcal{E}}^2\leq \left\langle s_{\rho_n
,\varepsilon},s_{0, \varepsilon}\right\rangle _{\mathcal{E}}\leq \left\|
s_{\rho_n, \varepsilon}\right\| _{\mathcal{E}}\left\| s_{0, \varepsilon}
\right\| _{\mathcal{E}},
\]
which implies that $\left\| s_{\rho_n, \varepsilon}\right\|
_{\mathcal{E}}\leq \left\| s_{0, \varepsilon}\right\|
_{\mathcal{E}}$.
Thus, by Cauchy-Schwarz inequality and the definition of $\overline{V}_n$.
\[
	n\left\| s_{0, \varepsilon}-s_{\rho_n, \varepsilon}\right\| _{n,\Omega}^2
	\leq\rho_n \left\| s_{\rho_n, \varepsilon}\right\| _{\mathcal{E}}\left\| s_{0, \varepsilon}-s_{\rho_n, \varepsilon}\right\| _{\mathcal{E}}
	\leq \overline{V}_n\rho_n \left\| s_{\rho_n, \varepsilon}\right\|_{\mathcal{E}}\left\| s_{0, \varepsilon}-s_{\rho_n, \varepsilon}\right\|_{n,\Omega}.
\]
Hence, we have
\begin{eqnarray}
\left\| s_{0, \varepsilon}-s_{\rho_n, \varepsilon}\right\| _{n,\Omega}\leq
n^{-1}\overline{V}_n \rho_n \left\| s_{0, \varepsilon}\right\|
_{\mathcal{E}}.
\label{EQ:pls-ps3}
\end{eqnarray}
Combining (\ref{DEF:Vn}) and (\ref{EQ:pls-ps3}) yields that
$$
\left\| s_{0, \varepsilon}-s_{\rho_n, \varepsilon}\right\| _{\infty,\Omega}\leq V_n
\left\| s_{0, \varepsilon}-s_{\rho_n, \varepsilon}\right\| _{n,\Omega}\leq
n^{-1}V_n \overline{V}_n \rho_n \left\| s_{0, \varepsilon}\right\|
_{\mathcal{E}}.
$$
Finally, we use Markov's inequality to get
$$
\left\| s_{0, \varepsilon}\right\|_{\mathcal{E}} \le \frac{C_1}{|\triangle|^2} \left\|s_{0, \varepsilon}
\right\|_{L^2(\Omega)} .
$$
It therefore follows
\[
	\|s_{\rho_n, \varepsilon}\|_{\infty,\Omega} 
	\le \|s_{0,\epsilon}\|_{\infty,\Omega} 
		+ \|s_{0, \varepsilon}- s_{\rho_n,\epsilon}\|_{\infty,\Omega} 
	\le\|s_{0, \varepsilon}\|_{\infty,\Omega}
		+ \frac{\rho_n}{n} V_n \bar{V}_n \frac{C_1}{|\triangle|^2}
		\|s_{0, \varepsilon}\|_{L^2(\Omega)}.
\]
According to Lemma \ref{LEM:epshatorder}, we have $\|s_{0, \varepsilon}\|_{\infty,\Omega}=O_P\left\{ n^{-1/2}|\triangle|^{-3/2} (\log n)^{1/2}\right\}$. 
The conclusion of Proposition \ref{PROP:maxnorm} follows from Proposition \ref{PROP:variance}.

\end{proof}

\vskip .10in \noindent \textbf{B.3. Variance of the noise term} \vskip .10in
\label{SSEC:VarNoise}

In this section, we first derive the size of the asymptotic conditional variance given in Theorem \ref{THM:varianceorder}.

\begin{theorem}
\label{THM:varianceorder} 
Under Assumptions (A1) -- (A4), if $\rho_n=o(n/N^2)$, we have with probability approaching one as $n \rightarrow \infty$
\[
	\frac{C_{1}c_{\sigma}^{2}}{n (1+n^{-1}|\triangle|^{-4}\rho_n)^2|\triangle |^3}
	\le \mathrm{Var}\left\{s_{\rho_n,\epsilon}(\bs{p})|\mathbb{P}\right\}
	\le \frac{C_{2}C_{\sigma}^{2}}{n|\triangle|^3},~\bs{p}\in \Omega,
\]
for positive constants $C_{1}$ and $C_{2}$.
\end{theorem}

\begin{proof} 
Note that
\begin{eqnarray}
s_{\rho_n,\epsilon}(\bs{p})=\mathbf{B}
(\bs{p})^{\top}\widehat{\mathbf{c}}_{\rho_n,\varepsilon},
\label{EQ:espilon}
\end{eqnarray}
where $\widehat{\mathbf{c}}_{\rho_n,\varepsilon }$ is the coefficient vector for $%
s_{\rho_n,\epsilon}$ using basis functions $B_{\xi}$, $\xi \in
\mathcal{M}$ and $\mathbf{B}(\bs{p})=[B_{\xi}(\bs{p})$, $\xi
\in \mathcal{M}]^{\top}$ is the vector of basis functions. Note that
\[
	\widehat{\mathbf{c}}_{\rho_n,\varepsilon }
	=\left[\sum_{i=1}^{n}B_{\xi}\left( \bs{p}_i\right) B_{\zeta }\left( \bs{p}_i\right) +\rho_n\langle B_{\xi},B_{\zeta }\rangle _{\mathcal{E}}\right] _{\xi,\zeta \in\mathcal{M}}^{-1}
		\left[ \sum_{i=1}^{n}B_{\xi}\left( \bs{p}_i\right) \sigma \left( \bs{p}_i\right) \varepsilon_i\right] _{\xi \in \mathcal{M}}
\]
and $\mathrm{Var}\{s_{\rho_n,\epsilon}(\bs{p})|\mathbb{P}\}
=\mathbf{B}(\bs{p})^{\top}\mathrm{E}\left( \widehat{\mathbf{c}}_{\rho_n,\varepsilon }\widehat{\mathbf{c}}_{\rho_n,\varepsilon }^{\top}|\mathbb{P}\right) \mathbf{B}(\bs{p})$.
Recall the definition of $\bs{\Gamma}_{\rho_n}$ in (\ref{DEF:Gamma_lambda}) in main part,
we have $\mathrm{E}\left( \widehat{\mathbf{c}}_{\rho_n,\varepsilon }\widehat{\mathbf{c}}_{\rho_n
,\varepsilon }^{\top}| \mathbb{P}\right)$ equal to
\[
\bs{\Gamma}_{\rho_n}^{-1}\mathrm{E}\left( \left. \left[ \frac{1}{n}%
\sum_{i=1}^{n}B_{\xi}\left( \bs{p}_i\right)\sigma \left(
\bs{p}_i\right) \varepsilon_i\right] _{\xi \in
\mathcal{M}}\left[ \frac{1}{n}\sum_{i=1}^{n}B_{\xi}\left(
\bs{p}_i\right) \sigma \left(
\bs{p}_i\right)\varepsilon_i\right] _{\xi \in \mathcal{M}}^{\top}\right| %
\mathbb{P}\right) \bs{\Gamma}_{\rho_n}^{-1}.
\]
The central conditional expectation term in the above line
satisfies that
\[
	\frac{c_{\sigma}^2}{n}\mathbf{\Gamma}_0
	\leq \frac{1}{n^2}\left[ \sum_{i=1}^{n}B_{\xi}\left( \bs{p}_i\right) B_{\zeta }\left( \bs{p}_i\right) \sigma^2 \left(\bs{p}_i\right)\right] _{\xi,\zeta \in \mathcal{M}}
	\leq\frac{C_{\sigma}^2}{n}\mathbf{\Gamma}_0.
\]
That is,
\[
	n^{-1}c_{\sigma}^2\mathbf{B}(\bs{p})^{\top}\bs{\Gamma}_{\rho_n}^{-1}\mathbf{\Gamma}_0\bs{\Gamma}_{\rho_n}^{-1}\mathbf{B}(\bs{p}) 
	\leq\mathrm{Var}\left\{s_{\rho_n,\epsilon}\left(\bs{p}\right)|\mathbb{P}\right\}
	\leq n^{-1}C_{\sigma}^2\mathbf{B}(\bs{p})^{\top}\bs{\Gamma}_{\rho_n}^{-1}\mathbf{\Gamma}_0\bs{\Gamma}_{\rho_n}^{-1}\mathbf{B}(\bs{p}).
\]
Let $\alpha_{\min}(\rho_n )$ and $\alpha_{\max}(\rho_n )$ be the
smallest
and largest eigenvalues of the positive definite matrix $\bs{\Gamma}_{\rho_n}$%
. It follows easily that with probability approaching one,
\[
	\frac{c_{\sigma}^2}{n} \alpha_{\max}(\rho_n )^{-2}\alpha_{\min}(0)\Vert \mathbf{B}(\bs{p})\Vert^2
	\le \mathrm{Var}\left\{ s_{\rho_n,\epsilon}(\bs{p})|\mathbb{P}\right\} 
	\le \frac{C _{\sigma}^2}{n} \alpha_{\min}(\rho_n )^{-2}\alpha_{\max}(0)\Vert \mathbf{B}(\bs{p})\Vert^2 .
\]
Note that $\Vert \mathbf{B}(\bs{p})\Vert^2=\sum_{\xi \in \mathcal{M}}B_{\xi}^2(\bs{p})$ is bounded above by a constant $C_2<\infty $ and below by $C_1>0$ for any point $\bs{p}\in \Omega $. Indeed, if $C_1=0$ for a point $\bs{p}$, then $B_{\xi}(\bs{p})=0$ for all $\xi \in \mathcal{M}$. That is, we have $0=\sum_{\xi \in \mathcal{M}}B_{\xi}(\bs{p})$. It follows that these basis functions are linearly dependent which is a contradiction. On the other hand, for any fixed point $\bs{p}\in \Omega$, there is at most $C_2$ nonzero terms in the above summation, where $C_2$ is dependent on the smallest angle of the triangulation $\triangle$. Note that $B_{\xi}$ is uniformly bounded for all $\xi \in \mathcal{M}$, say bounded by $1$. Thus, we know $C_2<\infty $. We summarize the above discussion to get
\[
	n^{-1}C_1c_{\sigma}^2\alpha_{\max}(\rho_n )^{-2}\alpha_{min}(0)
	\le \mathrm{Var}\left\{s_{\rho_n,\epsilon}(\bs{p})|\mathbb{P}\right\} 
	\le n^{-1}C_2C_{\sigma}^2\alpha_{\min}(\rho_n )^{-2}\alpha_{\max}(0)
\]
with probability approaching one.

We now spend some effort to estimate the largest and smallest eigenvalues of $\bs{\Gamma}_{\rho_n}$. It is easy to see that for any vector $\mathbf{a}=[a_{\xi},\xi \in \mathcal{M}]^{\top}$,
\begin{eqnarray*}
	\mathbf{a}^{\top}\bs{\Gamma}_{\rho_n}\mathbf{a} 
	&=&\mathbf{a}^{\top}\left[ \frac{1}{n}\sum_{i=1}^n B_{\xi}\left( \bs{p}_i\right) B_{\zeta }\left( \bs{p}_i\right) +\frac{\rho_n }{n}\langle B_{\xi},B_{\zeta }\rangle _{\mathcal{E}}\right] _{\xi,\zeta \in \mathcal{M}}\mathbf{a} \\
	&=&\left[ \frac{1}{n}\sum_{i=1}^{n}\left\{ \sum_{\xi \in\mathcal{M}}a_{\xi}B_{\xi}\left( \bs{p}_i\right) \right\}^2\right] +\frac{\rho_n }{n}\mathcal{E}\left(\sum_{\xi \in \mathcal{M}}a_{\xi}B_{\xi}\right).
\end{eqnarray*}

Let $s_a=\sum_{\xi \in \mathcal{M}}a_{\xi}B_{\xi}\in \mathcal{S}_d^r(\triangle)$
be the spline associated with vector $\mathbf{a}=(a_\xi, \xi\in
\mathcal{M})^T$. By (\ref{EQ:normratio}), we have
\[
	\Vert s_a\Vert _{n,\Omega}^2 
	=\left\{ 1+O_P\left(\sqrt{\frac{\log n}{n/N}}\right) \right\} \Vert s_a\Vert _{L^2(\Omega)}^2
	\le K_2\left\{ 1+O_P\left( \sqrt{\frac{\log n}{n/N}}\right) \right\} |\triangle|^3\Vert \mathbf{a}\Vert^2.
\]
Here we have used the stability conditions in Lemma \ref{LEM:normequity}. Furthermore, using Markov's inequality, we have
\[
	\frac{\rho_n}{n}\mathcal{E}\left(\sum_{\xi \in\mathcal{M}}a_{\xi}B_{\xi}\right)
	\le \frac{\rho_n}{n}\frac{C}{|\triangle|^4}\Vert s_a\Vert _{L^2(\Omega)}^2
	\le\frac{\rho_n }{n}\frac{C}{|\triangle|}K_2\Vert\mathbf{a}\Vert^2.
\]
Thus, the largest eigenvalue $\alpha_{\max}(\rho_n )$ of the matrix $\bs{\Gamma}_{\rho_n}$ in (\ref{DEF:Gamma_lambda}) is less than or equal to
\[
	K_2\left\{1+O_P\left( \sqrt{\frac{\log n}{n/N}}\right)\right\} |\triangle|^3
		+K_2\frac{\rho_n}{n}\frac{C}{|\triangle|}.
\]
Thus, we have with probability approaching one
\[
	\left\{
	\begin{array}{ll}
		\alpha_{\max}(0)\le C|\triangle|^3,\quad  & \hbox{ if }\rho_n =0 \\
		\alpha_{\max}(\rho_n )\le C \left(|\triangle|^3+\frac{\rho_n }{n|\triangle|}\right), & \hbox{ if }\rho_n >0
	\end{array}
	\right.
\]
for positive constant $C$. On the other hand, we use Lemma \ref{LEM:normequity} and (\ref{EQ:normratio}) to have
\[
	\Vert s_a\Vert _{n,\Omega}^2 
	=\left\{ 1+O_P\left(\sqrt{\frac{\log n}{n/N}}\right) \right\} \Vert s_a\Vert _{L^2(\Omega)}^2 
	\ge K_1\left\{ 1+O_P\left( \sqrt{\frac{\log n}{n/N}}\right)\right\} |\triangle|^3\Vert \mathbf{a}\Vert^2.
\]
Therefore, the smallest eigenvalue $\alpha_{\min}(\rho_n )$ of the matrix $\bs{\Gamma}_{\rho_n}$ in (\ref{DEF:Gamma_lambda}) is greater than 
$K_1\{ 1+O_P( \sqrt{Nn^{-1}\log n})\} |\triangle|^3
	=C|\triangle|^3 $. 
Summarizing the above discussions to conclude that for $\rho_n
=0$, we have with probability approaching one,
\[
	c_{\beta}c_{\sigma}^2n^{-1}|\triangle|^{-3}
	\le\mathrm{Var}\left\{ s_{\rho_n,\epsilon}(\bs{p})|\mathbb{P}\right\}
	\le C_{\beta}C_{\sigma}^2n^{-1}|\triangle|^{-3}.
\]
This establishes the result in the case for $\rho_n=0$. Similar for the case $\rho_n >0$. We have therefore completed the proof.
\end{proof}

The above variance result can be more precise when spline space $\mathcal{S}^{-1}_0(\triangle)$ is considered. The next lemma provides the pointwise variance of $\widehat{s}_{\rho_n,\varepsilon}$ when using
spline space $\mathcal{S}^{-1}_0(\triangle)$.

\begin{lemma}
\label{LEM:var-cons}
Consider piecewise constant spline space $\mathcal{S}^{-1}_0(\triangle)$ and suppose that the tetrahedra are of equal size. Under Assumptions (A1)--(A4), if the density function $f(\bs{p})$ of $\bs{p}$ is continuous and positive on $\Omega$, then
\[
	\mathrm{Var}\left\{s_{\rho_n,\epsilon}\left( \bs{p}\right)\right\}
	=\frac{1}{n}\frac{\sigma^2\left( \bs{p}\right)}{f\left( \bs{p} \right) V_T}\{1+o(1)\},
\]
where $V_T$ is the volume of the tetrahedron $T$ as defined before.
\end{lemma}

\begin{proof} 
When using spline space $\mathcal{S}^{-1}_0(\triangle)$, the space of piecewise constant functions over $\triangle$, we have ${\cal E}(s)=0$ for all $s\in\mathcal{S}^{-1}_0(\triangle)$. Thus,
following (\ref{EQ:espilon}), we have $s_{\rho_n,\varepsilon
}\left( \bs{p}\right) =\mathbf{B}
(\bs{p})^{\top}\widehat{\mathbf{c}}_{\rho_n,\varepsilon }$, where 
\[
	\widehat{\mathbf{c}}_{\rho_n,\varepsilon}
	=\left[\left\| B_{\xi}\right\| _{n,\Omega}^{-2}\frac{1}{n}\sum_{i=1}^{n}\varepsilon_iB_{\xi}\left( \bs{p}_i\right) \sigma \left(\bs{p}_i\right) \right] _{\xi \in \mathcal{M}}.
\]
Let $\tilde{s}_{\rho_n,\varepsilon }\left( \bs{p}\right) =\mathbf{B}
(\bs{p})^{\top}\tilde{\mathbf{c}}_{\rho_n,\varepsilon }$, where
\[
\tilde{\mathbf{c}}_{\rho_n,\varepsilon }=\left[ \left\| B_{\xi
}\right\| _{L^2(\Omega)}^{-2}\frac{1}{n}\sum_{i=1}^{n}\varepsilon
_iB_{\xi}\left( \bs{p}_i\right) \sigma \left(
\bs{p}_i\right) \right] _{\xi \in \mathcal{M}}.
\]
For any $\bs{p}\in \Omega$,
\[
\left| s_{\rho_n,\epsilon}\left( \bs{p}\right) -\tilde{s}%
_{\rho_n,\varepsilon }\left( \bs{p}\right) \right| \leq \left| \tilde{s}%
_{\rho_n,\varepsilon }\left( \bs{p}\right) \right| \max_{\xi \in \mathcal{%
M}}\left|{\left\| B_{\xi}\right\| _{L^2(\Omega)}^2}/{\left\|
B_{\xi}\right\| _{n,\Omega}^2}-1\right|.
\]
According to Lemma \ref{LEM:Rnorder}, we have
\[
\max_{\xi \in \mathcal{M}}\left| \frac{\left\| B_{\xi}\right\|
 _{L^2(\Omega)}^2}{\left\| B_{\xi}\right\| _{n,\Omega}^2}-1\right|\leq \frac{R_n}{%
1-R_n}=O_P\left\{ \sqrt{{N\log n}/{n}}\right\}.
\]
Thus,
\[
\left| s_{\rho_n,\epsilon}\left( \bs{p}\right) -\tilde{s}%
_{\rho_n,\varepsilon }\left( \bs{p}\right) \right|\leq
O_P\left\{ \left( N\log n/n\right)^{1/2} \left| \tilde{s}%
_{\rho_n,\varepsilon }\left( \bs{p}\right) \right| \right\}.
\]
Hence, finding the asymptotic variance of $\widehat{s}_{\rho_n,\varepsilon }(\bs{p})$ is equivalent to finding the asymptotic variance of $\tilde{s}_{\rho_n,\varepsilon }(\bs{p})$. Next we calculate the pointwise variance of
$\tilde{s}_{\rho_n,\varepsilon }(\bs{p})$. Note that
\begin{eqnarray*}
	\mathrm{Var}\left\{\tilde{s}_{\rho_n,\varepsilon }\left(\bs{p}\right) \right\} 
	&=& \mathbf{B}(\bs{p})^{\top}\mathrm{E}\left(\tilde{\mathbf{c}}_{\rho_n,\varepsilon }\tilde{\mathbf{c}}_{\rho_n,\varepsilon }^{\top}\right) \mathbf{B}(\bs{p}) \\
	&=& \sum_{\xi \in\mathcal{M}}B_{\xi}^2\left( \bs{p}\right) 
		\mathrm{E}\left\{\frac{1}{n}\sum_{i=1}^{n}\varepsilon_i
		B_{\xi}\left(\bs{p}_i\right) \sigma \left( \bs{p}_i\right) 
		\left\|B_{\xi}\right\| _{L^2(\Omega)}^{-2}\right\}^2.
\end{eqnarray*}
For any $\bs{p}\in \Omega$, let $T_{\xi(\bs{p})}$ be
the tetrahedron that contains $\bs{p}$. It is easy to see that
\[ 
	\left\| B_{\xi}\right\| _{L^2(\Omega)}^2
	=\int_{\Omega}B_{\xi}^2\left(\bs{p}\right)f\left(\bs{p}\right)\mathrm{d}\bs{p}
	=\int_{T _{\xi(\bs{p})}}f\left(\bs{p}\right) \mathrm{d}\bs{p}.
\]
For any continuous function $g$, let 
$\omega (g,\varrho )
	=\sup_{\bs{p},\bs{p}^{\prime}\in\Omega,\Vert\bs{p}-\bs{p}^{\prime}\Vert \leq \varrho}
	|g(\bs{p})-g(\bs{p}^{\prime })|$ 
be the moduli of continuity of $g$ on $\Omega$. Then for any $\bs{p}^{\prime}\in T _{\xi (\bs{p})}$, we can write 
$\{ f\left( \bs{p}^{\prime}\right) -\omega (f, d_T)\} V_T
	\leq \left\| B_{\xi}\right\| _{L^2(\Omega)}^2
	\leq\{ f\left( \bs{p}^{\prime}\right) +\omega (f, d_T)\} V_T$ 
with $d_T$ being the diameter of the smallest ball containing $T$.
Next, we have
\begin{eqnarray*}
	n\mathrm{Var}\left\{\tilde{s}_{\rho_n,\varepsilon }\left(\bs{p}\right) \right\}  
	&=&n\sum_{\xi\in\mathcal{M}}\left\|B_{\xi}\right\|_{L^2(\Omega)}^{-4}
		B_{\xi}^2\left( \bs{p}\right) \mathrm{E}\left\{\frac{1}{n}\sum_{i=1}^{n}
		\varepsilon_iB_{\xi}\left(\bs{p}_i\right)\sigma\left(\bs{p}_i\right)\right\}^2 \\
	&=&\sum_{\xi\in\mathcal{M}}\left\|B_{\xi}\right\|_{L^2(\Omega)}^{-4}
		B_{\xi}^2\left(\bs{p}\right)\mathrm{E}\left\{B_{\xi}^2\left( \bs{p}\right) 
		\sigma^2\left( \bs{p}\right)\right\}\\
	&=&\sum_{\xi \in \mathcal{M}}\left\| B_{\xi}\right\|_{L^2(\Omega)}^{-4}B_{\xi}^2\left( \bs{p}\right) \int_{\Omega}B_{\xi}^2\left(\bs{p}^{\prime}\right) \sigma^2\left( \bs{p}^{\prime}\right) f\left( \bs{p}^{\prime}\right) \mathrm{d}\bs{p}^{\prime}.
\end{eqnarray*}
By the continuity of functions $\sigma^2(\bs{p})$ and $f(\bs{p})$, we have
\begin{eqnarray*}
	\mathrm{Var}\left\{\tilde{s}_{\rho_n,\varepsilon }\left(\bs{p}\right) \right\}  
	&=&\frac{1}{n}\sum_{\xi \in\mathcal{M}}
		\left\|B_{\xi}\right\|_{L^2(\Omega)}^{-4}B_{\xi}^2\left(\bs{p}\right)\\
	&& \times\left[ \sigma^2\left(\bs{p}\right) f\left( \bs{p}\right) V_T
		+\int_{{T_{\xi\left(\bs{p}\right)}}}
		\left\{\sigma^2\left( \bs{p}^{\prime}\right) f\left( \bs{p}^{\prime}\right) 
		-\sigma^2\left(\bs{p}\right)f\left( \bs{p}\right) \right\} 
		\mathrm{d}\bs{p}^{\prime}\right]  \\
	&\leq&\frac{1}{n}
		\frac{\left\{\sigma^2\left(\bs{p}\right)f\left(\bs{p}\right) 
		+\omega (\sigma^2f,|T |)\right\}V_T}
		{\left[ \left\{ f\left( \bs{p}\right) -\omega(f,|T |)\right\} V_T\right]^2}
	=\frac{1}{n}\frac{\sigma^2\left(\bs{p}\right)}{f\left(\bs{v}\right)V_T}\{1+o(1)\}.
\end{eqnarray*}
The conclusion follows.
\end{proof}

\vskip .10in \noindent \textbf{B.4. Proof of Theorem \ref{THM:normality}} \vskip .10in

\begin{proof}[Proof of Theorem \ref{THM:normality}]
To prove Theorem \ref{THM:normality}, we first show that under Assumptions (A1)--(A4), 
\begin{equation}
	\frac{s_{\rho_n,\epsilon}\left(\bs{p}\right)}
		{\sqrt{\mathrm{Var}\left\{s_{\rho_n,\epsilon}\left(\bs{p}\right) |\mathbb{P}\right\}}}
		\stackrel{D}{\longrightarrow} N\left( 0,1\right), \quad n \to \infty.
\label{EQ:normality}
\end{equation}

Let $\{B_{\xi}\}_{\xi \in \mathcal{M}}$ be the constructed spline basis functions for $\mathcal{S}$, where $\mathcal{M}$ stands for the index set for spline bases.
Note that $s_{\rho_n,\epsilon}(\bs{p}) = \mathbf{B}(\bs{p})^{\top}\widehat{\mathbf{c}}_{\rho_n,\varepsilon}$,
where $\widehat{\mathbf{c}}_{\rho_n,\varepsilon }$ is the coefficient vector for $s_{\rho_n,\epsilon}$ with basis functions $B_{\xi}$, $\xi \in
\mathcal{M}$ and $\mathbf{B}(\bs{p})=[B_{\xi}(\bs{p})$, $\xi
\in \mathcal{M}]^{\top}$ is the vector of basis functions. Then, we have
\begin{equation*}
	\widehat{\mathbf{c}}_{\rho_n,\varepsilon }
	=\left[\sum_{i=1}^{n}B_{\xi}\left(\bs{p}_i\right) B_{\zeta }\left(\bs{p}_i\right) +\rho_n\langle B_{\xi},B_{\zeta }\rangle _{\mathcal{E}}\right] _{\xi,\zeta \in\mathcal{M}}^{-1}
		\left[\sum_{i=1}^{n}B_{\xi}\left(\bs{p}_i\right) \sigma \left(\bs{p}_i\right) \varepsilon_i\right] _{\xi \in \mathcal{M}}.
\end{equation*}
Recall the definition of $\bs{\Gamma}_{\rho_n}$ in (\ref{DEF:Gamma_lambda}), then we have
\begin{eqnarray*}
	s_{\rho_n,\epsilon}\left(\bs{p}\right)
	&=&\mathbf{B}(\bs{p})^{\top}\bs{\Gamma}_{\rho_n}^{-1}
		\left[\frac{1}{n}\sum_{i=1}^{n}B_{\xi}\left(\bs{p}_i\right)
		\sigma\left(\bs{p}_i\right)\varepsilon_i\right]_{\xi\in\mathcal{M}} 
	=\frac{1}{n}\sum_{i=1}^{n}\mathbf{B}(\bs{p})^{\top}\bs{\Gamma}_{\rho_n}^{-1}\mathbf{B} (\bs{p}_i)\sigma \left(
\bs{p}_i\right)\varepsilon_i.
\end{eqnarray*}
Let $a_i=n^{-1}\mathbf{B}(\bs{p})^{\top}\bs{\Gamma}_{\rho_n}^{-1}\mathbf{B}(\bs{p}_i)\sigma \left( \bs{p}_i\right)$. Obviously,
$s_{\rho_n,\epsilon}\left(\bs{p}\right)=\sum_{i=1}^{n}a_i\varepsilon_i$, and
\[
	a_i^2=n^{-2}\mathbf{B}(\bs{p})^{\top}\bs{\Gamma}_{\rho_n}^{-1}
		\mathbf{B} (\bs{p}_i)\mathbf{B} (\bs{p}_i)^{\top}\bs{\Gamma}_{\rho_n}^{-1}
		\mathbf{B}(\bs{p})\sigma^2 \left( \bs{p}_i\right).
\]
For any vector  ${\bf c}=(c_\xi, \xi\in {\cal M})^T$, we have
\begin{eqnarray*}
	{\bf c}^{\top} \mathbf{B} (\bs{p}_i)\mathbf{B} (\bs{p}_i)^{\top} {\bf c} 
	=\left\{\sum_{\xi\in {\cal M}} c_\xi B_\xi(\bs{p}_i)\right\}^2 
	\le \|{\bf c}\|_2^2\sum_{\xi\in {\cal M}} B^2_\xi(\bs{p}_i),
\end{eqnarray*}
which means the maximal eigenvalue of the matrix $\mathbf{B}(\bs{p}_i)\mathbf{B}(\bs{p}_i)^{\top}$ is bounded by
$\sum_{\xi\in {\cal M}} B_\xi^2(\bs{p}_i)$. Consequently, according to the proof of Theorem \ref{THM:varianceorder}  in the supplementary material (Section B), with probability approaching one,
\begin{eqnarray*}
	a_i^2 
	&\leq& n^{-2}\sum_{\xi\in {\cal M}} B^2_\xi(\bs{p}_i)\mathbf{B}(\bs{p})^{\top}
		\bs{\Gamma}_{\rho_n}^{-1}\bs{\Gamma}_{\rho_n}^{-1}
		\mathbf{B}(\bs{p}) \sigma^2 \left( \bs{p}_i\right) \\
	&\leq & \frac{C_{\sigma}^2}{n^2\left| \triangle \right|^6}\sum_{\xi\in {\cal M}} B^2_\xi(\bs{p}_i) \mathbf{B}(\bs{p})^{\top}\mathbf{B}(\bs{p})
	=\frac{C_{\sigma}^2}{n^2\left| \triangle \right|^6}\sum_{\xi\in \mathcal{M}}B_{\xi}^2\left( \bs{p}\right) \sum_{\xi\in{\cal M}} B^2_\xi(\bs{p}_i).
\end{eqnarray*}

On the other hand, 
$\sum a_i^2=\mathrm{Var}\{s_{\rho_n,\epsilon}(\bs{p})|\mathbb{P}\}$. 
Using similar arguments in the proof of Theorem
\ref{THM:varianceorder}  in the supplementary material (Section B), for $\rho_n =0$, we have 
$\sum a_i^2\geq c_{\sigma}^2n^{-1}|\triangle|^{-3}
	\sum_{\xi\in\mathcal{M}}B_{\xi}^2(\bs{p})$ 
with probability approaching one. Therefore, there exists some positive constant $C$, such that
\[
	\frac{\displaystyle \max_{1\leq i\leq n}a_i^2}{\sum a_i^2}
	\leq \frac{C C_{\sigma}^2}{n\left| \triangle\right|^3c_{\sigma}^2}\sum_{\xi \in \mathcal{M}}B_{\xi}^2\left( \bs{p}_i\right) 
	=O_P\left( Nn^{-1} \right) 
	=o_P\left( 1\right) .
\]
Hence,  
$\left. \sum a_i\varepsilon_i\right/ \left( \sum a_i^2\right)^{1/2}
	\rightarrow N\left( 0,1\right) $ 
by Linderberg-Feller CLT. Similarly for $\rho_n >0$, we have with probability approaching one,
\[
	\sum a_i^2 \geq \frac{c_{\sigma}^2}
		{n\left(1+n^{-1}|\triangle|^{-4}\rho_n\right )^2|\triangle|^3}
		\sum_{\xi \in \mathcal{M}}B_{\xi}^2\left( \bs{p}\right),
\]
and
\begin{eqnarray*}
	\frac{\displaystyle \max_{1\leq i\leq n}a_i^2}{\sum a_i^2}
	\leq\frac{CC_{\sigma}^2\left(1+n^{-1}|\triangle|^{-4}\rho_n\right)^2}
		{n\left| \triangle\right|^3c_{\sigma}^2}
		\sum_{\xi \in \mathcal{M}}B_{\xi}^2 \left(\bs{p}_i\right) 
	=O_P\left( Nn^{-1} \right) 
	=o_P\left(1\right), 
\end{eqnarray*}
which gives (\ref{EQ:normality}).

Note that under Assumptions (A3), (A4') and (A5), the bias term in (\ref{EQ:decomposition}) is negligible compared to the order of 
$[\mathrm{Var}\{s_{\rho_n,\epsilon}(\bs{p})|\mathbb{P}\}]^{1/2}$ given in Theorem \ref{THM:varianceorder}  in the supplementary material (Section B). Thus, Theorem \ref{THM:normality} follows directly from (\ref{EQ:normality}).
\end{proof}

\bibliographystyle{asa}
\bibliography{reference}

\begin{thebibliography}{52}
\newcommand{\enquote}[1]{``#1''}
\expandafter\ifx\csname natexlab\endcsname\relax\def\natexlab#1{#1}\fi

\bibitem[{Arnone et~al.(2022)Arnone, Kneip, Nobile, and Sangalli}]{Arnone:22}
Arnone, E., Kneip, A., Nobile, F., and Sangalli, L.~M. (2022), \enquote{Some
  first results on the consistency of spatial regression with partial
  differential equation regularization,} \textit{Statistica Sinica}, 32,
  209--238.

\bibitem[{Awanou et~al.(2005)Awanou, Lai, and Wenston}]{Awanou:Lai:Wenston:05}
Awanou, G., Lai, M.~J., and Wenston, P. (2005), \enquote{The multivariate
  spline method for scattered data fitting and numerical solutions of partial
  differential equations,} \textit{Wavelets and splines: Athens 2005}, 24--74.

\bibitem[{Berger et~al.(2017)Berger, Tagliasacchi, Seversky, Alliez,
  Guennebaud, Levine, Sharf, and Silva}]{Berger:etal:17}
Berger, M., Tagliasacchi, A., Seversky, L.~M., Alliez, P., Guennebaud, G.,
  Levine, J.~A., Sharf, A., and Silva, C.~T. (2017), \enquote{A survey of
  surface reconstruction from point clouds,} \textit{Computer Graphics Forum},
  36, 301--329.

\bibitem[{Bosq(1998)}]{Bosq:98}
Bosq, D. (1998), \textit{Nonparametric statistics for stochastic processes},
  Springer-Verlag, New York.

\bibitem[{Chung et~al.(2018)Chung, Wang, and Wu}]{chung2018discrete}
Chung, M.~K., Wang, Y., and Wu, G. (2018), \enquote{Discrete heat kernel
  smoothing in irregular image domains,} in \textit{2018 40th Annual
  International Conference of the IEEE Engineering in Medicine and Biology
  Society (EMBC)}, pp. 5101--5104.

\bibitem[{DeVore and Lorentz(1993)}]{DeVore:Lorentz:93}
DeVore, R.~A. and Lorentz, G.~G. (1993), \textit{Constructive approximation},
  vol. 303, Berlin: Springer-Verlag.

\bibitem[{Duchon(1976)}]{Duchon:77}
Duchon, J. (1976), \enquote{Fonctions-spline {\`a} {\'e}nergie invariante par
  rotation,} Tech. Rep.~27, Universit{\'e} de Grenoble.

\bibitem[{Eilers and Marx(1996)}]{eilers1996flexible}
Eilers, P.~H. and Marx, B.~D. (1996), \enquote{Flexible smoothing with
  B-splines and penalties,} \textit{Statistical science}, 11, 89--121.

\bibitem[{Fang and Boas(2009)}]{Fang:2009}
Fang, Q. and Boas, D.~A. (2009), \enquote{Tetrahedral mesh generation from
  volumetric binary and grayscale images,} in \textit{2009 IEEE International
  Symposium on Biomedical Imaging: From Nano to Macro}, pp. 1142--1145.

\bibitem[{Goldsmith et~al.(2014)Goldsmith, Huang, and
  Crainiceanu}]{goldsmith2014smooth}
Goldsmith, J., Huang, L., and Crainiceanu, C.~M. (2014), \enquote{Smooth
  scalar-on-image regression via spatial Bayesian variable selection,}
  \textit{Journal of Computational and Graphical Statistics}, 23, 46--64.

\bibitem[{Green and Silverman(1994)}]{Green:Silverman:94}
Green, P.~J. and Silverman, B.~W. (1994), \textit{Nonparametric regression and
  generalized linear models: a roughness penalty approach}, Chapman and Hall,
  London.

\bibitem[{Guo et~al.(2010)Guo, Duan, and Shum}]{guo2010non}
Guo, J., Duan, X., and Shum, C. (2010), \enquote{Non-isotropic Gaussian
  smoothing and leakage reduction for determining mass changes over land and
  ocean using {GRACE} data,} \textit{Geophysical Journal International}, 181,
  290--302.

\bibitem[{Hall and Horowitz(2013)}]{Hall:Horowitz:13}
Hall, P. and Horowitz, J. (2013), \enquote{A simple bootstrap method for
  constructing nonparametric confidence bands for functions,} \textit{The
  Annals of Statistics}, 41, 1892--1921.

\bibitem[{Han et~al.(2011)Han, Kamber, and Pei}]{han2011data}
Han, J., Kamber, M., and Pei, J. (2011), \enquote{Data Mining: Concepts and
  Techniques Third Edition [M],} \textit{The Morgan Kaufmann Series in Data
  Management Systems}, 5, 83--124.

\bibitem[{Huang(2003)}]{Huang:03b}
Huang, J.~Z. (2003), \enquote{Local asymptotics for polynomial spline
  regression,} \textit{The Annals of Statistics}, 31, 1600--1635.

\bibitem[{Huo et~al.(2020)Huo, Liu, Zhang, and Song}]{huo2020smoothed}
Huo, S., Liu, G., Zhang, J., and Song, C. (2020), \enquote{A smoothed finite
  element method for octree-based polyhedral meshes with large number of
  hanging nodes and irregular elements,} \textit{Computer Methods in Applied
  Mechanics and Engineering}, 359, 112646.

\bibitem[{Jamin et~al.(2015)Jamin, Alliez, Yvinec, and Boissonnat}]{Jamin:2015}
Jamin, C., Alliez, P., Yvinec, M., and Boissonnat, J.-D. (2015),
  \enquote{CGALmesh: A generic framework for Delaunay mesh generation,}
  \textit{ACM Transactions on Mathematical Software (TOMS)}, 41, 1--24.

\bibitem[{Kang(2020)}]{kang2020consistent}
Kang, Y. (2020), \enquote{Consistent blind image deblurring using
  jump-preserving extrapolation,} \textit{Journal of Computational and
  Graphical Statistics}, 29, 372--382.

\bibitem[{Kim and Wang(2021)}]{kim2020generalized}
Kim, M. and Wang, L. (2021), \enquote{Generalized spatially varying coefficient
  models,} \textit{Journal of Computational and Graphical Statistics}, 30,
  1--10.

\bibitem[{Lai(1989)}]{Lai:89}
Lai, M.~J. (1989), \enquote{On construction of bivariate and trivariate vertex
  splines on arbitrary mixed grid partitions,} Phd dissertation, Texas A\&M
  University.

\bibitem[{Lai and Schumaker(2007)}]{Lai:Schumaker:07}
Lai, M.~J. and Schumaker, L.~L. (2007), \textit{Spline functions on
  triangulations}, Cambridge University Press.

\bibitem[{Lai and Wang(2013)}]{Lai:Wang:13}
Lai, M.~J. and Wang, L. (2013), \enquote{Bivariate penalized splines for
  regression,} \textit{Statistica Sinica}, 23, 1399--1417.

\bibitem[{Lee et~al.(2016)Lee, Song, and Jo}]{Lee:Song:Jo:16}
Lee, H., Song, S., and Jo, S. (2016), \enquote{3D reconstruction using a sparse
  laser scanner and a single camera for outdoor autonomous vehicle,} in
  \textit{2016 IEEE 19th International Conference on Intelligent Transportation
  Systems (ITSC)}, pp. 629--634.

\bibitem[{Li and Ruppert(2008)}]{Li:Ruppert:08}
Li, Y. and Ruppert, D. (2008), \enquote{On the asymptotics of penalized
  splines,} \textit{Biometrika}, 95, 415--436.

\bibitem[{Lindgren et~al.(2011)Lindgren, Rue, and
  Lindstr{\"o}m}]{lindgren2011explicit}
Lindgren, F., Rue, H., and Lindstr{\"o}m, J. (2011), \enquote{An explicit link
  between Gaussian fields and Gaussian Markov random fields: the stochastic
  partial differential equation approach,} \textit{Journal of the Royal
  Statistical Society Series B (Statistical Methodology)}, 73, 423--498.

\bibitem[{Liu and Zhao(2021)}]{Liu:Zhao:21}
Liu, R. and Zhao, Y. (2021), \enquote{Empirical likelihood inference for
  generalized additive partially linear models,} \textit{TEST}, 30, 569--585.

\bibitem[{Ma(2012)}]{Ma:12}
Ma, S. (2012), \enquote{Two-step spline estimating equations for generalized
  additive partially linear models with large cluster sizes,} \textit{Annals of
  Statistics}, 40, 2943--2972.

\bibitem[{Mammen(1993)}]{Mammen:93}
Mammen, E. (1993), \enquote{Bootstrap and wild bootstrap for high dimensional
  linear models,} \textit{The Annals of Statistics}, 21, 255--285.

\bibitem[{Morris and Carroll(2006)}]{Morris:Carroll:06}
Morris, J.~S. and Carroll, R.~J. (2006), \enquote{Wavelet-based functional
  mixed models,} \textit{Journal of the Royal Statistical Society Series B
  (Statistical Methodology)}, 68, 179--199.

\bibitem[{Persson and Strang(2004)}]{Persson:Strang:04}
Persson, P.-O. and Strang, G. (2004), \enquote{A simple mesh generator in
  MATLAB,} \textit{SIAM Review}, 46, 329--345.

\bibitem[{Project(2020)}]{cgal:eb-20b}
Project, T.~C. (2020), \textit{CGAL user and reference manual}, CGAL Editorial
  Board, 5th ed.

\bibitem[{Ramsay(2002)}]{Ramsay:02}
Ramsay, T. (2002), \enquote{Spline smoothing over difficult regions,}
  \textit{Journal of the Royal Statistical Society Series B (Statistical
  Methodology)}, 64, 307--319.

\bibitem[{Reiss and Ogden(2010)}]{Reiss:Ogden:10}
Reiss, P.~T. and Ogden, R.~T. (2010), \enquote{Functional generalized linear
  models with images as predictors,} \textit{Biometrics}, 66, 61--69.

\bibitem[{Roberts et~al.(2017)Roberts, Bahn, Ciuti, Boyce, Elith,
  Guillera-Arroita, Hauenstein, Lahoz-Monfort, Schr{\"o}der, Thuiller,
  et~al.}]{Roberts:etal:17}
Roberts, D.~R., Bahn, V., Ciuti, S., Boyce, M.~S., Elith, J., Guillera-Arroita,
  G., Hauenstein, S., Lahoz-Monfort, J.~J., Schr{\"o}der, B., Thuiller, W.,
  et~al. (2017), \enquote{Cross-validation strategies for data with temporal,
  spatial, hierarchical, or phylogenetic structure,} \textit{Ecography}, 40,
  913--929.

\bibitem[{Rusinkiewicz and Levoy(2000)}]{Rusinkiewicz:Levoy:00}
Rusinkiewicz, S. and Levoy, M. (2000), \enquote{QSplat: A multiresolution point
  rendering system for large meshes,} in \textit{Proceedings of the 27th Annual
  Conference on Computer Graphics and Interactive Techniques}, USA: ACM
  Press/Addison-Wesley Publishing Co., SIGGRAPH '00, pp. 343--352.

\bibitem[{Shao(2021)}]{Shao:21}
Shao, J. (2021), \enquote{Testing object detection for autonomous driving
  systems via 3D reconstruction,} in \textit{2021 IEEE/ACM 43rd International
  Conference on Software Engineering: Companion Proceedings (ICSE-Companion)},
  pp. 117--119.

\bibitem[{Shewchuk(1998)}]{Shewchuk:98}
Shewchuk, J.~R. (1998), \enquote{A condition guaranteeing the existence of
  higher-dimensional constrained Delaunay triangulations,} in
  \textit{Proceedings of the Fourteenth Annual Symposium on Computational
  Geometry}, New York, NY, USA: Association for Computing Machinery, SCG '98,
  pp. 76--85.

\bibitem[{Si(2015)}]{Si:2015}
Si, H. (2015), \enquote{TetGen, a Delaunay-based quality tetrahedral mesh
  generator,} \textit{ACM Transactions on Mathematical Software (TOMS)}, 41,
  1--36.

\bibitem[{Stone(1994)}]{Stone:94}
Stone, C.~J. (1994), \enquote{The use of polynomial splines and their tensor
  products in multivariate function estimation,} \textit{The Annals of
  Statistics}, 22, 118--171.

\bibitem[{Stone(1982)}]{Stone:82}
Stone, J.~C. (1982), \enquote{Optimal global rates of convergence for
  nonparametric regression,} \textit{The Annals of Statistics}, 10, 1040--1053.

\bibitem[{Valavi et~al.(2019)Valavi, Elith, Lahoz-Monfort, and
  Guillera-Arroita}]{Valavi:18}
Valavi, R., Elith, J., Lahoz-Monfort, J.~J., and Guillera-Arroita, G. (2019),
  \enquote{blockCV: An R package for generating spatially or environmentally
  separated folds for k-fold cross-validation of species distribution models,}
  \textit{Methods in Ecology and Evolution}, 10, 225--232.

\bibitem[{Wang and Ranalli(2007)}]{Wang:Ranalli:07}
Wang, H. and Ranalli, M.~G. (2007), \enquote{Low-rank smoothing splines on
  complicated domains,} \textit{Biometrics}, 63, 209--217.

\bibitem[{Wang and Yang(2009)}]{Wang:Yang:09}
Wang, J. and Yang, L. (2009), \enquote{Polynomial spline confidence bands for
  regression curves,} \textit{Statistica Sinica}, 19, 325--342.

\bibitem[{Wang et~al.(2020{\natexlab{a}})Wang, Wang, Lai, and
  Gao}]{Wang:Wang:Lai:Gao:20}
Wang, L., Wang, G., Lai, M.-J., and Gao, L. (2020{\natexlab{a}}),
  \enquote{Efficient estimation of partially linear models for data on
  complicated domains via bivariate penalized splines over triangulations,}
  \textit{Statistica Sinica}, 30, 347--369.

\bibitem[{Wang et~al.(2020{\natexlab{b}})Wang, Xue, and
  Yang}]{Wang:Xue:Yang:20}
Wang, L., Xue, L., and Yang, L. (2020{\natexlab{b}}), \enquote{Estimation of
  additive frontier functions with shape constraints,} \textit{Journal of
  Nonparametric Statistics}, 32, 262--293.

\bibitem[{Wilhelm et~al.(2016)Wilhelm, Ded{\`e}, Sangalli, and
  Wilhelm}]{wilhelm2016igs}
Wilhelm, M., Ded{\`e}, L., Sangalli, L.~M., and Wilhelm, P. (2016),
  \enquote{IGS: an IsoGeometric approach for Smoothing on surfaces,}
  \textit{Computer Methods in Applied Mechanics and Engineering}, 302, 70--89.

\bibitem[{Wong et~al.(2016)Wong, Lee, Paul, and Peng}]{wong2016fiber}
Wong, R.~K., Lee, T.~C., Paul, D., and Peng, J. (2016), \enquote{Fiber
  direction estimation, smoothing and tracking in diffusion MRI,} \textit{The
  Annals of Applied Statistics}, 10, 1137--1156.

\bibitem[{Wood(2003)}]{Wood:03}
Wood, S.~N. (2003), \enquote{Thin plate regression splines,} \textit{Journal of
  the Royal Statistical Society Series B (Statistical Methodology)}, 65,
  95--114.

\bibitem[{Xu(2019)}]{Xu:19}
Xu, Y. (2019), \enquote{Multivariate spline method for scattered data fitting,
  curve and surface reconstruction, and numerical solution to Poisson equations
  via domain decomposition method,} Ph.D. thesis, University of Georgia.

\bibitem[{Yu et~al.(2020)Yu, Wang, Wang, Liu, and Yang}]{Yu:etal:20}
Yu, S., Wang, G., Wang, L., Liu, C., and Yang, L. (2020), \enquote{Estimation
  and inference for generalized geoadditive models,} \textit{Journal of the
  American Statistical Association}, 115, 761--774.

\bibitem[{Yue and Speckman(2010)}]{yue2010nonstationary}
Yue, Y. and Speckman, P.~L. (2010), \enquote{Nonstationary spatial Gaussian
  Markov random fields,} \textit{Journal of Computational and Graphical
  Statistics}, 19, 96--116.

\bibitem[{Zhu et~al.(2014)Zhu, Fan, and Kong}]{Zhu:Fan:Kong:14}
Zhu, H., Fan, J., and Kong, L. (2014), \enquote{Spatially varying coefficient
  model for neuroimaging data with jump discontinuities,} \textit{Journal of
  the American Statistical Association}, 109, 1084--1098.

\end{thebibliography}

\end{document}